\documentclass[twocolumn]{autart}

\usepackage{cite}
\usepackage{amsmath,amssymb,amsfonts}
\usepackage{graphicx}
\usepackage{mathrsfs}
\usepackage[dvipsnames]{xcolor}
\usepackage[colorlinks=true,allcolors=Violet]{hyperref}
\usepackage{float}
\usepackage{enumitem}
\usepackage[bb=boondox]{mathalfa}

\usepackage{flushend}

\newcommand{\bmat}[1]{\begin{bmatrix}#1\end{bmatrix}}
\newcommand{\smat}[1]{\left[\begin{smallmatrix}#1\end{smallmatrix}\right]}

\newcommand{\smatNoB}[1]{\begin{smallmatrix}#1\end{smallmatrix}}
\newcommand{\As}{{A^\star}}
\newcommand{\Bs}{{B^\star}}
\newcommand{\Cs}{{C^\star}}
\newcommand{\ells}{{\ell^\star}}
\newcommand{\tu}[1]{\textup{#1}}
\newcommand{\real}{\mathbb{R}}
\newcommand{\dotsS}{\rotatebox{0}{\tiny\dots}}
\newcommand{\vdotsS}{\rotatebox{90}{\tiny\dots}}
\newcommand{\ddotsS}{\rotatebox{135}{\tiny\dots}}

\renewcommand{\bf}[1]{\mathbf{#1}}

\newcommand{\aug}{\rm{aug}}

\newenvironment{smallarray}[1]
 {\null\,\vcenter\bgroup\scriptsize
  \arraycolsep=.13885em
  \hbox\bgroup$\array{@{}#1@{}}}
 {\endarray$\egroup\egroup\,\null}

\DeclareMathOperator{\rank}{rank}
\DeclareMathOperator{\im}{im}
\DeclareMathOperator{\eig}{eig}

\newtheorem{assumption}{Assumption}
\newtheorem{problem}{Problem}
\newtheorem{theorem}{Theorem}
\newtheorem{proposition}{Proposition}
\newtheorem{lemma}{Lemma}

\newenvironment{proof}{{\itshape Proof.}}{\hspace*{0pt}\hfill$\qed$}
\newcounter{remarkcounter}
\newenvironment{remark}{\refstepcounter{remarkcounter}{\bfseries Remark~\arabic{remarkcounter}.}}{\hspace*{0pt}\hfill$\blacksquare$}

\allowdisplaybreaks
\graphicspath{{./images/}}

\edef\endfrontmatter{%
  \unexpanded\expandafter{\endfrontmatter}% the current code
  \noexpand\endNoHyper % add \endNoHyper at the end to match \NoHyper
}
\begin{document}

\begin{frontmatter}

\title{Controller Synthesis from Noisy-Input Noisy-Output Data \thanksref{footnoteinfo}}

\thanks[footnoteinfo]{Corresponding author Nima Monshizadeh. The work of Lidong Li was supported by the Chinese Scholarship Council.}

\author[Gro]{Lidong Li}\ead{l.li@rug.nl},
\author[Mil]{Andrea Bisoffi}\ead{andrea.bisoffi@polimi.it},
\author[Gro]{Claudio De Persis}\ead{c.de.persis@rug.nl}, 
\author[Gro]{Nima Monshizadeh}\ead{n.monshizadeh@rug.nl}

\address[Gro]{Engineering and Technology Institute, University of Groningen, 9747AG, The Netherlands}
\address[Mil]{Department of Electronics, Information, and Bioengineering, Politecnico di Milano, 20133, Italy}

\begin{keyword}                           
linear time-invariant system, data-based control, analysis of systems with uncertainties, robust control, output feedback control, dynamic controller
\end{keyword}
\begin{abstract}
We consider the problem of synthesizing a dynamic output-feedback controller for a linear system, using solely input-output data corrupted by measurement noise. 
To handle input-output data, an auxiliary representation of the original system is introduced. 
By exploiting the structure of the auxiliary system, we design a controller that robustly stabilizes all possible systems consistent with data. 
Notably, we also provide a novel solution to extend the results to generic multi-input multi-output systems.
The findings are illustrated by numerical examples.
\end{abstract}
\end{frontmatter}

%%%%%%%%%%%%%%%%%%%%%%%%%%%%%%%%%%%%%%%%%%%%%%%%%%%%%%%%%%%%%%%%%%%%%%%%%%%%%%%%%%%
%%%%%%%%%%%%%%%%%%%%%%%%%%%%%%%%%%%%%%%%%%%%%%%%%%%%%%%%%%%%%%%%%%%%%%%%%%%%%%%%%%%
%%%%%%%%%%%%%%%%%%%%%%%%%%%%%%%%%%%%%%%%%%%%%%%%%%%%%%%%%%%%%%%%%%%%%%%%%%%%%%%%%%%
%%%%%%%%%%%%%%%%%%%%%%%%%%%%%%%%%%%%%%%%%%%%%%%%%%%%%%%%%%%%%%%%%%%%%%%%%%%%%%%%%%%
\section{Introduction}

Facilitated by recent technological advancements in sensor and measurement devices, along with the widespread availability of data, data-driven control has become an increasingly popular research avenue in recent years. 
Typically, data has been used to identify a mathematical model describing the dynamics of the system. 
Such a model can, in turn, be used for control and monitoring purposes. 
An alternative approach, often referred to as direct data-driven control, is to bypass the model and directly use the data to synthesize the control algorithm \cite{CAMPI2002Virtual, TANASKOVIC2017, formulas2020, Florian2019DeePC}.

One of the most classical problems in control theory is stabilizing a linear time-invariant system. In the realm of data-driven control, this raises the question of how to design stabilizing controllers directly from the data obtained from the system.
In practice, such data come from multiple actuators and sensors and are not completely accurate. This leads to the problem of stabilizing a {multi-input multi-output} (MIMO) linear system using \textit{noisy-input noisy-output} data. As stated below, there have been numerous valuable efforts to tackle this challenging problem by approaching it from different angles. Nevertheless, to the best of our knowledge, a complete answer to this problem is still absent from the literature. The main motivation of this work is to fill this gap.

\noindent\textit{Related works:}\\
The central problem posed above has two key challenges: (i) the use of \textit{input-output} data rather than state measurements and (ii) the presence of noise on both input and output channels. To cope with the first challenge, \cite{formulas2020} proposes to work with an \emph{auxiliary representation} of the system which is built using
{time-shifted} inputs and outputs of the system \cite{GrahamAdaptive2009}.  
This auxiliary representation enables the designer to take advantage of data-driven control techniques that are applicable to input-state data; see e.g., \cite{berberich2021combining, Jong2023output, makdah2023output, Dai2023IO, William2023output, Guanru2023output, jo2022IOpredictive}.
However, most of the aforementioned works efficiently address only single-input single-output or multiple-input single-output cases. 
This limitation primarily arises from the fact that the auxiliary representation of a {minimal} {MIMO} linear system is not necessarily reachable \cite[p. 566]{LinearSystems2005}.
To tackle the challenge posed by the absence of reachability, \cite{mohammad2023MIMO} and \cite{Tomonori2022} aim at extracting the reachable mode from the auxiliary representation. 
Nevertheless, the aforementioned developments are not readily applicable to the practically notable scenarios where noise is present on \textit{both} input and output channels, bringing us to the second challenge of the problem.

Besides the concern of unreachability of the auxiliary representation, the presence of noise in input-output data further complicates the task of control design. 
For the case of process noise, which acts as an additive disturbance to the system equations, we refer to \cite{Henk2023IO} and \cite{FrankDissipativity2022}. 
Measurement noise substantially increases the complexity since the system dynamics further intertwine the useful signal and detrimental noise in the collected data.
The case of noisy state measurements is treated in \cite{formulas2020, Miller2023ErrorVariable}. In fact, by deriving a matrix elimination result and suitably leveraging the Petersen's lemma \cite{petersen1987stabilization}, one can establish a necessary and sufficient condition for robust stabilization of systems consistent with noisy-input noisy-\emph{state} data \cite{stateinputerrors}. 
A treatment of the noisy-input noisy-\emph{output} case has only been recently reported in \cite{miller2023superstabilizing} for single-input single-output systems, using superstability, which is a stronger notion than asymptotic stability.

\textit{Contribution and proposed approach:} \\
As reviewed above, the main contribution of this work is synthesizing a stabilizing controller using noisy-input noisy-output data acquired from a generic MIMO system.

Intuitively, when the magnitude of the noise is large, the valuable information is sunk within the data, making it practically impossible to design a controller using such data. 
Motivated by this, we provide conditions under which the noisy input-output data remain useful for control design.

Given a linear system of minimal order $n$, we distinguish between two cases. 
The first case is $n=p\ell$ where $p$ is the number of outputs and $\ell$ is the observability index of the system. In this case, the auxiliary system can be reachable and, after separating the unknown blocks of the auxiliary representation from the known ones, we can make use of Petersen's lemma in the context of data-driven control \cite{AndreaPetersen2022} to design a data-driven controller through a convenient linear matrix inequality. 
This controller can stabilize the auxiliary representation and, by a painstaking investigation of the relationship between the input-output evolution of the original and the auxiliary system, we can conclude that the controller works for the original system as well.

The second case is $n \ne p\ell$, where the auxiliary system becomes unreachable and it turns out that the data-based conditions that are exploited in the first case do not hold. 
To address this challenge, we introduce a novel technique where we form an augmented system by a parallel interconnection of the original system with another linear system. By imposing suitable conditions on the added dynamics, we bring the setup back to the first case, thereby extending the results to general multi-input multi-output linear systems. 
In particular, we show that the data-driven controller designed for the augmented system together with the artificially added linear system serve as a stabilizing dynamic output-feedback controller for the original system.

\noindent \textit{Outline:} \\
Section~\ref{sec:Preliminaries} includes the required notions and definitions. The relationship between the input-output evolution of the original system and its auxiliary representation is also discussed in this section as preliminaries. Then we formulate the problem of interest in Section~\ref{sec:problem}. Section~\ref{sec:results_pells=n} presents our main results {for both cases of $n = p\ell$ and $n \neq p \ell$}. The results are verified by numerical examples\footnote{The code of the numerical examples is available at https://github.com/BG5CPU/DataDrivenNoisyIO} in Section~\ref{sec:NumericalExamples}.

%%%%%%%%%%%%%%%%%%%%%%%%%%%%%%%%%%%%%%%%%%%%%%%%%%%%%%%%%%%%%%%%%%%%%%%%%%%%%%%%%%%
%%%%%%%%%%%%%%%%%%%%%%%%%%%%%%%%%%%%%%%%%%%%%%%%%%%%%%%%%%%%%%%%%%%%%%%%%%%%%%%%%%%
%%%%%%%%%%%%%%%%%%%%%%%%%%%%%%%%%%%%%%%%%%%%%%%%%%%%%%%%%%%%%%%%%%%%%%%%%%%%%%%%%%%
%%%%%%%%%%%%%%%%%%%%%%%%%%%%%%%%%%%%%%%%%%%%%%%%%%%%%%%%%%%%%%%%%%%%%%%%%%%%%%%%%%%
\section{Preliminaries}\label{sec:Preliminaries}

%%%%%%%%%%%%%%%%%%%%%%%%%%%%%%%%%%%%%%%%%%%%%%%%%%%%%%%%%%%%%%%%%%%%%%%%%%%%%%%%%%%
%%%%%%%%%%%%%%%%%%%%%%%%%%%%%%%%%%%%%%%%%%%%%%%%%%%%%%%%%%%%%%%%%%%%%%%%%%%%%%%%%%%
\subsection{Notation}\label{sec:notation}

For column vectors $v_1$, \dots, $v_n$, $(v_1,\dots, v_n)$ denotes the stacked vector $[v_1^\top \ \dots \ v_n^\top]^\top$.
The identity matrix of size $n$ and the zero matrix of size $m \times n$ are $I_n$ and $0_{m \times n}$: the indices are dropped when no confusion arises.
The largest and smallest eigenvalues of a symmetric matrix $M$ are $\lambda_{\max}(M)$ and $\lambda_{\min}(M)$.
The largest and smallest singular values of a matrix $M$ are $\sigma_{\max}(M)$ and $\sigma_{\min}(M)$.
The set of eigenvalues of a matrix $M$ is $\eig(M)$.
The 2-norm of a vector $v$ is $| v |$.
The induced 2-norm of a matrix $M$ is $\|M\|$ and is given by $\sigma_{\max}(M)$. 
For a symmetric matrix $\left[\begin{smallmatrix} M & N \\ N^{\top} & O \end{smallmatrix}\right]$, we may use the shorthand writing $\left[\begin{smallmatrix} M & N \\ \star & O \end{smallmatrix}\right]$.
For a positive semidefinite matrix $M$, $M^{1/2}$ is the unique positive semidefinite square root of $M$.
For a discrete-time signal $h \colon \mathbb{Z} \to \real^n$, $\{ h(k) \}_{k=k_0}^{k_1}$ is the sequence of values $h(k_0)$, $h(k_0+1)$, \dots, $h(k_1)$, where $k_0 \le k_1$ and $k_1$ possibly $\infty$.
For the signal $\{ h(k) \}_{k=0}^{\infty}$, its $\mathcal{L}_{\infty}$ norm is $\| h \|_{\mathcal{L}_{\infty}} := \sup_{k \geq 0} | h(k) | $.
Given an $n$-dimensional state-space model $x^{+} = Ax + Bu$, its reachability subspace is as $R(A,B):=\im [ A^{n-1}B \ \cdots \ AB \ B]$.

%%%%%%%%%%%%%%%%%%%%%%%%%%%%%%%%%%%%%%%%%%%%%%%%%%%%%%%%%%%%%%%%%%%%%%%%%%%%%%%%%%%
%%%%%%%%%%%%%%%%%%%%%%%%%%%%%%%%%%%%%%%%%%%%%%%%%%%%%%%%%%%%%%%%%%%%%%%%%%%%%%%%%%%
\subsection{The auxiliary system and its properties}

Consider a generic discrete-time LTI system
\begin{subequations}
\label{sys_x}
\begin{align}
x^+ & = A x + B u\\
y & = C x
\end{align}
\end{subequations}
where $x \in \real^{n}$, $u \in \real^{m}$ and $y \in \real^{p}$.
If the pair $(A, C)$ is observable, there exist $l$ such that
\begin{align*}
\rank \smat{C \\ C A  \\ \vdotsS \\ C A^{l-1}} = n;
\end{align*}
the minimum $l$ for which this rank condition holds is denoted by $\ell$, which is called the observability index (of pair $(A,C)$) \cite[p.~356-357]{kailath1980linear}.
Then, if the pair $(A, C)$ is observable, we can define the matrices
\begin{subequations}
\label{O_ell_T_ell_R_ell}
\begin{align}
\mathcal{O}_\ell & := 
\bmat{
C\\
CA\\
\vdots\\
CA^{\ell-1}
} \in \mathbb{R}^{p\ell\times n}, \label{O_ell}\\
\mathcal{T}_\ell & :=
\bmat{
0_{p\times m} & 0 & \dots & 0 & 0\\
CB                     & 0 & \dots & 0 & 0\\
CAB                    & CB                     & \dots & 0 & 0\\
\vdots                 & \vdots                 & \ddots &  \vdots & \vdots\\
CA^{\ell -2} B         & CA^{\ell -3} B         & \dots & CB &0_{p\times m}
} \in \mathbb{R}^{p\ell\times m\ell}, \label{T_ell}\\
\mathcal{R}_\ell & := \bmat{ A^{\ell-1} B & \ldots & A B & B } \in \real^{n \times m \ell}. \label{R_ell}
\end{align}
\end{subequations}
Since $\rank \mathcal{O}_\ell =n$ by construction, $\mathcal{O}_\ell$ possesses a left inverse $\mathcal{O}_\ell^{\tu{L}}$ that satisfies $\mathcal{O}_\ell^{\tu{L}} \mathcal{O}_\ell = I_n$.
Throughout the paper, we are also interested in the \emph{auxiliary system}
\begin{equation}
\label{sys_xi}
\xi^+ = \mathbf{A}_\ell \xi + \mathbf{B}_\ell v := (\mathbf{F}_\ell + \mathbf{L}_\ell Z_\ell) \xi + \mathbf{B}_\ell v
\end{equation}
where we define, for $\ell$ as above,
\begin{subequations}
\label{sys_xi:A0_L_B_Z}
\begin{align}
& \bmat{
\mathbf{F}_\ell & \mathbf{L}_\ell & \mathbf{B}_\ell
} := \label{sys_xi:A0_L_B} \\
& \bmat{
\left[ \begin{array}{c|c}
\begin{smallmatrix} 
0 & I_p & 0 & \cdots & 0 \\ 
0 & 0 & I_p & \cdots & 0 \\ 
\vdotsS & \vdotsS & \vdotsS & \ddotsS & \vdotsS \\ 
0 & 0 & 0 & \cdots & I_p \\ 
0 & 0 & 0 & \cdots & 0
\end{smallmatrix} 
& 0_{p \ell \times m \ell} \\ 
\hline 
0_{m \ell \times p \ell} & 
\begin{smallmatrix} 
0 & I_m & 0 & \cdots & 0 \\ 
0 & 0 & I_m & \cdots & 0 \\ 
\vdotsS & \vdotsS & \vdotsS & \ddotsS & \vdotsS \\ 
0 & 0 & 0 & \dotsS & I_m \\ 
0 & 0 & 0 & \dotsS & 0
\end{smallmatrix} 
\end{array} \right]
&
\left[%
\begin{array}{c}
\begin{smallmatrix} 
0 \\[2pt] 
0 \\[1pt] 
\vdotsS \\[1pt] 
0 \\[1pt]
I_p 
\end{smallmatrix} \\[1pt] 
\hline  
\begin{smallmatrix} 
0 \\[1pt]  
0 \\[1pt] 
\vdotsS \\[1pt] 
0 \\[1pt] 
0
\end{smallmatrix} 
\end{array} \right]
&
\left[ \begin{array}{c}
\begin{smallmatrix} 
0 \\[2pt] 
0 \\[1pt] 
\vdotsS \\[1pt] 
0 \\[1pt]
0\rotatebox{180}{\rule{0pt}{2.5pt}}
\end{smallmatrix} \\ 
\hline  
\begin{smallmatrix} 
0 \\[1pt] 
0 \\[1pt] 
\vdotsS \\[1pt] 
0 \\[1pt] 
I_m
\end{smallmatrix} 
\end{array} \right]
} \notag 
\end{align}
and, for $(A,B,C)$ in~\eqref{sys_x} and $\mathcal{O}_\ell$, $\mathcal{T}_\ell$, $\mathcal{R}_\ell$ in~\eqref{O_ell_T_ell_R_ell},%
\begin{align}
\label{sys_xi:Z}
& 
Z_\ell \! := \!
\bmat{Z_{1\ell} & Z_{2\ell}} \! := \!
\bmat{C A^\ell \mathcal{O}_\ell^{\tu{L}} & & C \mathcal{R}_\ell- C A^\ell \mathcal{O}_\ell^{\tu{L}} \mathcal{T}_\ell}.
\end{align}
\end{subequations}
Note that from~\eqref{sys_xi} and \eqref{sys_xi:A0_L_B_Z},
\begin{align}
\label{sys_xi:Aell}
\mathbf{A}_\ell  & = 
\left[
\begin{array}{c|c}
\smatNoB{
0 & I_p & 0 & \dots & 0\\
0 & 0  & I_p & \dots & 0\\
\vdotsS & \vdotsS & \vdotsS & \ddotsS & \vdotsS \\
0 & 0  & 0 & \dots & I_p
}
& 
\smatNoB{
0 & 0& 0& \ldots & 0\\
0 & 0& 0& \ldots & 0\\
\vdotsS & \vdotsS & \vdotsS & \ddotsS & \vdotsS \\
0 & 0& 0& \ldots & 0
}\\
\smatNoB{Z_{1\ell}}
&
\smatNoB{Z_{2\ell}} \\
\hline
0
& 
\smatNoB{
0 & I_m & 0& \ldots & 0\\
0 & 0 & I_m& \ldots & 0\\
\vdotsS & \vdotsS & \vdotsS & \ddotsS & \vdotsS \\
0 & 0 & 0& \ldots & I_m\\
0 & 0 & 0& \ldots & 0
}
\end{array}
\right], \notag \\
& = 
\left[
\begin{array}{c|c}
\smatNoB{
0 & I_p & 0 & \dots & 0\\
0 & 0  & I_p & \dots & 0\\
\vdotsS & \vdotsS & \vdotsS & \ddotsS & \vdotsS \\
0 & 0  & 0 & \dots & I_p
}
& 
\smatNoB{
0 & 0& 0& \ldots & 0\\
0 & 0& 0& \ldots & 0\\
\vdotsS & \vdotsS & \vdotsS & \ddotsS & \vdotsS \\
0 & 0& 0& \ldots & 0
}\\
\smatNoB{
C A^\ell \mathcal{O}_\ell^{\tu{L}}
}
&
\smatNoB{
C \mathcal{R}_\ell - C A^\ell \mathcal{O}_\ell^{\tu{L}} \mathcal{T}_\ell 
} \\
\hline
0
& 
\smatNoB{
0 & I_m & 0& \ldots & 0\\
0 & 0 & I_m& \ldots & 0\\
\vdotsS & \vdotsS & \vdotsS & \ddotsS & \vdotsS \\
0 & 0 & 0& \ldots & I_m\\
0 & 0 & 0& \ldots & 0
}
\end{array}
\right].
\end{align}
In Appendix~\ref{app:results_lin_sys}, we report some straightforward results for the linear system in~\eqref{sys_x} and the auxiliary linear system in~\eqref{sys_xi}, which are invoked in the proofs of the subsequent Lemmas~\ref{lemma:imH=R(A,B)}, \ref{lemma:AbfBbf_reach_pl=n}, \ref{lemma:io_sys_x_is_io_sys_xi}.
Define the matrix $H_\ell$ as
\begin{align}
\label{Hell}
H_\ell:=
\left[
\begin{array}{c|c}
\mathcal{O}_\ell & \mathcal{T}_\ell\\
\hline
0 & I_{m \ell}
\end{array}
\right] \in \mathbb{R}^{(p\ell+m\ell) \times (n+m\ell)}.
\end{align} 
We have the next result, which was claimed within \cite[Proof of Lemma~4]{Tomonori2022}, but not proven therein; the reachability subspace is defined in Section~\ref{sec:notation}.

\begin{lemma}
\label{lemma:imH=R(A,B)}
For $(A,C)$ observable and $(A,B)$ reachable, let $R(\mathbf{A}_\ell, \mathbf{B}_\ell)$ be the reachability subspace of the pair  $(\mathbf{A}_\ell, \mathbf{B}_\ell)$. Then, $\im H_\ell = R(\mathbf{A}_\ell, \mathbf{B}_\ell)$.
\end{lemma}
\begin{proof}
See Appendix~\ref{app:proof_lemma_imH=R(A,B)}.
\end{proof}

A consequence of Lemma~\ref{lemma:imH=R(A,B)} is the next result, which is relevant in the sequel.

\begin{lemma}
\label{lemma:AbfBbf_reach_pl=n}
For $(A,C)$ observable and $(A,B)$ reachable, the pair $(\mathbf{A}_\ell, \mathbf{B}_\ell)$ is reachable if and only if $p\ell =n$. 
\end{lemma}
\begin{proof}
See Appendix~\ref{app:proof_lemma_AbfBbf_reach_pl=n}.
\end{proof}

Finally, the next result illustrates how the input-output evolution of solutions to~\eqref{sys_x} can be captured by solutions to~\eqref{sys_xi}.

\begin{lemma}
\label{lemma:io_sys_x_is_io_sys_xi}
Given the matrices $(A,B,C)$ of \eqref{sys_x}, let $(A,C)$ be observable and $\mathbf{A}_\ell$, $\mathbf{B}_\ell$, $Z_\ell$ be as in \eqref{sys_xi}.
For each $\hat{x}$ and sequence $\{ u(k)\}_{k = 0}^\infty$, there exists $\hat{\xi}$ such that:
\begin{itemize}[nosep,noitemsep,left=0pt]
\item the solution $x(\cdot)$ to~\eqref{sys_x} with initial condition $x(0) = \hat x$ and input $\{u(k) \}_{k = 0}^\infty$,
\item the corresponding output response $y(\cdot) = C x(\cdot)$, and
\item the solution $\xi(\cdot)$ to \eqref{sys_xi} with initial condition $\xi(\ell) = \hat \xi$ and input $\{v(k)\}_{k = \ell}^\infty = \{ u(k) \}_{k = \ell}^\infty$
\end{itemize}
satisfy
\begin{subequations}
 \begin{align}
\label{io_sys_x_is_io_sys_xi:xi_k}
 \left[
\begin{array}{c}
\begin{smallmatrix}
y(k-\ell) \\
\vdotsS \\ 
y(k-1)\\
\end{smallmatrix}\\
\hline
\begin{smallmatrix}
u(k-\ell)\\
\vdotsS \\
u(k-1)
\end{smallmatrix}
\end{array}
\right]
= \xi(k) \quad & \forall k \ge \ell\\
y(k)= Z_{\ell} \xi(k) \quad & \forall k \ge \ell.
\label{io_sys_x_is_io_sys_xi:y_k}
\end{align}   
\end{subequations}
\end{lemma}
\begin{proof}
See Appendix~\ref{app:proof_lemma_io_sys_x_is_io_sys_xi}.
\end{proof}

We note that Lemma~\ref{lemma:io_sys_x_is_io_sys_xi} compares the solution to~\eqref{sys_x} for a generic initial condition $\hat{x}$ with the solution to~\eqref{sys_xi} with $\xi(\ell) = \hat{\xi}$ for a suitable $\hat{\xi}$ from $k = \ell$ onward.

%%%%%%%%%%%%%%%%%%%%%%%%%%%%%%%%%%%%%%%%%%%%%%%%%%%%%%%%%%%%%%%%%%%%%%%%%%%%%%%%%%%
%%%%%%%%%%%%%%%%%%%%%%%%%%%%%%%%%%%%%%%%%%%%%%%%%%%%%%%%%%%%%%%%%%%%%%%%%%%%%%%%%%%
%%%%%%%%%%%%%%%%%%%%%%%%%%%%%%%%%%%%%%%%%%%%%%%%%%%%%%%%%%%%%%%%%%%%%%%%%%%%%%%%%%%
%%%%%%%%%%%%%%%%%%%%%%%%%%%%%%%%%%%%%%%%%%%%%%%%%%%%%%%%%%%%%%%%%%%%%%%%%%%%%%%%%%%
\section{Problem formulation}\label{sec:problem}

Consider the discrete-time linear-time-invariant system
\begin{subequations}
\label{sys_x_star}
\begin{align}
x^+ & = \As x + \Bs u \label{sys_x_star:x}\\
y & = \Cs x \label{sys_x_star:y}
\end{align}
\end{subequations}
with state $x \in \real^n$, input $u \in \real^m$ and output $y \in \real^p$.
The matrices $\As$, $\Bs$, $\Cs$ are unknown to us and, instead of their knowledge, we rely on collecting noisy-input noisy-output measurements to design a controller for~\eqref{sys_x_star}, as we explain below.
Our prior knowledge on~\eqref{sys_x_star} is summarized in the next assumption.
\begin{assumption}
\label{ass:observ}
The pair $(\As, \Cs)$ is observable and the observability index $\ells$ of $(\As, \Cs)$ is known.
\end{assumption}

Knowing the observability index $\ells$ allows us to construct some of the quantities used in the sequel.
Based on Assumption~\ref{ass:observ}, we can define the unknown matrices $\mathcal{O}_\ells$, $\mathcal{T}_{\ells}$, $\mathcal{R}_{\ells}$ as in~\eqref{O_ell}, \eqref{T_ell}, \eqref{R_ell}, where, since $(\As,\Cs)$ is observable, $\rank \mathcal{O}_\ells = n$ and $\mathcal{O}_\ells$ possesses a left inverse $\mathcal{O}_\ells^{\tu{L}}$ that satisfies $\mathcal{O}_\ells^{\tu{L}} \mathcal{O}_\ells = I_n$.

\begin{figure}[H]
\centerline{\includegraphics[scale=0.9]{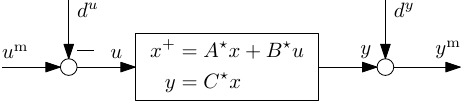}}
\caption{Scheme of data collection experiment.}
\label{fig:data_collection_exp}
\end{figure}

Input-output data are collected by performing an experiment on~\eqref{sys_x_star}.
Consider
\begin{align}
\label{u_m,y_m}
u^{\tu{m}} := u + d^u \text{ and } y^{\tu{m}} := y + d^y
\end{align}
where the measured input $u^{\tu{m}}$ differs from the actual input $u$ to~\eqref{sys_x_star} by unknown noise $d^u$, and the measured output $y^{\tu{m}}$ differs from the actual output $y$ of~\eqref{sys_x_star} by unknown noise $d^y$.
The data collection experiment is depicted in Figure~\ref{fig:data_collection_exp} and is as follows: for $k=0, \dots, T$, apply the signal $u^{\tu{m}}(k)$; along with noise $d^u(k)$, this results in (unknown) input $u(k) = u^{\tu{m}}(k) - d^{u}(k)$, (unknown) state $x(k)$ and (unknown) output $y(k)$ for some initial condition $x(0)$; measure the signal $y^{\tu{m}}(k) = y(k) + d^y(k)$.
The available data, on which our control design is based, are then $\{ u^{\tu{m}}(k), y^{\tu{m}}(k) \}_{k=0}^T$ and are gathered from
\begin{subequations}
\label{sys_x_star_ym_um}
\begin{align}
x^+ & = \As x + \Bs (u^{\tu{m}} - d^u ) \\
y^{\tu{m}} & = \Cs x + d^y.
\end{align}
\end{subequations}
For the sequel we define
\begin{subequations}
\label{Psi1_Psi0_U0_Delta10}
\begin{align}
\Psi_1 & :=
\left[
\begin{array}{c}
\begin{smallmatrix}
y^{\tu{m}}(1) & y^{\tu{m}}(2) & \dots & y^{\tu{m}}(T-\ells+1)\\
\vdotsS & \vdotsS &   & \vdotsS \\
y^{\tu{m}}(\ells) & y^{\tu{m}}(\ells+1) & \dots & y^{\tu{m}}(T) \\
\end{smallmatrix}\\[8pt]
\hline
\begin{smallmatrix}
u^{\tu{m}}(1) & u^{\tu{m}}(2) & \dots & u^{\tu{m}}(T-\ells+1) \rule{0pt}{8pt}\\
\vdotsS & \vdotsS &  & \vdotsS\\
u^{\tu{m}}(\ells) & u^{\tu{m}}(\ells+1) & \dots & u^{\tu{m}}(T)
\end{smallmatrix}
\end{array}
\right]\\
\Psi_0 & :=
\left[
\begin{array}{c}
\begin{smallmatrix}
y^{\tu{m}}(0) & y^{\tu{m}}(1) & \dots & y^{\tu{m}}(T-\ells)\\
\vdotsS & \vdotsS &   & \vdotsS \\
y^{\tu{m}}(\ells-1) & y^{\tu{m}}(\ells) & \dots & y^{\tu{m}}(T-1) \\
\end{smallmatrix}\\[8pt]
\hline
\begin{smallmatrix}
u^{\tu{m}}(0) & u^{\tu{m}}(1) & \dots & u^{\tu{m}}(T-\ells) \rule{0pt}{8pt}\\
\vdotsS & \vdotsS &  & \vdotsS\\
u^{\tu{m}}(\ells-1) & u^{\tu{m}}(\ells) & \dots & u^{\tu{m}}(T-1)
\end{smallmatrix}
\end{array}
\right]\\
U_1 & := 
\smat{u^{\tu{m}}(\ells) & u^{\tu{m}}(\ells+1) & \dots & u^{\tu{m}}(T)}
\end{align}
based on the available data $\{ u^{\tu{m}}(k), y^{\tu{m}}(k) \}_{k=0}^T$ and also the unknown
\begin{align}
\Delta_{10} & :=
\left[
\begin{array}{c}
\begin{smallmatrix}
d^y(\ells) & d^y(\ells+1) & \dots & d^y(T) \\
\end{smallmatrix}\\
\hline
\begin{smallmatrix}
d^y(0) & d^y(1) & \dots & d^y(T-\ells) \rule{0pt}{8pt}\\
\vdotsS & \vdotsS &   & \vdotsS \\
d^y(\ells-1) & d^y(\ells) & \dots & d^y(T-1) \\
\end{smallmatrix}\\[8pt]
\hline
\begin{smallmatrix}
d^u(0) & d^u(1) & \dots & d^u(T-\ells) \rule{0pt}{8pt}\\
\vdotsS & \vdotsS &  & \vdotsS\\
d^u(\ells-1) & d^u(\ells) & \dots & d^u(T-1)
\end{smallmatrix}
\end{array}
\right]. \label{Delta10}
\end{align}
\end{subequations}

For $\mathbf{F}_\ells$, $\mathbf{L}_\ells$ and $\mathbf{B}_\ells$ as in~\eqref{sys_xi:A0_L_B}, define
\begin{align}
Z_\ells & := 
\bmat{Z_{1\ells} & Z_{2\ells}} \notag\\
& :=
\bmat{C^\star {A^\star}^\ells \mathcal{O}_\ells^{\tu{L}} & &  C^\star \mathcal{R}_\ells - C^\star {A^\star}^\ells \mathcal{O}_\ells^{\tu{L}} \mathcal{T}_\ells} \label{Zstar} \\
\mathbf{A}_\ells & :=
\mathbf{F}_\ells + \mathbf{L}_\ells \bmat{Z_{1\ells} & Z_{2\ells}} \label{Astar} \\
& = 
\left[
\begin{array}{c|c}
\smatNoB{
0 & I_p & 0 & \dots & 0\\
0 & 0  & I_p & \dots & 0\\
\vdotsS & \vdotsS & \vdotsS & \ddotsS & \vdotsS \\
0 & 0  & 0 & \dots & I_p
}
& 
\smatNoB{
0 & 0& 0& \ldots & 0\\
0 & 0& 0& \ldots & 0\\
\vdotsS & \vdotsS & \vdotsS & \ddotsS & \vdotsS \\
0 & 0& 0& \ldots & 0
}\\
Z_{1\ells} & Z_{2\ells} \\
\hline
0
&
\smatNoB{
0 & I_m & 0& \ldots & 0\\
0 & 0 & I_m& \ldots & 0\\
\vdotsS & \vdotsS & \vdotsS & \ddotsS & \vdotsS \\
0 & 0 & 0& \ldots & I_m\\
0 & 0 & 0& \ldots & 0
}
\end{array}
\right] \notag \\
\mathbf{B}^{\tu{d}}_\ells & := \mathbf{L}_\ells \bmat{I_p & -Z_{1\ells} & -Z_{2\ells}} \label{Bstar_d} \\
& =
\smat{
0 & 0& 0\\
\vdotsS & \vdotsS & \vdotsS\\
0 & 0& 0\\
I_p & -Z_{1\ells} & -Z_{2\ells}\\
\hline
0 & 0 & 0\\
\vdotsS & \vdotsS & \vdotsS\\
0 & 0 & 0
}
\notag 
\end{align}
where $Z_\ells$ is as in~\eqref{sys_xi:Z} and $\mathbf{A}_\ells$ is as in~\eqref{sys_xi:Aell}.
With these definitions, we summarize the relations satisfied by data in the next result.
\begin{lemma}
\label{lemma:data_relation}
Data satisfy for $k=\ells, \dots, T$,
\begin{align}\label{exp_data_ymk}
& y^{\tu{m}}(k)   \! = \!  Z_{1\ells} \!
\smat{
y^{\tu{m}}(k - \ells)\\
\vdotsS \\
y^{\tu{m}}(k - 1)
}
\! + \!  Z_{2\ells} \!
\smat{
u^{\tu{m}}(k - \ells)\\
\vdotsS \\
u^{\tu{m}}(k - 1)
}  \\   
& \hspace{26pt}- \! Z_{1\ells} \!
\smat{
 d^y(k - \ells)\\
\vdotsS \\
d^y (k - 1)
}
\! - \! Z_{2\ells} \!
\smat{
d^u(k - \ells)\\
\vdotsS \\
d^u(k - 1) 
} \! + \! d^y(k), \notag \\
& \!\!\smat{y^{\tu{m}}(k-\ells+1)\\
\vdotsS\\
y^{\tu{m}}(k) \\
\hline
u^{\tu{m}}(k-\ells+1)\\
\vdotsS\\
u^{\tu{m}}(k)
} = (\mathbf{F}_\ells + \mathbf{L}_\ells [Z_{1\ells} \ Z_{2\ells}]) \smat{y^{\tu{m}}(k-\ells)\\
\vdotsS\\
y^{\tu{m}}(k-1)\\
\hline
u^{\tu{m}}(k-\ells)\\
\vdotsS\\
u^{\tu{m}}(k-1)
} \notag \\
& +  \mathbf{B}_\ells u^{\tu{m}}(k) 
+ \mathbf{L}_\ells  [I_p \ - \! Z_{1\ells} \ - \! Z_{2\ells}] \!
\smat{
d^y(k) \\
\hline d^y(k-\ells) \\
\vdotsS \\
d^y(k-1) \\
\hline d^u(k-\ells)\\
\vdotsS\\
d^u(k-1)
} \label{exp_data_veck}
\end{align}
and
\begin{align}
\Psi_1 & = (\mathbf{F}_\ells + \mathbf{L}_\ells Z_\ells ) \Psi_0 + \mathbf{B}_\ells U_1 + \mathbf{L}_\ells [I_p \ -Z_\ells] \Delta_{10}\notag \\
&  = \mathbf{A}_\ells \Psi_0 + \mathbf{B}_\ells U_1 + \mathbf{B}^{\tu{d}}_\ells \Delta_{10}. \label{exp_data_mat}
\end{align}
\end{lemma}

\begin{proof}
See Appendix~\ref{app:proof_lemma_data_relation}.
\end{proof}

In~\eqref{exp_data_mat}, $\mathbf{B}_\ells$ as in~\eqref{sys_xi:A0_L_B} is known whereas $\mathbf{A}_\ells$ in~\eqref{Astar} and $\mathbf{B}^{\tu{d}}_\ells$ in~\eqref{Bstar_d} contain the unknown $Z_\ells$.
Without any prior knowledge on $\Delta_{10}$, \eqref{exp_data_mat} would not provide sufficient information to design a stabilizing controller for~\eqref{sys_x_star}. Thus, we introduce  prior knowledge in the form of an energy bound on $\Delta_{10}$.
For some $\Theta = \Theta^\top \succeq 0$, define
\begin{equation}
\label{setD}
\mathcal{D} :=  \{ \Delta  \in  \real^{(p + p \ells +  m \ells) \times (T - \ells + 1)} : \Delta \Delta^\top \preceq \Theta \}
\end{equation}
and, for the sequel, partition $\Theta$ as
\begin{equation}
\label{Theta_partitioned}
\Theta =: \bmat{\Theta_{11} & \Theta_{12}\\ \Theta_{12}^\top & \Theta_{22}}
\end{equation}
with 
\[ \Theta_{11} = \Theta_{11}^\top \in \real^{p \times p} \text{, } \Theta_{22} = \Theta_{22}^\top \in \real^{ (p \ells + m \ells) \times  (p \ells + m \ells) }. \]
We assume then that the noise sequence acting during data collection belongs to $\mathcal{D}$, i.e., $\Delta_{10}\in\mathcal{D}$.
The interpretation of~\eqref{setD} as an energy bound is justified in the next remark.

\begin{remark}
Consider the permutation matrix
\begin{align*}
    \Pi := \smat{0 & I_{\ells p} & 0\\ I_p & 0 & 0\\ 0 & 0 & I_{\ells m}},
\end{align*}
which ensures that
\begin{align*}
\Pi \Delta_{10} = 
\left[
\begin{array}{cccc}
\begin{smallmatrix}
d^y(0) & d^y(1) & \dots & d^y(T-\ells)\\
\vdotsS & \vdotsS &   & \vdotsS \\
d^y(\ells) & d^y(\ells+1) & \dots & d^y(T)
\end{smallmatrix}
\\
\hline
\begin{smallmatrix}
d^u(0) & d^u(1) & \dots & d^u(T-\ells)\\
\vdotsS & \vdotsS &  & \vdotsS\\
d^u(\ells-1) & d^u(\ells) & \dots & d^u(T-1)
\end{smallmatrix}
\end{array}
\right].
\end{align*}
Then, $\Delta_{10} \in \mathcal{D}$ if and only if $\Pi \Delta_{10} \Delta_{10}^\top \Pi^\top \preceq \Pi \Theta \Pi^\top$, i.e.,
\begin{align*}
\sum_{j=0}^{T-\ells}
\left[
\begin{array}{cccc}
\begin{smallmatrix}
d^y(j)\\
\vdotsS \\
d^y(j+\ells)
\end{smallmatrix}
\\
\hline
\begin{smallmatrix}
d^u(j)\\
\vdotsS\\
d^u(j+\ells-1)
\end{smallmatrix}
\end{array}
\right]
\left[
\begin{array}{cccc}
\begin{smallmatrix}
d^y(j)\\
\vdotsS \\
d^y(j+\ells)
\end{smallmatrix}
\\
\hline
\begin{smallmatrix}
d^u(j)\\
\vdotsS\\
d^u(j+\ells-1)
\end{smallmatrix}
\end{array}
\right]^\top
\preceq \Pi \Theta \Pi^\top.
\end{align*}
This condition highlights that $\Theta$, modulo row permutations, expresses an energy bound on the sequence of vectors $\big(d^y(0), \dots, d^y(\ells), d^u(0), \dots, d^u(\ells-1)\big)$, \dots, $\big(d^y(T-\ells), \dots, d^y(T), d^u(T-\ells), \dots, d^u(T-1)\big)$.
\end{remark}

A typical way to obtain the energy bound given by $\Theta$ is presented in the next remark.

\begin{remark}
\label{rem:conversion_to_energy_bound}
Suppose we know that for some $\bar{d}^y \ge 0$ and $\bar{d}^u \ge 0$,
\begin{align*}
\| d^y \|_{\mathcal{L}_\infty}
\le \bar{d}^y
\text{ and }
\| d^u \|_{\mathcal{L}_\infty}
\le \bar{d}^u.
\end{align*}
Then, for each $k = \ell^\star \! - 1$, \dots, $T-1$,
\begin{align*}
& | \delta(k) |^2 := 
\left|
\smat{
d^y(k+1)\\
d^y(k-\ells+1)\\
\vdotsS\\
d^y(k)\\
d^u(k-\ells+1)\\
\vdotsS\\
d^u(k)\\
}
\right|^2 = \sum_{j=k-\ells+1}^{k+1} | d^y(j) |^2 \\
& + \sum_{j=k-\ells+1}^{k} | d^u(j) |^2  \ \le \ (\ells +1) (\bar{d}^y)^2 + \ells (\bar{d}^u)^2
\end{align*}
and, from~\eqref{Delta10},
\begin{align*}
\Delta_{10} \Delta_{10}^\top
& =
\sum_{k=\ells - 1}^{T-1}
\delta(k)
\delta(k)^\top
\preceq 
\sum_{k=\ells - 1}^{T-1}
\left| \delta(k) \right|^2 I \\
& \preceq (T - \ells +1) \big( (\ells +1) (\bar{d}^y)^2 + \ells (\bar{d}^u)^2 \big) I.
\end{align*}
In this way, we can take $\Theta$ as $(T - \ells +1) \big( (\ells +1) (\bar{d}^y)^2 + \ells (\bar{d}^u)^2 \big) I$.
\end{remark}

Under this setting, where the actual parameters $Z_\ells$ arising from~\eqref{sys_x} are unknown, we introduce the set of parameters $Z$ consistent with data in~\eqref{exp_data_mat} and with the noise bound in~\eqref{setD} as
\begin{align}\label{setC}
& \mathcal{C} := \Big\{ Z \in \real^{p \times (p\ells + m\ells)} \colon  \Psi_1 = (\mathbf{F}_\ells + \mathbf{L}_\ells Z ) \Psi_0 \notag \\
& \hspace*{15mm}+ \mathbf{B}_\ells U_1 + \mathbf{L}_\ells [I_p \ -Z]  \Delta, \ \ \Delta \in \mathcal{D} \Big\}. 
\end{align}
Particularly, \eqref{exp_data_mat} and $\Delta_{10}\in\mathcal{D}$ imply that $Z_\ells \in \mathcal{C}$.
In the sequel, we will then need to work with the set \eqref{setC}, instead of with the unknown $Z_\ells$.
Moreover, we make the next assumption on data.
\begin{assumption}
\label{ass:Psi0}
$\Psi_0 \Psi_0^\top \succ \Theta_{22}$.
\end{assumption}
The interpretation and implications of this assumption are discussed in Section~\ref{sec:disc_asmpt_pers_exc}.

As in the data collection experiment, we rely on current and past outputs and inputs to control~\eqref{sys_x_star} and render the origin asymptotically stable.
In other words, we use the feedback law
\begin{equation}
\label{u(k)=K_stack}
u(k) = \mathbf{K}
\left[
\begin{array}{c}
\begin{smallmatrix}
y(k-\ells)\\
\vdotsS\\
y(k-1)
\end{smallmatrix}\\
\hline
\begin{smallmatrix}
u(k-\ells)\\
\vdotsS\\
u(k-1)
\end{smallmatrix}
\end{array}
\right], \forall k \ge \ells
\end{equation}
for some matrix $\mathbf{K}$ to be designed.
More precisely, such feedback law corresponds to a dynamic controller
\begin{subequations}
\label{dyn_contr_star}
\begin{align}
\chi^+ & = \mathbf{F}_\ells \chi + \mathbf{L}_\ells y + \mathbf{B}_\ells u \label{dyn_contr_star:chi}\\
u & = \mathbf{K} \chi \label{dyn_contr_star:u=Kchi}
\end{align}
\end{subequations}
where the matrices $\mathbf{F}_\ells$, $\mathbf{L}_\ells$, $\mathbf{B}_\ells$ are completely known, see~\eqref{sys_xi:A0_L_B}.
Indeed, \eqref{dyn_contr_star:chi}  yields
\begin{align*}
& \chi_1^+ = \chi_2, \dots,\, \chi_{\ells-1}^+ = \chi_\ells, \, \chi_\ells^+ = y, \\
& \chi_{\ells+1}^+ = \chi_{\ells+2}, \dots,\, \chi_{\ells+\ells-1}^+ = \chi_{\ells+\ells}, \, \chi_{\ells+\ells}^+ = u
\end{align*}
and, when fed with input sequences $\{u(k)\}_{k = 0}^\infty$ and $\{y(k)\}_{k = 0}^\infty$, \eqref{dyn_contr_star:chi} ensures that the solution $\chi(\cdot)$ satisfies
\begin{align}
\label{chi_k_stack_past_io_star}
\left[
\begin{array}{c}
\begin{smallmatrix}
\chi_1(k)\\
\vdotsS\\
\chi_\ells(k)
\end{smallmatrix}\\
\hline
\begin{smallmatrix}
\chi_{\ells+1}(k)\\
\vdotsS\\
\chi_{\ells+\ells}(k)
\end{smallmatrix}
\end{array}
\right]
= 
\left[
\begin{array}{c}
\begin{smallmatrix}
y(k-\ells)\\
\vdotsS\\
y(k-1)
\end{smallmatrix}\\
\hline
\begin{smallmatrix}
u(k-\ells)\\
\vdotsS\\
u(k-1)
\end{smallmatrix}
\end{array}
\right], \forall k \ge \ells.
\end{align}
This shows that \eqref{dyn_contr_star:chi} creates the stack of the past $\ells$ values of output and input that are needed in~\eqref{u(k)=K_stack}.

With all the ingredients in place, we can give our problem statement.
\begin{problem}
\label{probl}
With collected data $\{ u^{\tu{m}}(k), y^{\tu{m}}(k) \}_{k=0}^T$ and under Assumptions~\ref{ass:observ}-\ref{ass:Psi0}, design a matrix $\mathbf{K}$ for the dynamic controller in~\eqref{dyn_contr_star} such that the feedback interconnection of~\eqref{sys_x_star} and \eqref{dyn_contr_star} ensures that $(x, \chi) = 0$ is globally asymptotically stable.
\end{problem}
As will be discussed later, the assumptions made in Problem~\ref{probl} impose a constraint on the number of outputs, as discussed in Section~\ref{sec:disc_asmpt_pers_exc}.
While we conclude this section by explaining our route to the results solving Problem~\ref{probl}, we reconsider this problem in Section~\ref{sec:results_pells_neq_n} to demonstrate how these results can be extended to a general MIMO system.

In the sequel we consider the auxiliary system
\begin{align}
\label{sys_xi_star}
\xi^+ = \mathbf{A}_\ells \xi + \mathbf{B}_\ells v = (\mathbf{F}_\ells + \mathbf{L}_\ells Z_\ells) \xi + \mathbf{B}_\ells v,
\end{align}
see \eqref{sys_xi}.
We do so because this system can capture the input-output evolution of solutions to~\eqref{sys_x_star} in the sense of Lemma~\ref{lemma:io_sys_x_is_io_sys_xi}, namely, under the assumptions of Lemma~\ref{lemma:io_sys_x_is_io_sys_xi} (suitable initial condition $\xi(\ells)$ and same input sequence $\{ v(k) \}_{k = \ells}^\infty = \{ u(k) \}_{k = \ells}^\infty$) we have
\begin{align}
\label{xi_k_stack_past_io_star}
\xi(k)
=
\left[
\begin{array}{c}
\begin{smallmatrix}
y(k-\ells)\\
\vdotsS\\
y(k-1)
\end{smallmatrix}\\
\hline
\begin{smallmatrix}
u(k-\ells)\\
\vdotsS\\
u(k-1)
\end{smallmatrix}
\end{array}
\right], \quad \forall k \ge \ells.
\end{align}
By~\eqref{chi_k_stack_past_io_star} and \eqref{xi_k_stack_past_io_star},
\begin{align*}
u(k) = \mathbf{K} \chi(k) = \mathbf{K} \xi(k), \quad \forall k \ge \ells.
\end{align*}
If we set $v(k) = u(k) = \mathbf{K} \xi(k)$ for all $k \ge \ells$, this corresponds, loosely speaking, to interconnecting the auxiliary system~\eqref{sys_xi_star} with the feedback $v = u = \mathbf{K} \xi$ as
\begin{align}
\label{sys_xi_star_closed_loop}
\xi^+ = \mathbf{A}_\ells \xi + \mathbf{B}_\ells \mathbf{K} \xi = (\mathbf{F}_\ells + \mathbf{L}_\ells Z_\ells + \mathbf{B}_\ells \mathbf{K} ) \xi
\end{align}
and motivates the relevance of~\eqref{sys_xi_star} for Problem~\ref{probl}.
In principle, if $Z_\ells$ were known in~\eqref{sys_xi_star_closed_loop}, we would like to render $\mathbf{F}_\ells + \mathbf{L}_\ells Z_\ells + \mathbf{B}_\ells \mathbf{K}$ Schur.
In lieu of the knowledge of $Z_\ells$, we need to exploit the information available from data and embedded in the set $\mathcal{C}$, and find $\mathbf{K}$ such that $\mathbf{F}_\ells + \mathbf{L}_\ells Z + \mathbf{B}_\ells \mathbf{K}$ is Schur for all $Z\in \mathcal{C}$.
This is equivalent to the robust control problem
\begin{subequations}
\label{rob_contr_probl_sys_xi}
\begin{align}
& \text{find} & & \mathbf{K}, P = P^\top \succ 0 \\
& \text{s.t.} & & (\mathbf{F}_\ells + \mathbf{L}_\ells Z + \mathbf{B}_\ells \mathbf{K}) P (\mathbf{F}_\ells + \mathbf{L}_\ells Z + \mathbf{B}_\ells \mathbf{K})^\top   \notag \\
& & & - P  \prec 0 \qquad \forall Z \in \mathcal{C}. \label{rob_contr_probl_sys_xi:ineq}
\end{align}
\end{subequations}

We show in Section~\ref{sec:reform_setC} (see Lemma~\ref{lemma:cons_asmpt_for_set_C} below) that, under Assumption~\ref{ass:Psi0}, we can rewrite the set $\mathcal{C}$ in a form instrumental to render $\mathbf{F}_\ells + \mathbf{L}_\ells Z + \mathbf{B}_\ells \mathbf{K}$ Schur for all $Z \in \mathcal{C}$. Moreover, the set $\mathcal{C}$ is bounded, which is beneficial in solving~\eqref{rob_contr_probl_sys_xi} (as opposed to an unbounded set $\mathcal{C}$).
Section~\ref{sec:stab_aux_sys} aims at finding a semidefinite program equivalent to~\eqref{rob_contr_probl_sys_xi} in terms of feasibility.
Section~\ref{sec:from_aux_sys_to_actual_sys} shows rigorously how to transfer the stabilization of~\eqref{sys_xi_star}, by a $\mathbf{K}$ that satifies~\eqref{rob_contr_probl_sys_xi}, to the stabilization of~\eqref{sys_x_star}, as required by Problem~\ref{probl}.
Moreover, Section~\ref{sec:disc_asmpt_pers_exc} discusses the implications of Assumption~\ref{ass:Psi0} and shows, as a key result (see Lemma~\ref{lemma:implicationSNRasmpt} below), that observability of $(\As,\Cs)$ and Assumption \ref{ass:Psi0} imply $p \ells = n$.
Based on this discussion, we conclude that our assumptions imply $p \ells = n$;
so, we first give a solution for the case $p \ells = n$; in Section~\ref{sec:results_pells_neq_n} we revisit Problem~\ref{probl} and show how to extend the results to the case $p \ells \neq n$.

%%%%%%%%%%%%%%%%%%%%%%%%%%%%%%%%%%%%%%%%%%%%%%%%%%%%%%%%%%%%%%%%%%%%%%%%%%%%%%%%%%%
%%%%%%%%%%%%%%%%%%%%%%%%%%%%%%%%%%%%%%%%%%%%%%%%%%%%%%%%%%%%%%%%%%%%%%%%%%%%%%%%%%%
%%%%%%%%%%%%%%%%%%%%%%%%%%%%%%%%%%%%%%%%%%%%%%%%%%%%%%%%%%%%%%%%%%%%%%%%%%%%%%%%%%%
%%%%%%%%%%%%%%%%%%%%%%%%%%%%%%%%%%%%%%%%%%%%%%%%%%%%%%%%%%%%%%%%%%%%%%%%%%%%%%%%%%%
\section{Results}
\label{sec:results_pells=n}

%%%%%%%%%%%%%%%%%%%%%%%%%%%%%%%%%%%%%%%%%%%%%%%%%%%%%%%%%%%%%%%%%%%%%%%%%%%%%%%%%%%
%%%%%%%%%%%%%%%%%%%%%%%%%%%%%%%%%%%%%%%%%%%%%%%%%%%%%%%%%%%%%%%%%%%%%%%%%%%%%%%%%%%
\subsection{Reformulation of set $\mathcal{C}$ of parameters consistent with data}
\label{sec:reform_setC}

The set $\mathcal{C}$ in~\eqref{setC} can be equivalently rewritten as
\begin{align*}
\mathcal{C} = \{ Z  \colon  \Psi_1 & - \mathbf{F}_\ells \Psi_0 - \mathbf{B}_\ells U_1 
= \mathbf{L}_\ells Z \Psi_0  \\
& + \mathbf{L}_\ells [I_p \ - \! Z]  \Delta, \ \ \Delta \Delta^\top \preceq \Theta \}.
\end{align*}
From~\eqref{exp_data_mat}, we observe that
\begin{align*}
\Psi_1 - \mathbf{F}_\ells \Psi_0 - \mathbf{B}_\ells U_1 
=
\mathbf{L}_\ells Z_\ells \Psi_0 + \mathbf{L}_\ells [I_p \ - \! Z_\ells] \Delta_{10}
\end{align*}
and, hence, all block rows of $\Psi_1 - \mathbf{F}_\ells \Psi_0 - \mathbf{B}_\ells U_1$ are zero except for block row $\ells$.
Block row $\ells$ can be obtained as
\begin{align*}
\mathbf{L}_\ells^\top (\Psi_1 - \mathbf{F}_\ells \Psi_0 - \mathbf{B}_\ells U_1) = \mathbf{L}_\ells^\top \Psi_1
\end{align*}
by the definition of $\mathbf{F}_\ells$ and $\mathbf{B}_\ells$ as in~\eqref{sys_xi:A0_L_B}.
Thanks to these observations, $\mathcal{C}$ is equivalently rewritten as
\begin{align*}
& \mathcal{C} = \{ Z  \colon \mathbf{L}_\ells^\top \Psi_1 -  Z \Psi_0
 = [I_p \ - \! Z]  \Delta, \ \Delta \Delta^\top \preceq \Theta \}
\end{align*}
where we use $\mathbf{L}_\ells^\top \mathbf{L}_\ells = I_p$.
To further arrange the set $\mathcal{C}$ in a convenient form, we use the next key result from~\cite{stateinputerrors}.

\begin{proposition}[\hspace*{-.8ex}{\cite[Prop.~1]{stateinputerrors}}]
\label{proposition:matrix_elim}
Consider matrices $E \in \real^{n_1 \times n_2}$, $F\in\real^{n_1 \times n_3}$, $G \in \real^{n_3 \times n_3}$ 
with $G = G^\top \succeq 0$. Then,
\begin{subequations}
\begin{align}
\label{matrix_elim_lemma_ineq}
E E^\top \preceq F G F^\top
\end{align}
if and only if there exists $D \in \real^{n_3 \times n_2}$ such that
\begin{align}
\label{matrix_elim_lemma_eq_ineq}
E = F D,\quad D D^\top \preceq G.
\end{align}
\end{subequations}
\end{proposition}

Since going from \eqref{matrix_elim_lemma_eq_ineq} to \eqref{matrix_elim_lemma_ineq} dispenses with matrix $D$, Proposition~\ref{proposition:matrix_elim} can be interpreted as a matrix elimination result.
By Proposition~\ref{proposition:matrix_elim},
\begin{align*}
\mathcal{C} = \Big\{ Z  \colon & (\mathbf{L}_\ells^\top \Psi_1 -  Z \Psi_0)(\mathbf{L}_\ells^\top \Psi_1 -  Z \Psi_0)^\top \\
& \preceq [I_p \ - \! Z] \Theta [I_p \ - \! Z]^\top  \\
& \overset{\eqref{Theta_partitioned}}{=}  [I_p \ - \! Z] \smat{\Theta_{11} & \Theta_{12}\\ \Theta_{12}^\top & \Theta_{22}} [I_p \ - \! Z]^\top 
\Big\}.
\end{align*}
By expanding the products and defining $\mathscr{A}$, $\mathscr{B}$, $\mathscr{C}$, we have
\begin{subequations}\label{eq:setC}
\begin{align}
& \mathcal{C} = \big\{ Z \colon Z \mathscr{A} Z^{\top} + Z \mathscr{B}^{\top} + \mathscr{B} Z^{\top} + \mathscr{C} \preceq 0 \big\}, \label{setC_ABC}\\
& \mathscr{A} := \Psi_0 \Psi_0^\top - \Theta_{22},\mathscr{B} :=  -  \mathbf{L}_\ells^\top \Psi_1 \Psi_0^{\top} + \Theta_{12}, \label{setC:AB}\\
& \mathscr{C} := \mathbf{L}_\ells^\top \Psi_1 \Psi_1^\top \mathbf{L}_\ells - \Theta_{11}. \label{setC:C}
\end{align}
\end{subequations}
Assumption~\ref{ass:Psi0} amounts to asking $\mathscr{A} \succ 0$, hence, $\mathscr{A}^{-1}$ exists.
We can then define
\begin{align}
& \mathscr{Z} := - \mathscr{B} \mathscr{A}^{-1}, \quad \mathscr{Q} := \mathscr{B} \mathscr{A}^{-1} \mathscr{B}^\top - \mathscr{C}. \label{setC:ZQ}
\end{align}
Assumption~\ref{ass:Psi0} guarantees the next result.
\begin{lemma}
\label{lemma:cons_asmpt_for_set_C}
Under Assumption~\ref{ass:Psi0}, i.e., $\mathscr{A} \succ 0$, we have that: 
\begin{align}
\mathcal{C} & = \big\{ Z\colon (Z - \mathscr{Z}) \mathscr{A} (Z - \mathscr{Z})^\top \preceq \mathscr{Q} \big\}; \label{setC_ZAQ} \\
\mathscr{Q}  & \succeq 0 \, ; \notag \\
\mathcal{C} & = \big\{ \mathscr{Z} + \mathscr{Q}^{1/2} \Upsilon \mathscr{A}^{-1/2} \colon \Upsilon \Upsilon^\top \preceq  I_p \big\}; \label{setC_unitBall}
\end{align}  
$\mathcal{C}$ is bounded with respect to any matrix norm.
\end{lemma}
\begin{proof}
See Appendix~\ref{app:proof_lemma_cons_asmpt_for_set_C}.
\end{proof}

%%%%%%%%%%%%%%%%%%%%%%%%%%%%%%%%%%%%%%%%%%%%%%%%%%%%%%%%%%%%%%%%%%%%%%%%%%%%%%%%%%%
%%%%%%%%%%%%%%%%%%%%%%%%%%%%%%%%%%%%%%%%%%%%%%%%%%%%%%%%%%%%%%%%%%%%%%%%%%%%%%%%%%%
\subsection{Stabilization of the auxiliary system}
\label{sec:stab_aux_sys}

After Problem~\ref{probl}, we motivated the relevance of the auxiliary system \eqref{sys_xi_star} in the context of our setting. 
Since the actual $Z_\ells$ is unknown, we would like to find $\mathbf{K}$ rendering $\mathbf{F}_\ells + \mathbf{L}_\ells Z + \mathbf{B}_\ells \mathbf{K}$ Schur for all $Z \in \mathcal{C}$, i.e., solve \eqref{rob_contr_probl_sys_xi}.
By the approach in~\cite{AndreaPetersen2022} and the use of Petersen's lemma \cite{petersen1987stabilization}, feasibility of~\eqref{rob_contr_probl_sys_xi} can be equivalently reformulated as in the next result.

\begin{lemma}
\label{lemma:Petersen}
Under Assumption~\ref{ass:Psi0}, feasibility of~\eqref{rob_contr_probl_sys_xi} is equivalent to feasibility of
\begin{subequations}
\label{rob_contr_probl_LMI}
\begin{align}
& \text{find} & & \mathbf{Y}, P = P^\top \succ 0 \label{rob_contr_probl_LMI:find}\\
& \text{s.t.} & & 
\bmat{
-P - \mathbf{L}_\ells\mathscr{C}\mathbf{L}_\ells^\top & \mathbf{F}_\ells P + \mathbf{B}_\ells \mathbf{Y} & \mathbf{L}_\ells \mathscr{B}\\
P \mathbf{F}_\ells^\top + \mathbf{Y}^\top \mathbf{B}_\ells^\top & - P & -P \\
\mathscr{B}^\top \mathbf{L}_\ells^\top & -P & -\mathscr{A}} \prec 0. \label{rob_contr_probl_LMI:ineq}
\end{align}
\end{subequations}
If \eqref{rob_contr_probl_LMI} is feasible, a $\mathbf{K}$ satisfying \eqref{rob_contr_probl_sys_xi} is $\mathbf{K} = \mathbf{Y}P^{-1}$.
\end{lemma}
\begin{proof}
See Appendix~\ref{app:proof_lemma_Petersen}.
\end{proof}

The matrix inequality \eqref{rob_contr_probl_LMI:ineq} is equivalent, from~\eqref{LMIinZQA} and by Schur complement, to the linear matrix inequality (in $\mathbf{Y}$ and $P$)
\begin{align*}
& 0 \succ 
\bmat{
- P + \mathbf{L}_\ells \mathscr{Q} \mathbf{L}_\ells^\top & \mathbf{F}_\ells P + \mathbf{L}_\ells \mathscr{Z} P + \mathbf{B}_\ells \mathbf{Y} & 0\\
\star & -P & P\\
\star & \star & -\mathscr{A}
}.
\end{align*}

%%%%%%%%%%%%%%%%%%%%%%%%%%%%%%%%%%%%%%%%%%%%%%%%%%%%%%%%%%%%%%%%%%%%%%%%%%%%%%%%%%%
%%%%%%%%%%%%%%%%%%%%%%%%%%%%%%%%%%%%%%%%%%%%%%%%%%%%%%%%%%%%%%%%%%%%%%%%%%%%%%%%%%%
\subsection{From stabilization of the auxiliary system to stabilization of the actual system}
\label{sec:from_aux_sys_to_actual_sys}

To transfer the result in Lemma~\ref{lemma:Petersen} on the stabilization of~\eqref{sys_xi_star} to a result on the stabilization of~\eqref{sys_x_star}, as required by Problem~\ref{probl}, we need to investigate the relation between their respective solutions.

We consider the generic system in~\eqref{sys_x} with $(A,C)$ observable so that the matrices $\mathcal{O}_\ell$, $\mathcal{T}_\ell$, $\mathcal{R}_\ell$ can be defined as in~\eqref{O_ell_T_ell_R_ell}, with $\rank \mathcal{O}_\ell = n$.
We consider also the system in~\eqref{sys_xi} with matrices $\mathbf{F}_\ell$, $\mathbf{L}_\ell$, $\mathbf{B}_\ell$, $Z_\ell$ selected, as in~\eqref{sys_xi:A0_L_B_Z}, from the matrices $(A,B,C)$ of~\eqref{sys_x}.
System \eqref{sys_x} is the counterpart of~\eqref{sys_x_star}, and \eqref{sys_xi} is the counterpart of~\eqref{sys_xi_star}. We use \eqref{sys_x} and \eqref{sys_xi} because the considerations in this section are general.

Suppose that there exists a control law $v = \mathbf{K} \xi$ that asymptotically stabilizes \eqref{sys_xi}, i.e., that makes the matrix $\mathbf{A}_\ell+\mathbf{B}_\ell\mathbf{K} = \mathbf{F}_\ell+\mathbf{L}_\ell Z_\ell+\mathbf{B}_\ell\mathbf{K}$ Schur.
Our goal for this section is to show that the same $\mathbf{K}$ can asymptotically stabilize \eqref{sys_x} if the dynamic controller
\begin{subequations}
\label{dyn_contr}
\begin{align}
\chi^+ & = \mathbf{F}_\ell \chi + \mathbf{L}_\ell y + \mathbf{B}_\ell u \label{dyn_contr:chi}\\
u & = \mathbf{K} \chi \label{dyn_contr:u=Kchi}
\end{align}
\end{subequations}
is put in feedback with~\eqref{sys_x}, as depicted in Figure~\ref{fig:plant+dyn_contr}.
The scheme in Figure~\ref{fig:plant+dyn_contr} can be equivalently rearranged as on the left of Figure~\ref{fig:x_chi_CL_vs_xi_CL}, which leads us to consider the combination of~\eqref{sys_x} and \eqref{dyn_contr:chi}:
\begin{subequations}
\label{sys_x_chi}
\begin{align}
x^+ &  = A x + B u, \quad y = C x, \label{sys_x_chi:x}\\
\chi^+ &  = \mathbf{F}_\ell \chi + \mathbf{L}_\ell y + \mathbf{B}_\ell u. \label{sys_x_chi:chi}
\end{align}
\end{subequations}

We have the relation for the input-output evolution of~\eqref{sys_xi} and \eqref{sys_x_chi} as in the next result.

\begin{lemma}
\label{lemma:io_(x,chi)_vs_xi}
Let $(A,C)$ be observable.
For each $\hat{x}$, $\hat{\chi}$ and sequence $\{ u(k) \}_{k = 0}^\infty$, there exists $\hat{\xi}$ such that:%
\begin{itemize}[nosep,noitemsep,left=0pt]
\item the solution $\smat{x(\cdot)\\ \chi(\cdot)}$ to~\eqref{sys_x_chi} with initial condition $\smat{x(0)\\ \chi(0)} = \smat{\hat{x}\\ \hat{\chi}}$ and with input $\{ u(k) \}_{k = 0}^\infty$,
\item the corresponding output response $y(\cdot) = C x(\cdot)$, and 
\item the solution $\xi(\cdot)$ to \eqref{sys_xi} with initial condition $\xi(\ell) = \hat{\xi}$ and input $\{ v(k) \}_{k = \ell}^\infty = \{ u(k) \}_{k = \ell}^\infty$
\end{itemize}
satisfy
\begin{align*}
\xi(k) = 
\left[
\begin{array}{c}
\begin{smallmatrix}
y(k-\ell)\\
\vdotsS\\
y(k-1)
\end{smallmatrix}\\
\hline
\begin{smallmatrix}
u(k-\ell)\\
\vdotsS\\
u(k-1)
\end{smallmatrix}
\end{array}
\right]
= \chi(k), \quad \forall k \ge \ell.
\end{align*}
\end{lemma}
\begin{proof}
See Appendix~\ref{app:proof_lemma_io_(x,chi)_vs_xi}.
\end{proof}

\begin{figure}[t]
\centerline{\includegraphics[scale=0.75]{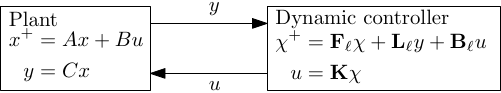}}
\caption{Plant and dynamic controller.}
\label{fig:plant+dyn_contr}
\end{figure}

\begin{figure}[t]
\centerline{\includegraphics[scale=0.7]{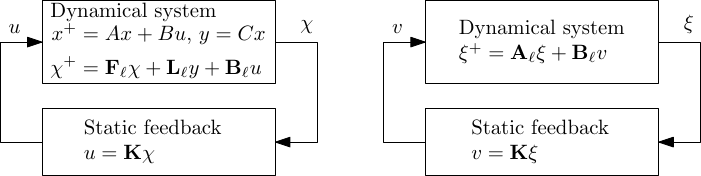}}
\caption{(Left) Plant in~\eqref{sys_x} and dynamic controller in~\eqref{dyn_contr}. (Right) Auxiliary system in~\eqref{sys_xi} and static controller.}
\label{fig:x_chi_CL_vs_xi_CL}
\end{figure}

With Lemma~\ref{lemma:io_(x,chi)_vs_xi} in place, we can turn to the main result of this section, which is the next one.

\begin{lemma}
\label{lemma:GAS_cl_aux_implies_GAS_cl}
Let $(A,C)$ be observable.
If $\mathbf{K}$ makes the matrix $\mathbf{A}_\ell+ \mathbf{B}_\ell \mathbf{K} = \mathbf{F}_\ell+\mathbf{L}_\ell Z_\ell+\mathbf{B}_\ell\mathbf{K}$ Schur, then $(x,\chi)=0$ is globally asymptotically stable for
\begin{subequations}
\label{sys_x_chi_CL}
\begin{align}
x^+ & = A x + B u, y = Cx \label{sys_x_chi_CL:x}\\
\chi^+ &  = \mathbf{F}_\ell \chi + \mathbf{L}_\ell y + \mathbf{B}_\ell u \label{sys_x_chi_CL:chi}\\
u & = \mathbf{K} \chi. \label{sys_x_chi_CL:u}
\end{align}
\end{subequations}
\end{lemma}
\begin{proof}
See Appendix~\ref{app:proof_lemma_GAS_cl_aux_implies_GAS_cl}.
\end{proof}

By standard properties of linear systems, an immediate consequence of Lemma~\ref{lemma:GAS_cl_aux_implies_GAS_cl} is that the matrix $\smat{
A & B \mathbf{K} \\
\mathbf{L}_\ell C & \mathbf{F}_\ell + \mathbf{B}_\ell \mathbf{K}
}$, which corresponds to the closed-loop system in~\eqref{sys_x_chi_CL}, is Schur.
The next result is obtained immediately from Lemma~\ref{lemma:Petersen}, Lemma~\ref{lemma:io_(x,chi)_vs_xi}, Lemma~\ref{lemma:GAS_cl_aux_implies_GAS_cl}.

\begin{theorem}\label{thm:pl_eq_n}
With collected data $\{u^{\tu{m}}(k), y^{\tu{m}}(k) \}_{k=0}^T$ and under Assumptions~\ref{ass:observ}-\ref{ass:Psi0}, suppose \eqref{rob_contr_probl_LMI} is feasible and $\mathbf{K}$ is a controller returned from~\eqref{rob_contr_probl_LMI}.
Then, $(x,\chi) = 0$ is globally asymptotically stable for the feedback interconnection of the unknown plant
\begin{subequations}
\label{sys_xstar_chi_CL}
\begin{align}
x^+ & = A^\star x + B^\star u, y = C^\star x
\end{align}
and the controller
\begin{align}
\chi^+ &  =  ( \mathbf{F}_\ells + \mathbf{B}_\ells \mathbf{K} )  \chi + \mathbf{L}_\ells y \label{sys_xstar_chi_CL_1} \\
u & = \mathbf{K} \chi. \label{sys_xstar_chi_CL_2}
\end{align}
\end{subequations}
Moreover, for each $k \ge \ells$, 
\begin{align*}
\chi(k)= \big(y(k \! - \! \ells),  \dots  , y(k \! - \! 1), u(k \! - \! \ells),  \dots , u(k \! - \! 1)\big).
\end{align*}
\end{theorem}
\begin{proof}
Lemma~\ref{lemma:Petersen} ensures that a $\mathbf{K}$ returned from~\eqref{rob_contr_probl_LMI} makes the matrix $\mathbf{F}_\ells + \mathbf{L}_\ells Z + \mathbf{B}_\ells \mathbf{K}$ Schur for all $Z \in \mathcal{C}$, and in particular for $Z = Z_\ells$, corresponding to $(A^\star, B^\star, C^\star)$.
Since the matrix $\mathbf{A}_\ells + \mathbf{B}_\ells \mathbf{K}$ is Schur, Lemma~\ref{lemma:GAS_cl_aux_implies_GAS_cl} ensures that $(x, \chi)=0$ is globally asymptotically stable for~\eqref{sys_xstar_chi_CL}.
Moreover, Lemma~\ref{lemma:io_(x,chi)_vs_xi} ensures that $\chi(k)= \big(y(k-\ells), \dots, y(k-1), u(k-\ells), \dots, u(k-1)\big)$ for each $k \ge \ells$.
\end{proof}

Theorem~\ref{thm:pl_eq_n} provides our solution to Problem~\ref{probl}, since \eqref{sys_xstar_chi_CL_1} and \eqref{sys_xstar_chi_CL_2} are equivalent to \eqref{dyn_contr_star:chi} and \eqref{dyn_contr_star:u=Kchi}.

%%%%%%%%%%%%%%%%%%%%%%%%%%%%%%%%%%%%%%%%%%%%%%%%%%%%%%%%%%%%%%%%%%%%%%%%%%%%%%%%%%%
%%%%%%%%%%%%%%%%%%%%%%%%%%%%%%%%%%%%%%%%%%%%%%%%%%%%%%%%%%%%%%%%%%%%%%%%%%%%%%%%%%%
\subsection{Interpretation and implications of Assumption~\ref{ass:Psi0}}
\label{sec:disc_asmpt_pers_exc}

To discuss Assumption~\ref{ass:Psi0}, we introduce
\begin{subequations}
\label{N0_S0}
\begin{align}
N_0 := &
\left[
\begin{array}{c}
\begin{smallmatrix}
d^y(0) & d^y(1) & \dots & d^y(T-\ells)\\
\vdotsS & \vdotsS &   & \vdotsS \\
d^y(\ells-1) & d^y(\ells) & \dots & d^y(T-1) \\
\end{smallmatrix}\\[8pt]
\hline
\begin{smallmatrix}
d^u(0) & d^u(1) & \dots & d^u(T-\ells) \rule{0pt}{8pt}\\
\vdotsS & \vdotsS &  & \vdotsS\\
d^u(\ells-1) & d^u(\ells) & \dots & d^u(T-1)
\end{smallmatrix}
\end{array}
\right], \label{N0}\\
S_0 :=& 
\left[
\begin{array}{c}
\begin{smallmatrix}
y(0) & y(1) & \dots & y(T-\ells)\\
\vdotsS & \vdotsS &   & \vdotsS \\
y(\ells-1) & y(\ells) & \dots & y(T-1) \\
\end{smallmatrix}\\[8pt]
\hline
\begin{smallmatrix}
u(0) & u(1) & \dots & u(T-\ells)  \rule{0pt}{8pt}\\
\vdotsS & \vdotsS &  & \vdotsS\\
u(\ells-1) & u(\ells) & \dots & u(T-1)
\end{smallmatrix}
\end{array}
\right] \label{S0}
\end{align}
\end{subequations}
where $S_0$ corresponds to the ``clean'' data we would collect from~\eqref{sys_x_star} if input and output noise $d^u$ and $d^y$ were absent.
From~\eqref{u_m,y_m} and \eqref{Psi1_Psi0_U0_Delta10}, we have that
\begin{equation}
\label{Psi0=S0+N0}
\Psi_0 = S_0 + N_0.
\end{equation}
The next result shows that observability of $(\As,\Cs)$ and Assumption~\ref{ass:Psi0} imply $p\ells = n$.
\begin{lemma}
\label{lemma:implicationSNRasmpt}
Let $(\As,\Cs)$ be observable. Then, 
\begin{align*}
& \text{Assumption~\ref{ass:Psi0}} \implies \rank S_0 = (p+m)\ells \\
& \implies \rank \mathcal{O}_{\ells} = p \ells \implies p\ells=n.
\end{align*}
\end{lemma}
\begin{proof}
See Appendix~\ref{app:proof_lemma_implicationSNRasmpt}.
\end{proof}

Besides the implication of Assumption~\ref{ass:Psi0} in Lemma~\ref{lemma:implicationSNRasmpt}, the next result provides a sufficient condition for Assumption~\ref{ass:Psi0}, which is of signal-to-noise-ratio type.
\begin{lemma}
\label{lemma:SNR}
If
\begin{align}
\label{ratio_sing_values}
\frac{ \sigma_{\min}(S_0 S_0^\top) }{\sigma_{\max} (\Theta_{22})} > 4,
\end{align}
Assumption~\ref{ass:Psi0} holds.
\end{lemma}
\begin{proof}
See Appendix~\ref{app:proof_lemma_SNR}.
\end{proof}

%%%%%%%%%%%%%%%%%%%%%%%%%%%%%%%%%%%%%%%%%%%%%%%%%%%%%%%%%%%%%%%%%%%%%%%%%%%%%%%%%%%
%%%%%%%%%%%%%%%%%%%%%%%%%%%%%%%%%%%%%%%%%%%%%%%%%%%%%%%%%%%%%%%%%%%%%%%%%%%%%%%%%%%
%%%%%%%%%%%%%%%%%%%%%%%%%%%%%%%%%%%%%%%%%%%%%%%%%%%%%%%%%%%%%%%%%%%%%%%%%%%%%%%%%%%
%%%%%%%%%%%%%%%%%%%%%%%%%%%%%%%%%%%%%%%%%%%%%%%%%%%%%%%%%%%%%%%%%%%%%%%%%%%%%%%%%%%
\subsection{Extension to the case $p \ells > n$}
\label{sec:results_pells_neq_n}

{Observability of $(\As, \Cs)$, as in Assumption~\ref{ass:observ}, implies that $p \ells \ge n$.
We have shown in Lemma~\ref{lemma:implicationSNRasmpt} that, when combined with Assumption~\ref{ass:Psi0}, it implies $p \ells = n$.
Then, when $p\ells > n$, Assumption~\ref{ass:Psi0} is violated and, for instance, Lemma~\ref{lemma:Petersen} cannot be applied to synthesize a controller.
We address this situation and extend the previous results in this subsection.
To this end, we propose to augment the system in~\eqref{sys_x_star} with additional dynamics as}
\begin{subequations}
\label{AugmentSYS}
\begin{align}
\bmat{x^{+} \\ x_a^{+}} & = \underbrace{\bmat{\As & 0 \\ 0 & A_a}}_{=: A_{\aug}}  \bmat{x \\ x_a} + \underbrace{\bmat{\Bs \\ B_a}}_{=: B_{\aug}} u,  \\ 
y_{\aug} & = \underbrace{\bmat{ \Cs & C_a}}_{=: C_{\aug}} \bmat{x \\ x_a}
\end{align}
\end{subequations}
for artificial state\slash input\slash output matrices
\begin{align*}
    A_a \! \in \! \real^{(p \ells \! - n) \! \times \! (p \ells \! - n)}, B_a \! \in \! \real^{(p \ells  \! -  n) \times m}, C_a \! \in \! \real^{p \times (p \ells \! -  n)}.
\end{align*}
The augmented system can be viewed as the parallel connection of the original system $(\As,\Bs,\Cs)$ and the introduced artificial system $(A_a,B_a,C_a)$, as in Figure~\ref{fig:parallel_connection}.
\begin{figure}[H]
\begin{center}
\includegraphics[scale=0.8]{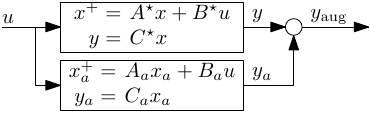}
\caption{Parallel connection of original and artificial systems.}\label{fig:parallel_connection}
\end{center}	
\end{figure}

The idea pursued here is to apply the controller synthesis procedure discussed before to the augmented system~\eqref{AugmentSYS}, which is of dimension $n_{\aug}:=p\ells$. 
So, we consider the matrix
\begin{equation}\label{eq:obserNonsinguAug}
\mathcal{O}_{\ells}^{\aug} := 
\bmat{
\Cs & C_a \\ 
\Cs \As & C_a A_a \\ 
\vdots & \vdots \\ 
\Cs \As^{\ells - 1} & C_a A_a^{\ells - 1}
} \in \real^{n_{\aug} \times n_{\aug}},
\end{equation}
i.e., the $\ells$-step observability matrix of the augmented system \eqref{AugmentSYS}, and make the next assumption on $\mathcal{O}_{\ells}^{\aug}$.

\begin{assumption}
\label{ass:obserNonsinguAug}
The selected artificial state and output matrices $A_a$ and $C_a$ are such that the matrix $\mathcal{O}_{\ells}^{\aug}$ is nonsingular.
\end{assumption}

Assumption~\ref{ass:obserNonsinguAug} implies that the observability index of the augmented system remains $\ells$.
{Moreover, if $\mathcal{O}_{\ells}^{\aug}$ is nonsingular, $\mathcal{O}_{\ells}$ (corresponding to the first $n$ columns of $\mathcal{O}_{\ells}^{\aug}$) must be full column rank and, thus, Assumption~\ref{ass:obserNonsinguAug} subsumes Assumption~\ref{ass:observ}}.
We note that in the sequel we impose a data-dependent condition that serves as a sufficient condition for Assumption \ref{ass:obserNonsinguAug}, see Lemma~\ref{lemma:assumptions_implication+obs_idx}.

As mentioned before, we introduce the artificial system to bring the setup back to the case where the state dimension is equal to $p \ells$. 
An immediate obstacle is that the data collection phase does not exactly match the previous case in Figure \ref{fig:data_collection_exp}. This obstacle is discussed and addressed next. 
Due to the presence of noise, we can collect input-output data from the augmented system \eqref{AugmentSYS} through the noisy augmented dynamics
\begin{subequations}\label{eq:aug-measured}
\begin{align}
\bmat{x^+ \\ {x_a^\tu{m}}^+} & = \bmat{\As & 0 \\ 0 & A_a} \bmat{x \\ x_a^\tu{m}}  + \bmat{\Bs \\ B_a} u^{\tu{m}} \! - \! \bmat{\Bs \\ 0} d^u \label{eq:aug-measured-1} \\ 
y^{\tu{m}}_{\aug} & = \bmat{\Cs & C_a} \bmat{x \\ x_a^\tu{m}} + d^y \overset{\eqref{sys_x_star_ym_um}}{=}  y^{\tu{m}} + C_a x_a^{\tu{m}} \label{eq:aug-measured-2}
\end{align}
\end{subequations}
where the superscript $``\tu{m}"$ denotes the quantities that can be measured.
The input noise $d^u$ affects only the dynamics of the actual system and not the artificial one, and \eqref{eq:aug-measured-1} can be rewritten as 
\begin{align*}
\bmat{x^+ \\ {x_a^\tu{m}}^+} = \bmat{\As & 0 \\ 0 & A_a} \bmat{x \\ x_a^\tu{m}} + \bmat{\Bs \\ B_a} (u^{\tu{m}} \! - \! d^u)+\bmat{0 \\ B_a} d^u.
\end{align*}
Note that neglecting the last term on the right hand side of the previous equation would result in the same setup as in Figure~\ref{fig:data_collection_exp}, where the actual system $(\As, \Bs, \Cs)$ is replaced by the augmented one $(A_{\aug}, B_{\aug}, C_{\aug})$.
While that term cannot be neglected, it can be pushed to the output channel as
\begin{subequations}\label{aug-measured-g}
\begin{align}
\bmat{x^+ \\ x_a^+} & = \bmat{\As & 0 \\ 0 & A_a} \bmat{x \\ x_a} + \bmat{\Bs \\ B_a} (u^{\tu{m}} - d^u) \label{aug-measured-g-1} \\
\hat y^{\tu{m}}_{\aug} & := \bmat{\Cs & C_a} \bmat{x \\ x_a} + d^y + d^{y_a} \label{aug-measured-g-2}
\end{align}  
where, for each $k \ge 0$,
\begin{align}
d^{y_a} (k):= \sum_{j=0}^{k-1} C_a A_a^{k-j-1} B_a d^{u}(j). \label{aug-measured-g-3}
\end{align}
\end{subequations}
The input-output behavior of the two systems \eqref{eq:aug-measured} and \eqref{aug-measured-g} are identical because, starting from the same initial condition $(x(0), x_a^{\tu{m}}(0))=(x(0), x_a(0))$ and applying the same input/noise sequences $\{ u^{\tu{m}}(k)\}_{k=0}^{\infty}$, $\{ d^{u}(k)\}_{k=0}^{\infty}$, $\{ d^{y}(k)\}_{k=0}^{\infty}$, we have 
\begin{equation}\label{eq:y_aug_hat}
\{y_{\aug}^{\tu{m}}(k) \}_{k=0}^{\infty} = \{ {\hat y}_{\aug}^{\tu{m}}(k) \}_{k=0}^{\infty},
\end{equation}
see also Figure \ref{fig:prlel_conect_data}.

\begin{figure}[t]
\begin{center}
\includegraphics[scale=0.85]{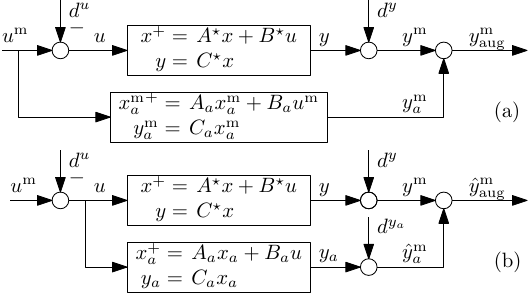}
\caption{Data collection on the augmented system. Subfigures~(a) and (b) correspond respectively to~\eqref{eq:aug-measured} and \eqref{aug-measured-g}.}
\label{fig:prlel_conect_data}
\end{center}	
\end{figure}

In summary, bearing in mind \eqref{aug-measured-g} and \eqref{eq:y_aug_hat}, we can proceed to work with the augmented system with state $x_{\aug} := (x,x_a)$ and $(A_{\aug},B_{\aug},C_{\aug})$ defined in~\eqref{AugmentSYS}:
\begin{subequations}\label{augmentSYS_data}
\begin{align}
x_{\aug}^+  & \! = \!  A_{\aug} x_{\aug} + B_{\aug} (u^{\tu{m}} - d^u), \\
y_{\aug}^{\tu{m}} & \! = \! C_{\aug} x_{\aug} + d^{y} \! + \! d^{y_a} =: C_{\aug} x_{\aug} + d^{y_{\aug}}, \label{augmentSYS_data:output}
\end{align}
\end{subequations}
which is the counterpart of~\eqref{sys_x_star_ym_um}.
Now, for the measured data $\{ u^{\tu{m}}(k), y_{\aug}^{\tu{m}}(k) \}_{k=0}^{T}$, we construct the data matrices
\begin{subequations}\label{Psi1_Psi0_U0_Delta10_aug}
\begin{align}
\Psi_1^{\aug} & :=
\left[
\begin{array}{c}
\begin{smallmatrix}
y^{\tu{m}}_{\aug}(1) & y^{\tu{m}}_{\aug}(2) & \dots & y^{\tu{m}}_{\aug}(T- \ells +1)\\
\vdotsS & \vdotsS &  & \vdotsS \\
y^{\tu{m}}_{\aug}(\ells) & y^{\tu{m}}_{\aug}(\ells +1) & \dots & y^{\tu{m}}_{\aug}(T) \\ 
\end{smallmatrix} \\[8pt]
\hline
\begin{smallmatrix}
u^{\tu{m}}(1) & u^{\tu{m}}(2) & \dots & u^{\tu{m}}(T-\ells +1)\rule{0pt}{8pt} \\
\vdotsS & \vdotsS &   & \vdotsS \\
u^{\tu{m}}(\ells) & u^{\tu{m}}(\ells +1) & \dots & u^{\tu{m}}(T)
\end{smallmatrix}
\end{array}
\right]\\
\Psi_0^{\aug} & :=
\left[
\begin{array}{c}
\begin{smallmatrix}
y^{\tu{m}}_{\aug}(0) & y^{\tu{m}}_{\aug}(1) & \dots & y^{\tu{m}}_{\aug}(T- \ells)\\
\vdotsS & \vdotsS &  & \vdotsS \\
y^{\tu{m}}_{\aug}(\ells-1) & y^{\tu{m}}_{\aug}(\ells) & \dots & y^{\tu{m}}_{\aug}(T-1) \\
\end{smallmatrix} \\[8pt]
\hline
\begin{smallmatrix}
u^{\tu{m}}(0) & u^{\tu{m}}(1) & \dots & u^{\tu{m}}(T- \ells) \rule{0pt}{8pt}\\
\vdotsS & \vdotsS &  & \vdotsS\\
u^{\tu{m}}(\ells-1) & u^{\tu{m}}(\ells) & \dots & u^{\tu{m}}(T-1)
\end{smallmatrix}
\end{array}
\right]\\
U_1^{\aug} & := 
\left[ \begin{smallmatrix} u^{\tu{m}}(\ells) & u^{\tu{m}}(\ells +1) & \dots & u^{\tu{m}}(T) \end{smallmatrix} \right] \\
\Delta_{10}^{\aug} & :=
\left[
\begin{array}{c}
\begin{smallmatrix}
d^{y_{\aug}}(\ells)  & d^{y_{\aug}}(\ells+1)  & \dots & d^{y_{\aug}}(T) \\
\end{smallmatrix} \\[0pt]
\hline
\begin{smallmatrix}
d^{y_{\aug}}(0)  & d^{y_{\aug}}(1)  & \dots & d^{y_{\aug}}(T- \ells) \rule{0pt}{8pt}\\
\vdotsS & \vdotsS &  & \vdotsS \\
d^{y_{\aug}}(\ells-1)  & d^{y_{\aug}}(\ells)  & \dots & d^{y_{\aug}}(T-1)  \\
\end{smallmatrix}\\[8pt]
\hline
\begin{smallmatrix}
d^u(0) & d^u(1) & \dots & d^u(T- \ells)\rule{0pt}{8pt}\\
\vdotsS & \vdotsS &  & \vdotsS \\
d^u(\ells-1) & d^u(\ells-1) & \dots & d^u(T-1)
\end{smallmatrix}
\end{array}
\right]. \label{Delta_10_aug}
\end{align}
\end{subequations}
These matrices are counterparts of those in \eqref{Psi1_Psi0_U0_Delta10} and satisfy the counterpart of~\eqref{exp_data_mat}, namely,
\begin{align}\label{exp_data_mat_aug}
\Psi_1^{\aug} = & (\mathbf{F}_{\ells} + \mathbf{L}_{\ells} Z_{\ells}^{\aug} ) \Psi_0^{\aug} + \mathbf{B}_{\ells} U_1^{\aug} \notag \\
& + \mathbf{L}_{\ells} \bmat{I_p & -Z_{\ells}^{\aug}} \Delta_{10}^{\aug}
\end{align}
where $\mathbf{F}_{\ells}$, $\mathbf{L}_{\ells}$ and $\mathbf{B}_{\ells}$ are as in~\eqref{sys_xi:A0_L_B}, with $\ell=\ells$, and 
\begin{align}
\mathcal{T}_{\ells}^{\aug} & \! := \! \! 
\smat{
0_{p\times m} & 0 & \dotsS & 0 & 0 \\
C_{\aug} B_{\aug} & 0 & \dotsS & 0 & 0\\
C_{\aug} A _{\aug} B_{\aug} & C_{\aug} B_{\aug} & \dotsS & 0 & 0\\
\vdotsS & \vdotsS & \ddotsS & \vdotsS & \vdotsS \\
C_{\aug} A_{\aug}^{\ells \! -2} B_{\aug} & C_{\aug} A_{\aug}^{\ells \! -3} B_{\aug} & \dotsS & C_{\aug} B_{\aug} & 0_{p\times m}
} \notag \\
\mathcal{R}_{\ells}^{\aug} & \! := \! \! \bmat{ A_{\aug}^{\ells \! -1} B_{\aug} & \ldots & A_{\aug} B_{\aug} & B_{\aug} } \notag \\
Z_{\ells}^{\aug} & := \label{eq:Z_ells_aug} \\
&  \smat{C_{\aug} A_{\aug}^{\ells} ({\mathcal{O}_{\ells}^{\aug}})^{-1} & & C_{\aug} \mathcal{R}_{\ells}^{\aug} - C_{\aug} A_{\aug}^{\ells} (\mathcal{O}_{\ells}^{\aug})^{-1} \mathcal{T}_{\ells}^{\aug} }, \notag
\end{align}
with ${\mathcal{O}_{\ells}^{\aug}}$ in \eqref{eq:obserNonsinguAug}, which is invertible by Assumption \ref{ass:obserNonsinguAug}.
We stress that the matrix $Z_{\ells}^{\aug}$ still contains the actual matrices $A^\star$, $B^\star$ and $C^\star$, so it is unknown to us.

As we did in Section~\ref{sec:problem}, we assume that the noise sequence $\Delta_{10}^{\aug}$ in~\eqref{Delta_10_aug} satisfies an energy bound, i.e., $\Delta_{10}^{\aug} \in \mathcal{D}_{\aug}$ with $\mathcal{D}_{\aug}$ defined, for some $\Theta^{\aug} = {\Theta^{\aug}}^\top \succeq 0$, as
\begin{align}
\mathcal{D}_{\aug} :=  \{ \Delta  \in  \real^{(p + p\ells + m\ells) \times (T- \ells + 1)} : \Delta \Delta^\top \preceq \Theta^{\aug} \}.  \label{eq:setD_aug}
\end{align}
We partition $\Theta^{\aug}$ as
\begin{equation}
\label{Theta_partitioned_aug}
\Theta^{\aug} =: \bmat{\Theta_{11}^{\aug} & \Theta_{12}^{\aug} \\ {\Theta_{12}^{\aug}}^\top & \Theta_{22}^{\aug}}
\end{equation}
with $\Theta_{11}^{\aug} = {\Theta_{11}^{\aug}}^\top \in \real^{p \times p}$ and $\Theta_{22}^{\aug} = {\Theta_{22}^{\aug}}^\top \in \real^{(p \ells + m \ells) \times(p \ells + m \ells)}$.
Following Remark~\ref{rem:conversion_to_energy_bound}, we present a way to obtain $\Theta^{\aug}$ in the next remark.

\begin{remark}\label{rem:conversion_to_energy_bound_aug}
Suppose we know that for some $\bar{d}^y \ge 0$ and $\bar{d}^u \ge 0$,
\begin{align}
\| d^y \|_{\mathcal{L}_\infty}
\le \bar{d}^y
\text{ and }
\| d^u \|_{\mathcal{L}_\infty}
\le \bar{d}^u. \label{Linf_norm_dy_du}
\end{align}
For $\Delta_{10}^{\aug}$ in~\eqref{Delta_10_aug}, we have $d^{y_{\aug}} = d^y +d^{y_a}$, see \eqref{augmentSYS_data:output}, with $d^{y_a}$ as in~\eqref{aug-measured-g-3}.
If we can find $\bar{d}^{y_a}$ such that
\begin{align}
\label{bound_dya}
\| d^{y_a} \|_{\mathcal{L}_\infty} \le \bar{d}^{y_a},
\end{align}
the very same steps in Remark~\ref{rem:conversion_to_energy_bound} would allow us to take $\Theta^{\aug}$ as $(T - \ells +1) \big( (\ells +1) (\bar{d}^y+ \bar{d}^{y_a} )^2 + \ells (\bar{d}^u)^2 \big) I$.
Then, we show how to obtain \eqref{bound_dya} if we select the artificial state matrix $A_a$, which is a design parameter of ours, with $\| A_a \| < 1$.
We have
\begin{align}
\| {d}^{y_a} \|_{\mathcal{L}_{\infty}} \overset{\eqref{aug-measured-g-3}}{\le} &
\sup_{k \geq 0} \Big\{ \sum_{j=0}^{k-1} \|C_a\| \| A_a \|^{k-1-j} \|B_a\| |d^{u}(j) | \Big\}  \notag \\
\overset{\eqref{Linf_norm_dy_du}}{\leq} &  \|C_a\| \sup_{k \geq 0} \Big\{ \sum_{j=0}^{k-1} \| A_a \|^{j} \Big\} \|B_a\| \bar{d}^{u} \notag \\  
= & {\|C_a\| \|B_a\| \bar{d}^{u}}/{(1-\| A_a \|)}, \label{eq:ArtificialOutputNoiseBound}
\end{align}
and we can thus take $\bar{d}^{y_a}$ as $\|C_a\| \|B_a\| \bar{d}^{u}/(1-\| A_a \|)$.
\end{remark}

We now impose the next condition, which is the counterpart of Assumption \ref{ass:Psi0}.
\begin{assumption}\label{ass:Psi0_aug}
$\Psi_0^{\aug} {\Psi_0^{\aug}}^\top \succ \Theta_{22}^{\aug}$.
\end{assumption}

The next lemma contains some relevant observations on the assumptions we have made so far and on the observability index of the augmented system.

\begin{lemma}
\label{lemma:assumptions_implication+obs_idx}
If $\Psi_0^{\aug} {\Psi_0^{\aug}}^\top \succ \Theta_{22}^{\aug}$ as per Assumption~\ref{ass:Psi0_aug}, then $\mathcal{O}_{\ells}^{\aug}$ is nonsingular as per Assumption~\ref{ass:obserNonsinguAug} and $(A_{\aug}, C_{\aug})$ is observable with observability index equal to $\ells$.
\end{lemma}
\begin{proof}
See Appendix~\ref{app:proof_lemma_assumptions_implication+obs_idx}.
\end{proof}

We discuss the consequences of Lemma~\ref{lemma:assumptions_implication+obs_idx}.
First, the \emph{data-dependent} condition in Assumption~\ref{ass:Psi0_aug} is a sufficient condition for Assumption~\ref{ass:obserNonsinguAug}; so Assumption~\ref{ass:obserNonsinguAug} is not invoked in the subsequent results.
Second, the fact that the observability index of $(A_{\aug}, C_{\aug})$ is equal to $\ells$ allows us to fall back to the setting of Sections~\ref{sec:reform_setC}--\ref{sec:from_aux_sys_to_actual_sys} for the \emph{augmented} system.

Since the matrix $Z_{\ells}^{\aug}$ in \eqref{eq:Z_ells_aug} is unknown, we introduce the set of matrices consistent with data as in~\eqref{exp_data_mat_aug} and with the noise bound in~\eqref{eq:setD_aug}, namely,
\begin{align}\label{setC_aug}
& \mathcal{C}_{\aug}  :=  \Big\{  Z \in \real^{p \times (p \ells + m \ells)} \colon \Psi_1^{\aug} = (\mathbf{F}_\ells + \mathbf{L}_\ells Z ) \Psi_0^{\aug} \notag \\
& \hspace*{8mm} + \mathbf{B}_\ells U_1^{\aug} + \mathbf{L}_\ells [I_p \ -Z]  \Delta, \ \Delta \in \mathcal{D}_{\aug} \Big\}.
\end{align}
It is worth noting that 
$
    Z_{\ells}^{\aug} \in \mathcal{C}_{\aug}
$
because $\Delta_{10}^{\aug} \in \mathcal{D}_{\aug}$.
By the same steps as in~\eqref{eq:setC}, the set $\mathcal{C}_{\aug}$ can be equivalently rewritten as
\begin{subequations}\label{eq:setC_aug}
\begin{align}
\mathcal{C}_{\aug} \! = \! \big\{ Z & \colon Z \mathscr{A}_{\aug} Z^{\top} \! + \! Z {\mathscr{B}_{\aug}}^{\top} \! + \! \mathscr{B}_{\aug} Z^{\top} \! + \! \mathscr{C}_{\aug} \preceq 0 \big\}, \label{setC_ABC_aug}\\
& \mathscr{A}_{\aug} := \Psi_0^{\aug} {\Psi_0^{\aug}}^\top - \Theta_{22}^{\aug}, \label{setC:A_aug} \\
& \mathscr{B}_{\aug} :=  -  \mathbf{L}_{\ells}^\top \Psi_1^{\aug} {\Psi_0^{\aug}}^{\top} + \Theta_{12}^{\aug}, \label{setC:B_aug}\\
& \mathscr{C}_{\aug} := \mathbf{L}_{\ells}^\top \Psi_1^{\aug} {\Psi_1^{\aug}}^\top \mathbf{L}_{\ells} - \Theta_{11}^{\aug}. \label{setC:C_aug}
\end{align}
\end{subequations}

We would like to find $\mathbf{K}$ such that $\mathbf{F}_{\ells} + \mathbf{L}_{\ells} Z + \mathbf{B}_{\ells} \mathbf{K}$ is Schur for all $Z\in \mathcal{C}_{\aug}$, i.e., to solve the robust stabilization problem
\begin{subequations}
\label{rob_contr_probl_sys_xi_aug}
\begin{align}
& \text{find} & & \mathbf{K}, P = P^\top \succ 0 \\
& \text{s.t.} & & (\mathbf{F}_{\ells} + \mathbf{L}_{\ells} Z + \mathbf{B}_{\ells} \mathbf{K}) P (\mathbf{F}_{\ells} + \mathbf{L}_{\ells} Z + \mathbf{B}_{\ells} \mathbf{K})^\top \notag \\ 
& & & - P \prec 0 \qquad \forall Z \in \mathcal{C}_{\aug}. \label{rob_contr_probl_sys_xi:ineq:aug}
\end{align}
\end{subequations}
The problem \eqref{rob_contr_probl_sys_xi_aug} is the counterpart of \eqref{rob_contr_probl_sys_xi} and is solved by the next lemma.

\begin{lemma}
\label{lemma:Petersen_aug}
Under Assumption \ref{ass:Psi0_aug}, feasibility of~\eqref{rob_contr_probl_sys_xi_aug} is equivalent to feasibility of
\begin{subequations}
\label{rob_contr_probl_LMI_aug}
\begin{align}
& \text{find} & & \mathbf{Y}, P = P^\top \succ 0 \label{rob_contr_probl_LMI:find_aug}\\
& \text{s.t.} & & 
\left[\begin{matrix}
-P + \mathbf{L}_{\ells} \mathscr{C}_{\aug} \mathbf{L}_{\ells}^\top & \mathbf{F}_{\ells} P + \mathbf{B}_{\ells} \mathbf{Y} & \mathbf{L}_{\ells} \mathscr{B}_{\aug}\\
P \mathbf{F}_{\ells}^\top + \mathbf{Y}^\top \mathbf{B}_{\ells}^\top & - P & -P \\
\mathscr{B}_{\aug}^\top \mathbf{L}_{\ells}^\top & -P & -\mathscr{A}_{\aug}
\end{matrix}\right] \notag \\
& & & \prec 0.
\end{align}
\end{subequations}
If \eqref{rob_contr_probl_LMI_aug} is feasible, a $\mathbf{K}$ satisfying \eqref{rob_contr_probl_sys_xi_aug} is $\mathbf{K} = \mathbf{Y}P^{-1}$.
\end{lemma}

\begin{proof}
The proof is the same as the proof of Lemma~\ref{lemma:Petersen} with $\mathscr{A}_{\aug}$, $\mathscr{B}_{\aug}$, $\mathscr{C}_{\aug}$ substituting $\mathscr{A}$, $\mathscr{B}$, $\mathscr{C}$.
\end{proof}

This leads to the main result of this section.

\begin{theorem}\label{thm:pl_neq_n_}
With collected data $\{u^{\tu{m}}(k), y^{\tu{m}}_{\aug}(k)\}_{k=0}^T$ and under Assumption~\ref{ass:Psi0_aug}, suppose \eqref{rob_contr_probl_LMI_aug} is feasible and $\mathbf{K}$ is a controller returned from~\eqref{rob_contr_probl_LMI_aug}. 
Then, $(x, x_a, \chi) = 0$ is globally asymptotically stable for the feedback interconnection of the unknown plant
\begin{align}\label{eq:plant_aug}
x^+ = \As x + \Bs u, \ y = \Cs x 
\end{align}
and the controller
\begin{subequations}\label{eq:controller_aug}
\begin{align}
\bmat{x_a^+ \\ \chi^+} &= \bmat{ A_a & B_a \mathbf{K} \\ \mathbf{L}_{\ells} C_a & \mathbf{F}_{\ells} + \mathbf{B}_{\ells} \mathbf{K} } \bmat{x_a \\ \chi} + \bmat{ 0 \\ \mathbf{L}_{\ells} } y   \\
u &= \mathbf{K} \chi.
\end{align} 
\end{subequations}
\end{theorem}
\begin{proof}
By Lemma \ref{lemma:Petersen_aug}, the same reasoning as in Theorem~\ref{thm:pl_eq_n} concludes that $(x, x_a, \chi) = 0$ is globally asymptotically stable for the feedback interconnection of the augmented system
\begin{align*}
\bmat{x^{+} \\ x_a^{+}} & = \bmat{\As & 0 \\ 0 & A_a} \bmat{x \\ x_a} + \bmat{\Bs \\ B_a} u  \\ 
y_{\aug} & = \Cs x + C_a x_a  
\end{align*}
and the controller
\begin{align*}
\chi^+ &  = \mathbf{F}_{\ells} \chi + \mathbf{L}_{\ells} y_{\aug} + \mathbf{B}_{\ells} u  \\
u & = \mathbf{K} \chi.
\end{align*}
This interconnection is equivalent to~\eqref{eq:plant_aug}, \eqref{eq:controller_aug}.
\end{proof}

%%%%%%%%%%%%%%%%%%%%%%%%%%%%%%%%%%%%%%%%%%%%%%%%%%%%%%%%%%%%%%%%%%%%%%%%%%%%%%%%%%%
%%%%%%%%%%%%%%%%%%%%%%%%%%%%%%%%%%%%%%%%%%%%%%%%%%%%%%%%%%%%%%%%%%%%%%%%%%%%%%%%%%%
%%%%%%%%%%%%%%%%%%%%%%%%%%%%%%%%%%%%%%%%%%%%%%%%%%%%%%%%%%%%%%%%%%%%%%%%%%%%%%%%%%%
%%%%%%%%%%%%%%%%%%%%%%%%%%%%%%%%%%%%%%%%%%%%%%%%%%%%%%%%%%%%%%%%%%%%%%%%%%%%%%%%%%%

\section{Numerical examples}
\label{sec:NumericalExamples}

%%%%%%%%%%%%%%%%%%%%%%%%%%%%%%%%%%%%%%%%%%%%%%%%%%%%%%%%%%%%%%%%%%%%%%%%%%%%%%%%%%%
%%%%%%%%%%%%%%%%%%%%%%%%%%%%%%%%%%%%%%%%%%%%%%%%%%%%%%%%%%%%%%%%%%%%%%%%%%%%%%%%%%%
\subsection{Example for the case $p \ells = n$}
Consider the unstable batch reactor process in \cite[\S 2.6]{Green1995}. Assuming that the control input is piecewise constant in the sample period of 0.2~s, we discretize the continuous-time system by the zero-order-hold method and use
\begin{equation*}
    \left[\begin{array}{c|c} A^{\star} & B^{\star} \\ \hline C^{\star} & 0 \end{array}\right] =
    \left[\begin{smallarray}{cccc|cc} 
    1.427 & 0.039 & 0.854 & -0.622 & 0.034 & -0.305 \\ 
    -0.096 & 0.455 & -0.034 & 0.109 & 0.787 & 0.008 \\ 
    0.115 & 0.538 & 0.384 & 0.529 & 0.571 & -0.380 \\ 
    -0.012 & 0.537 & 0.122 & 0.777 & 0.570 & -0.050 \\ 
    \hline 
    1 & 0 & 1 & -1 & 0 & 0 \\ 
    0 & 1 & 0 & 0 & 0 & 0 \end{smallarray}\right]
\end{equation*}
where $n=4$, $m = 2$, $p=2$ and $\ells = 2$. 
The values of $(A^{\star}, B^{\star}, C^{\star})$ are used only for data generation, but not for controller design. According to \eqref{Zstar},
\[ 
Z_\ells \! = \! \smat{ -0.374 & -0.714 & 1.870 & 1.870 & -1.311 & 0.317 & 0.035 & -0.634 \\
-0.016 & -0.289 & -0.037 & 1.173 & -0.524 & 0.007 & 0.787 & 0.008 }.
\]

Let $\omega \thicksim U[a, b]$ denote a random variable $\omega$ uniformly distributed in the interval $[a, b]$. 
We apply inputs $u^{\tu{m}}_1$, $u^{\tu{m}}_2 \thicksim U[-20,20]$, input noise $d^{u}_1$, $d^{u}_2 \thicksim U[-0.01, 0.01]$ and output noise $d^{y}_1$, $d^{y}_2 \thicksim U[-0.01, 0.01]$ to collect data.
Since the system is unstable, 10 independent experiments are performed and only the first 4 data samples are collected from each experiment. 
We store the data points into $\Psi_1 \in \real^{8 \times 20}$, $\Psi_0$, $U_1$ as in~\eqref{Psi1_Psi0_U0_Delta10}. 
Bearing in mind \eqref{setD}, the upper bounds on the noise (i.e., 0.01 for input noise, and 0.01 for output noise), and the size of the noise data matrix $\Delta_{10}$ (i.e., $10 \times 20$), we have $\Theta = 0.02I_{10}$, see Remark \ref{rem:conversion_to_energy_bound}. 
By Lemma~\ref{lemma:Petersen}, solving \eqref{rob_contr_probl_LMI} via CVX \cite{cvx,gb08} yields the gain matrix
\[ \mathbf{K}  = \tiny{\left[\begin{smallmatrix} 0.204 & 3.541 & -0.927 & -9.371 & 6.444 & -0.225 & -0.992 & 0.249 \\ -1.076 & -3.607 & 4.550 & 9.645 & -6.590 & 0.950 & 0.494 & -1.235 \end{smallmatrix}\right]}.\]
The eigenvalues of the closed-loop system \eqref{sys_xstar_chi_CL} are
\[\begin{aligned}
 & \eig \left( 
\smat{ \As & \Bs \mathbf{K} \\ \mathbf{L}_\ells \Cs & \mathbf{F}_\ells + \mathbf{B}_\ells \mathbf{K} } 
\right) 
\! = \! \{   -0.433, \ -0.401, \ 0.021,  \\ 
& \hspace*{20pt}  0.459, \ 0.102 \pm i0.411, \ 0.483 \pm i0.301 , \ 0, \ 0, \ 0, \ 0 \}.
\end{aligned}\]
These eigenvalues are all inside the unit disk and the closed-loop system \eqref{sys_xstar_chi_CL} is asymptotically stable.

%%%%%%%%%%%%%%%%%%%%%%%%%%%%%%%%%%%%%%%%%%%%%%%%%%%%%%%%%%%%%%%%%%%%%%%%%%%%%%%%%%%
%%%%%%%%%%%%%%%%%%%%%%%%%%%%%%%%%%%%%%%%%%%%%%%%%%%%%%%%%%%%%%%%%%%%%%%%%%%%%%%%%%%
\subsection{Example for the case $p \ells > n$}
Consider the unstable discrete-time system given by
\begin{equation*}
\left[\begin{array}{c|c} A^{\star} & B^{\star} \\ \hline C^{\star} & 0	\end{array}\right] = \left[\begin{smallarray}{ccc|cc} 0 & 1 & 0 & 1 & 0 \\ -1 & 0 & 0 & 0 & 1 \\ 0 & 0 & 1 & 1 & 0 \\ \hline 1 & 0 & 1 & 0 & 0 \\ 0 & 1 & 1 & 0 & 0 \end{smallarray}\right].
\end{equation*}
In this case, $n = 3$, $m = 2$, $p = 2$, $\ells = 2$, and $p \ells=4>3=n$, thus, we need to use the results in Section~\ref{sec:results_pells_neq_n}, namely, Lemma \ref{lemma:Petersen_aug} and Theorem \ref{thm:pl_neq_n_}. 
We choose an artificial system with $n_a=4-3=1$ and 
\begin{equation*}
    \left[\begin{array}{c|c} A_a & B_a \\ \hline C_a & 0 \end{array}\right] = \left[\begin{smallarray}{c|cc} 0 & 1 & 1 \\ \hline 1 & 0 & 0 \\ 1 & 0 & 0 \end{smallarray}\right].
\end{equation*}
The augmented system is then
\begin{equation*}
    \left[\begin{array}{c|c} A_{\aug} & B_{\aug} \\ \hline C_{\aug} & 0 \end{array}\right] = \left[\begin{smallarray}{cccc|cc} 0 & 1 & 0 & 0 & 1 & 0 \\ -1 & 0 & 0 & 0 & 0 & 1 \\ 0 & 0 & 1 & 0 & 1 & 0 \\ 0 & 0 & 0 & 0 & 1 & 1 \\ \hline 1 & 0 & 1 & 1 & 0 & 0 \\ 0 & 1 & 1 & 1 & 0 & 0 \end{smallarray}\right].
\end{equation*}
According to \eqref{eq:Z_ells_aug},
$ Z_{\ells}^{\aug} = \smat{0 &  0 & 0 & 1 & -1 & -1 & 3 & 1 \\ 1 & -1 & 0 & 1 & -2 & -2 & 2 & 2}. $

We apply inputs $u^{\tu{m}}_1$, $u^{\tu{m}}_2 \thicksim U[-2,2]$, input noise $d^{u}_1$, $d^{u}_2 \thicksim U[-0.01, 0.01]$ and output noise $d^{y}_1$, $d^{y}_2 \thicksim U[-0.01, 0.01]$ to collect 32 data samples. The data collection scheme is the one in Figure \ref{fig:prlel_conect_data}a. 
We store the data points into $\Psi_1^{\aug} \in \real^{8 \times 30}$, $\Psi_0^{\aug}$, $U_1^{\aug}$ as in~\eqref{Psi1_Psi0_U0_Delta10_aug}.
By \eqref{eq:ArtificialOutputNoiseBound}, $\| {d}^{y_a} \|_{\mathcal{L}_{\infty}} = 0.02$. 
Bearing in mind \eqref{eq:setD_aug}, the upper bounds on the noise (i.e., 0.01 for input noise, 0.01 for output noise, and 0.02 for the artificial output noise), and the size of the noise data matrix $\Delta_{10}^{\aug}$ (i.e., $10 \times 30$), we have $\Theta^{\aug} = 0.174I_{10}$, see Remark \ref{rem:conversion_to_energy_bound_aug}.
It can be verified that data satisfy Assumption \ref{ass:Psi0_aug}.
Solving \eqref{rob_contr_probl_LMI_aug} in Lemma \ref{lemma:Petersen_aug} yields the gain matrix
\[ \mathbf{K} = \left[\begin{smallmatrix}  -0.290  &  0.299  &  -0.153  &  -0.343  &  0.822  &  0.792  &  -0.557  &  -0.262 \\  -0.013  &  0.013  &  -0.114  &  0.119 & 0.018  &  0.017  &  0.163  &  -0.009 \end{smallmatrix}\right]. \]

The eigenvalues of the closed-loop system \eqref{eq:plant_aug}, \eqref{eq:controller_aug} are
\[\begin{aligned}
    \eig\left( \smat{ \As & 0 & \Bs \mathbf{K} \\ 0 & A_a & B_a \mathbf{K} \\  \mathbf{L}_\ell \Cs & \mathbf{L}_\ell C_a & \mathbf{F}_\ell + \mathbf{B}_\ell \mathbf{K} } \right) = \{  0.0017, \ -0.1743 & ,  \\ 
    -0.1418 \pm i0.1900, \ 0.3982 \pm i0.1183 & ,  \\
    0.0468 \pm i0.8469, \ 0, \ 0, \ 0, \ 0 & \}.
\end{aligned}\]
These eigenvalues are all inside the unit disk and the closed-loop system is asymptotically stable.

\begin{figure}[t]
\begin{center}
    \includegraphics[width=\columnwidth]{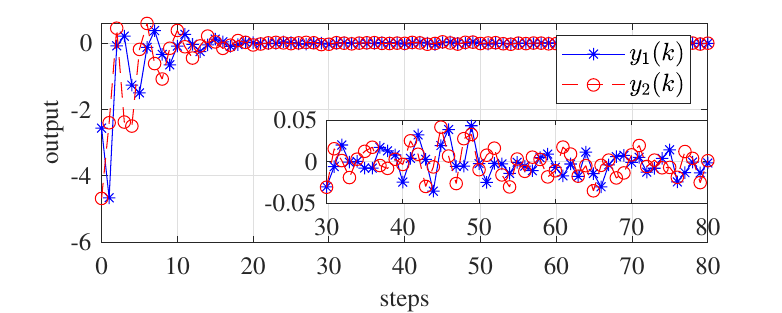}
    \vspace{-0.5cm}
    \caption{Controller execution in the presence of measurement noise.}\label{fig:case2sim}
\end{center}	
\end{figure}

Finally, we examine the effect of measurement noise \emph{during} the execution of the controller. 
In this case, the closed-loop system \eqref{eq:plant_aug}, \eqref{eq:controller_aug} modifies to
\begin{subequations}
\begin{align*}
x^+ = \As x + \Bs ( u^{\tu{m}} - d^{u} ), \quad y^{\tu{m}} = \Cs x + d^{y}
\end{align*}
and 
\begin{align*}
\bmat{x_a^+ \\ \chi^+} &= \bmat{ A_a & B_a \mathbf{K} \\ \mathbf{L}_{\ells} C_a & \mathbf{F}_{\ells} + \mathbf{B}_{\ells} \mathbf{K} } \bmat{x_a \\ \chi} + \bmat{ 0 \\ \mathbf{L}_{\ells} } y^{\tu{m}}   \\
u^{\tu{m}} &= \mathbf{K} \chi.
\end{align*} 
\end{subequations}
We select $x(0) = (-0.07, -2.19, -2.49)$, $x_a(0)=0$, $\chi(0)=0$ and the input noise as $d^{u}_1$, $d^{u}_2 \thicksim U[-0.01, 0.01]$ and the output noise as $d^{y}_1$, $d^{y}_2 \thicksim U[-0.01, 0.01]$.
As expected, by bounded-input-bounded-output stability properties \cite[\S 6.10, Thm.~10.17]{LinearSystems2005}, the outputs remain in a neighborhood of the origin as shown in Figure~\ref{fig:case2sim}.

%%%%%%%%%%%%%%%%%%%%%%%%%%%%%%%%%%%%%%%%%%%%%%%%%%%%%%%%%%%%%%%%%%%%%%%%%%%%%%%%%%%
%%%%%%%%%%%%%%%%%%%%%%%%%%%%%%%%%%%%%%%%%%%%%%%%%%%%%%%%%%%%%%%%%%%%%%%%%%%%%%%%%%%
%%%%%%%%%%%%%%%%%%%%%%%%%%%%%%%%%%%%%%%%%%%%%%%%%%%%%%%%%%%%%%%%%%%%%%%%%%%%%%%%%%%
%%%%%%%%%%%%%%%%%%%%%%%%%%%%%%%%%%%%%%%%%%%%%%%%%%%%%%%%%%%%%%%%%%%%%%%%%%%%%%%%%%%
\section{Conclusion}
\label{sec:Conclusion}
In this paper, we have synthesized a dynamic controller directly from input-output data in the presence of measurement noise.
To deal with input-output data, an auxiliary state-space representation is utilized.
We distinguish between two cases: $p \ells = n$ and $p \ells \neq n$.
In the former case, we capitalize on the structure of the auxiliary representation, and design a controller that stabilizes all systems consistent with data. 
In the latter case, we introduce an augmented system comprising the original system. The controller for the augmented system is designed using the same procedure as in the first case, and subsequently, the controller for the original system is derived from the controller for the augmented system.
The results are validated by comprehensive numerical examples. 
This work specifically focuses on the stabilization problem because it serves as a prototypical control problem. 
Future research includes the extension of our approach to address other control problems such as $\mathcal{H}_2$, $\mathcal{H}_{\infty}$ and tracking control using input-output data. 
Additionally, we anticipate potential applications and extensions of the proposed methodology to nonlinear systems.

%%%%%%%%%%%%%%%%%%%%%%%%%%%%%%%%%%%%%%%%%%%%%%%%%%%%%%%%%%%%%%%%%%%%%%%%%%%%%%%%%%%
%%%%%%%%%%%%%%%%%%%%%%%%%%%%%%%%%%%%%%%%%%%%%%%%%%%%%%%%%%%%%%%%%%%%%%%%%%%%%%%%%%%
%%%%%%%%%%%%%%%%%%%%%%%%%%%%%%%%%%%%%%%%%%%%%%%%%%%%%%%%%%%%%%%%%%%%%%%%%%%%%%%%%%%
%%%%%%%%%%%%%%%%%%%%%%%%%%%%%%%%%%%%%%%%%%%%%%%%%%%%%%%%%%%%%%%%%%%%%%%%%%%%%%%%%%%
\appendix

%%%%%%%%%%%%%%%%%%%%%%%%%%%%%%%%%%%%%%%%%%%%%%%%%%%%%%%%%%%%%%%%%%%%%%%%%%%%%%%%%%%
%%%%%%%%%%%%%%%%%%%%%%%%%%%%%%%%%%%%%%%%%%%%%%%%%%%%%%%%%%%%%%%%%%%%%%%%%%%%%%%%%%%
\section{Ancillary results for the linear systems \eqref{sys_x} and \eqref{sys_xi}}
\label{app:results_lin_sys}

In this section we present some straightforward results for the linear system in~\eqref{sys_x} and the auxiliary linear system in~\eqref{sys_xi}.
Relevant relations between input, state, output of solutions to~\eqref{sys_x} are characterized next.

\begin{claim}
\label{claim:sol_sys_x}
If $(A,C)$ is observable, for each $\hat{x}\in\real^n$ and sequence $\{u(k)\}_{k = 0}^\infty$, let $x(\cdot)$ be the solution to~\eqref{sys_x} with initial condition $x(0) = \hat{x}$ and input $\{u(k)\}_{k = 0}^\infty$, and $y(\cdot)= C x(\cdot)$ be the corresponding output response.
Then, for each $k \ge \ell$,
\begin{subequations}
\begin{align}
& \smat{
y(k-\ell)\\
\vdotsS \\
y(k-1)
}
=
\mathcal{O}_\ell x(k-\ell)
+
\mathcal{T}_\ell 
\smat{
u(k-\ell)\\
\vdotsS \\
u(k-1)
}, \label{sol_sys_x:o_s_i} \\
& x(k) =
A^\ell \mathcal{O}_\ell^{\tu{L}}
\smat{
y(k-\ell)\\
\vdotsS \\
y(k-1)
}
+
\big(
\mathcal{R}_\ell -A^\ell \mathcal{O}_\ell^{\tu{L}} \mathcal{T}_\ell
\big)
\smat{
u(k-\ell)\\
\vdotsS \\
u(k-1)
}, \label{sol_sys_x:s_o_i} \\
&  y(k) \! = \! C A^\ell \mathcal{O}_\ell^{\tu{L}}
\smat{
y(k-\ell)\\
\vdotsS \\
y(k-1)
}
\! + \!
\big(
C \mathcal{R}_\ell \! - \! CA^\ell \mathcal{O}_\ell^{\tu{L}} \mathcal{T}_\ell
\big)
\smat{
u(k-\ell)\\
\vdotsS \\
u(k-1)
}. \label{sol_sys_x:o_o_i}
\end{align}
\end{subequations}
\end{claim}

\begin{proof}
By definition, the solution $x(\cdot)$ to~\eqref{sys_x} with initial condition $x(0) = \hat{x}$ and input $\{u(k)\}_{k = 0}^\infty$ satisfies, for each $k \ge \ell$,
\begin{align*}
\smat{
y(k-\ell)\\
\vdotsS \\
y(k-1)
}
& =
\smat{
C\\
CA\\
\vdotsS\\
CA^{\ell-1}
} x(k-\ell) \\
& \qquad
+
\smat{
0 & 0 & \dots & 0 & 0\\
CB                     & 0 & \dots & 0 & 0\\
CAB                    & CB                     & \dots & 0 & 0\\
\vdotsS                & \vdotsS                 & \rotatebox{135}{\tiny\dots} &  \vdotsS & \vdotsS\\
CA^{\ell -2} B         & CA^{\ell -3} B         & \dots & CB &0
} \smat{
u(k-\ell)\\
\vdotsS \\
u(k-1)
}\\
& \overset{\eqref{O_ell_T_ell_R_ell}}{=}
\mathcal{O}_\ell x(k-\ell)
+
\mathcal{T}_\ell 
\smat{
u(k-\ell)\\
\vdotsS \\
u(k-1)
},
\end{align*}
i.e., \eqref{sol_sys_x:o_s_i} holds.
\eqref{sol_sys_x:o_s_i} implies that
\begin{align}
x(k-\ell) = 
\mathcal{O}_\ell^{\tu{L}} 
\Bigg(
\smat{
y(k-\ell)\\
\vdotsS \\
y(k-1)
}
-
\mathcal{T}_\ell
\smat{
u(k-\ell)\\
\vdotsS \\
u(k-1)
}
\Bigg)
\label{sol_sys_x:s_o_i_ell_times_before}
\end{align}
since $(A,C)$ is observable and $\mathcal{O}_\ell$ possesses a left inverse. 
By definition, the solution $x(\cdot)$ to~\eqref{sys_x} with initial condition $x(0) = \hat{x}$ and input $\{u(k)\}_{k = 0}^\infty$ satisfies, for each $k \ge \ell$,
\begin{align*}
x(k) & = A^\ell x(k-\ell) + A^{\ell -1} B u(k-\ell) + \dots + B u(k-1),
\end{align*}
as can be proven by mathematical induction; thus,
\begin{align*}
& x(k) \overset{\eqref{O_ell_T_ell_R_ell}}{=} A^\ell x(k-\ell) + \mathcal{R}_\ell 
\smat{
u(k-\ell)\\
\vdotsS \\
u(k-1)
} \\
& \overset{\eqref{sol_sys_x:s_o_i_ell_times_before}}{=}
A^\ell 
\mathcal{O}_\ell^{\tu{L}} 
\Bigg(
\smat{
y(k-\ell)\\
\vdotsS \\
y(k-1)
}
-
\mathcal{T}_\ell
\smat{
u(k-\ell)\\
\vdotsS \\
u(k-1)
}
\Bigg)
+ \mathcal{R}_\ell 
\smat{
u(k-\ell)\\
\vdotsS \\
u(k-1)
},
\end{align*}
i.e., \eqref{sol_sys_x:s_o_i} holds.
\eqref{sol_sys_x:o_o_i} follows from $y(\cdot) = C x(\cdot)$.
\end{proof}

In Lemma~\ref{lemma:imH=R(A,B)}, we need to determine the reachability subspace of $(\mathbf{A}_\ell,\mathbf{B}_\ell)$ in~\eqref{sys_xi}. 
Hence, we provide an analytic expression for the forced response of~\eqref{sys_xi} next.

\begin{claim}
\label{claim:forced_response_xi}
Let $\xi(\cdot)$ be the solution to~\eqref{sys_xi} with initial condition equal to $0$ and input sequence $\{ v(k) \}_{k = 0}^\infty$, and introduce the fictitious $v(-1-\ell)=  v(-\ell)= \dots = v(-1) = 0$. 
The solution $\xi(\cdot)$ satisfies, for each $k \ge  0$,
\begin{align}
\label{all_steps:xi_k}
\xi(k) = 
\left[
\begin{array}{c}
\begin{smallmatrix}
\xi_1(k) \\
\xi_2(k) \\
\vdotsS\\
\xi_{\ell-1}(k)\\
\xi_\ell(k) \\
\end{smallmatrix}\\
\hline
\begin{smallmatrix}
\xi_{\ell+1}(k)\\
\xi_{\ell+2}(k)\\
\vdotsS\\
\xi_{\ell+\ell-1}(k)\\
\xi_{\ell+\ell}(k)\\
\end{smallmatrix}
\end{array}
\right]
=
\left[
\begin{array}{c}
\begin{smallmatrix}
\sum_{j=0}^{k - 1 - \ell} C A^{k - 1 - \ell - j} B v(j) \\
\sum_{j=0}^{k - \ell} C A^{k - \ell - j} B v(j) \\
\vdotsS\\
\sum_{j=0}^{k - 3} C A^{k - 3 - j} B v(j) \\
\sum_{j=0}^{k - 2} C A^{k - 2 - j} B v(j) \\
\end{smallmatrix}\\
\hline
\begin{smallmatrix}
v(k-\ell)\\
v(k-\ell+1)\\
\vdotsS\\
v(k-2)\\
v(k-1)\\
\end{smallmatrix}
\end{array}
\right]
\end{align}
or, equivalently,
\begin{align}
\label{all_steps:xi_k_alt}
\xi(k) 
=
\left[
\begin{array}{c}
\mathcal{O}_\ell \sum_{j=0}^{k-1-\ell} A^{k-1-\ell-j} B v(j) + \mathcal{T}_\ell \smat{v(k-\ell)\\
\vdotsS\\
v(k-1)\\}
\\
\hline
\begin{smallmatrix}
v(k-\ell)\\
\vdotsS\\
v(k-1)\\
\end{smallmatrix}
\end{array}
\right].
\end{align}
\end{claim}

\begin{proof}
We prove the first part of the statement, i.e., the validity of \eqref{all_steps:xi_k}, by mathematical induction.
The base case corresponds to showing \eqref{all_steps:xi_k} for $k = 0$; namely, we would like to show
\begin{align}
\label{xi_1}
\xi(0) =
\left[
\begin{array}{c}
\begin{smallmatrix}
\sum_{j=0}^{-1 - \ell} C A^{-1 - \ell - j} B v(j) \\
\vdotsS\\
\sum_{j=0}^{- 2} C A^{- 2 - j} B v(j) \\
\end{smallmatrix}\\
\hline
\begin{smallmatrix}
v(-\ell)\\
\vdotsS\\
v(-1)\\
\end{smallmatrix}
\end{array}
\right] =
\left[
\begin{array}{c}
\begin{smallmatrix}
0\\
\vdotsS\\
0 \\
\end{smallmatrix}\\
\hline
\begin{smallmatrix}
0\\
\vdotsS\\
0\\
\end{smallmatrix}
\end{array}
\right].
\end{align}
This is true since the sums are empty and the fictitious $v(-1-\ell)=v(-\ell)= \dots = v(-1)$ are zero.
The induction step corresponds to showing that if \eqref{all_steps:xi_k} holds for $k \ge 0$, then \eqref{all_steps:xi_k} holds also for $k+1$; namely, we would like to show
\begin{align}
\label{xi_k+1}
\xi(k+1) =
\left[
\begin{array}{c}
\begin{smallmatrix}
\sum_{j=0}^{k - \ell} C A^{k - \ell - j} B v(j) \\
\sum_{j=0}^{k + 1 - \ell} C A^{k +1 - \ell - j} B v(j) \\
\vdotsS\\
\sum_{j=0}^{k - 2} C A^{k - 2 - j} B v(j) \\
\sum_{j=0}^{k - 1} C A^{k - 1 - j} B v(j) \\
\end{smallmatrix}\\
\hline
\begin{smallmatrix}
v(k-\ell+1)\\
v(k-\ell+2)\\
\vdotsS\\
v(k-1)\\
v(k)\\
\end{smallmatrix}
\end{array}
\right].
\end{align}
To do so, we compute $\xi(k+1)$ as
\begin{align*}
& \xi(k+1) \overset{\eqref{sys_xi}}{=} \mathbf{A}_\ell \xi(k) + \mathbf{B}_\ell v(k)\\
& \overset{\eqref{sys_xi:A0_L_B_Z},\eqref{sys_xi:Aell},\eqref{all_steps:xi_k}}{=} 
\left[
\begin{array}{c}
\begin{smallmatrix}
\sum_{j=0}^{k - \ell} C A^{k - \ell - j} B v(j) \\
\vdotsS\\
\sum_{j=0}^{k - 2} C A^{k - 2 - j} B v(j) \\
\smat{
C A^\ell \mathcal{O}_\ell^{\tu{L}} &\big|& C \mathcal{R}_\ell - C A^\ell \mathcal{O}_\ell^{\tu{L}} \mathcal{T}_\ell
}
\xi(k) 
\begin{matrix}
\rule{0pt}{12pt}
\end{matrix}
\\
\end{smallmatrix}\\
\hline
\begin{smallmatrix}
v(k-\ell+1)\\
\vdotsS\\
v(k-1)\\
v(k)\\
\end{smallmatrix}
\end{array}
\right].
\end{align*}
We focus on block row $\ell$. Its subterms $\left[\begin{array}{c|c}0 & C \mathcal{R}_\ell \end{array}\right] \xi(k)$ and $\left[\begin{array}{c|c}0 & \mathcal{T}_\ell \end{array}\right] \xi(k)$ are
\begin{align*}
& C \mathcal{R}_\ell \smat{v(k-\ell) \\ \vdotsS \\ v(k-1)}  
\overset{\eqref{R_ell}}{=} C \bmat{ A^{\ell-1} B & \ldots & A B & B } \smat{v(k-\ell) \\ \vdotsS \\ v(k-1)} \\
& = C A^{\ell-1} B v(k - \ell) + \dots + C B v(k-1) , \\
& \mathcal{T}_\ell \smat{v(k-\ell) \\ \vdotsS \\ v(k-1)} 
\overset{\eqref{T_ell}}{=} 
\smat{
0\\
CB \\
\vdotsS\\
C A^{\ell-2} B
} v(k-\ell) +
\smat{
0\\
0\\
CB\\
\vdotsS\\
C A^{\ell - 3} B
} v(k - \ell + 1) \\
& + \dots +
\smat{0\\ \vdotsS \\ 0\\ CB} v(k-2) +
\smat{0\\ \vdotsS\\ 0} v(k-1)\\
& = \smat{
0\\
CB v(k-\ell)\\
C A B v(k-\ell) + C B v(k-\ell+1)\\
\vdotsS\\
C A^{\ell-2} B v(k-\ell) + C A^{\ell-3} B v(k-\ell+1) + \dots + C B v(k-2)}.
\end{align*}
In block row $\ell$, these two subterms lead to \eqref{forced_response_xi:long_eq}, displayed over two columns.
\begin{figure*}
\begin{equation}
\label{forced_response_xi:long_eq}
\begin{aligned}
&\xi_\ell(k+1) \overset{\eqref{all_steps:xi_k}}{=} C A^\ell \mathcal{O}_\ell^{\tu{L}} 
\smat{
\sum_{j=0}^{k - 1 - \ell} C A^{k - 1 - \ell - j} B v(j) \\
\vdotsS\\
\sum_{j=0}^{k - 2} C A^{k - 2 - j} B v(j) \\
} -
C A^\ell \mathcal{O}_\ell^{\tu{L}} 
\smat{
0\\
CB v(k-\ell)\\
C A B v(k-\ell) + C B v(k-\ell+1)\\
\vdotsS\\
C A^{\ell-2} B v(k-\ell) + C A^{\ell-3} B v(k-\ell+1) + \dots + C B v(k-2)}\\
& \quad + C A^{\ell-1} B v(k - \ell) + \dots + C B v(k-1) \\
& =
C A^\ell \mathcal{O}_\ell^{\tu{L}}
\text{$\smat{
C A^{k-1-\ell} B v(0) & + & \dots & + & C B v(k-1-\ell) &  & & & & & & & \hfill  \\
C A^{k-\ell} B v(0) & + & \dots & + & C A B v(k-1-\ell) & + & C B v(k-\ell)& & & & & && - &  C B v(k-\ell) \hfill \\
&  & \vdots  &  & & & & & & & & & & & \\
C A^{k-3} B v(0) & + & \dots & + & C A^{\ell - 2} B v(k-1-\ell) & + & C A^{\ell-3} B v(k-\ell) & + & \dots  & + & C B v(k-3) & & & - & C A^{\ell-3} B v(k-\ell)  - \dots  -  C B v(k-3) \hfill \\
C A^{k-2} B v(0) & + & \dots & + & C A^{\ell-1} B v(k-1-\ell) & + & C A^{\ell-2} B v(k-\ell) & + & \dots  & + & C A B v(k-3) & + & C B v(k-2)& - & C A^{\ell-2} B v(k-\ell)  -  \dots  - C B v(k-2) \hfill \\
}$
}
\\
& \quad + C A^{\ell-1} B v(k - \ell) + \dots + C B v(k-1) \\
& =
C A^\ell \mathcal{O}_\ell^{\tu{L}}
\smat{
C A^{k-1-\ell} B v(0) & + & C A^{k-1-\ell-1} B v(1) & + & \dots & + & C B v(k-1-\ell)\\
C A^{k-\ell} B v(0) & + & C A^{k-\ell-1} B v(1) & + & \dots & +  & C A B v(k-1-\ell)\\
\vdotsS\\
C A^{k-3} B v(0) & + & C A^{k-4} B v(1) & + & \dots & + & C A^{\ell-2} B v(k-1-\ell)\\
C A^{k-2} B v(0) & + & C A^{k-3} B v(1) & + & \dots & + & C A^{\ell-1} B v(k-1-\ell)\\
} + C A^{\ell-1} B v(k - \ell) + \dots + C B v(k-1) \\
& =
C A^\ell \mathcal{O}_\ell^{\tu{L}}
\smat{
C A^{k-1-\ell} B & \dots & C B\\
C A^{k-\ell} B & \dots & C A B\\
\vdotsS & & \vdotsS\\
C A^{k-3} B & \dots & C A^{\ell-2} B\\
C A^{k-2} B & \dots & C A^{\ell-1} B
}
\smat{v(0)\\ \vdotsS \\ v(k-1-\ell)}  + C A^{\ell-1} B v(k - \ell) + \dots + C B v(k-1) \\
& =
C A^\ell \mathcal{O}_\ell^{\tu{L}}
\smat{
C \cdot A^{k-1-\ell} B & \dots & C \cdot B\\
C A \cdot A^{k-1-\ell} B & \dots & C A \cdot B\\
\vdotsS & & \vdotsS\\
C A^{\ell-2} \cdot A^{k-1-\ell} B & \dots & C A^{\ell-2} \cdot B\\
C A^{\ell-1} \cdot A^{k-1-\ell} B & \dots & C A^{\ell-1} \cdot B
}
\smat{v(0)\\ \vdotsS \\ v(k-1-\ell)}  + C A^{\ell-1} B v(k - \ell) + \dots + C B v(k-1) \\
& =
C A^\ell \mathcal{O}_\ell^{\tu{L}}
\bmat{
\mathcal{O}_\ell \cdot A^{k-1-\ell} B & \mathcal{O}_\ell \cdot A^{k-2-\ell} B & \dots & \mathcal{O}_\ell \cdot A  B & \mathcal{O}_\ell \cdot B}
\smat{v(0)\\ \vdotsS \\ v(k-1-\ell)} + C A^{\ell-1} B v(k - \ell) + \dots + C B v(k-1) \\
& =
\bmat{
C A^\ell \cdot A^{k-1-\ell} B & C A^\ell \cdot A^{k-2-\ell} B & \dots & C A^\ell \cdot A  B & C A^\ell \cdot B}
\smat{v(0)\\ \vdotsS \\ v(k-1-\ell)}  + C A^{\ell-1} B v(k - \ell) + \dots + C B v(k-1) \\
& =
C A^{k-1} B v(0) + C  A^{k-2} B v(1) + \dots + C A^{\ell+1}  B v(k-2-\ell) + C A^\ell B v(k-1-\ell) + C A^{\ell-1} B v(k - \ell) + \dots + C B v(k-1) \\
& =  \sum_{j=0}^{k-1} C A^{k-1-j} B v(j) 
\end{aligned}
\end{equation}
\hrulefill
\end{figure*}
Hence, from~\eqref{forced_response_xi:long_eq},
\begin{align*}
\xi(k+1) =
\left[
\begin{array}{c}
\begin{smallmatrix}
\sum_{j=0}^{k - \ell} C A^{k - \ell - j} B v(j) \\
\vdotsS\\
\sum_{j=0}^{k - 2} C A^{k - 2 - j} B v(j) \\
\sum_{j=0}^{k-1} C A^{k-1-j} B v(j) \\
\end{smallmatrix}\\
\hline
\begin{smallmatrix}
v(k-\ell+1)\\
\vdotsS\\
v(k)\\
\end{smallmatrix}
\end{array}
\right].
\end{align*}
The so-obtained $\xi(k+1)$ is equal to \eqref{xi_k+1}, as we needed to show.

We prove the second part of the statement, namely, the equivalence of \eqref{all_steps:xi_k} and \eqref{all_steps:xi_k_alt} and, specifically, the equivalence for $\xi_1$, \dots, $\xi_\ell$.
From~\eqref{all_steps:xi_k}, we have, for each $k \ge 0$, \eqref{forced_response_xi:long_eq_2}, which is displayed over two columns.
\begin{figure*}
\begin{equation}
\label{forced_response_xi:long_eq_2}
\begin{aligned}
& \smat{
\xi_1(k) \\
\xi_2(k) \\
\vdotsS\\
\xi_{\ell-1}(k) \\
\xi_\ell(k) \\
}
=
\smat{
\sum_{j=0}^{k - 1 - \ell} C A^{k - 1 - \ell - j} B v(j) \\
\sum_{j=0}^{k - 1 - \ell} C A^{k - \ell - j} B v(j) \\
\vdotsS\\
\sum_{j=0}^{k - 1 - \ell} C A^{k - 3 - j} B v(j) \\
\sum_{j=0}^{k - 1 - \ell} C A^{k - 2 - j} B v(j) \\
}
+
\smat{
0 \\
\sum_{j=k - \ell}^{k - \ell} C A^{k - \ell - j} B v(j) \\
\vdotsS \\
\sum_{j=k - \ell}^{k - 3} C A^{k - 3 - j} B v(j) \\
\sum_{j=k - \ell}^{k - 2} C A^{k - 2 - j} B v(j) \\
} =
\smat{
\sum_{j=0}^{k - 1 - \ell} C \cdot A^{k - 1 - \ell - j} B v(j) \\
\sum_{j=0}^{k - 1 - \ell} CA \cdot A^{k - 1 - \ell - j} B v(j) \\
\vdots \\
\sum_{j=0}^{k - 1 - \ell} CA^{\ell-2} \cdot A^{k - 1 - \ell - j} B v(j) \\
\sum_{j=0}^{k - 1 - \ell} CA^{\ell-1} \cdot A^{k - 1 - \ell - j} B v(j) \\
}\\
&
+
\smat{
0 \\
C B v(k-\ell) \\
\vdotsS \\
C A^{\ell-3} B v(k-\ell)+C A^{\ell - 4} B v(k-\ell+1)+\dots+C A B v(k-4)+C B v(k-3) \\
C A^{\ell-2} B v(k-\ell)+C A^{\ell - 3} B v(k-\ell+1)+\dots+C A B v(k-3)+C B v(k-2) \\
}= 
\smat{
C\\
CA\\
\vdotsS\\
CA^{\ell-2}\\
CA^{\ell-1}\\
}
\sum_{j=0}^{k - 1 - \ell} A^{k - 1 - \ell - j} B v(j)\\
& 
+
\smat{
0 & 0 & \dots & 0 & 0\\
CB                     & 0 & \dots & 0 & 0\\
CAB                    & CB                     & \dots & 0 & 0\\
\vdotsS                 & \vdotsS                 & \ddotsS &  \vdotsS & \vdotsS\\
CA^{\ell -2} B         & CA^{\ell -3} B         & \dots & CB & 0
} 
\smat{
v(k-\ell) \\
v(k-\ell+1)\\
\vdotsS\\
v(k-2) \\
v(k-1)
} 
\overset{\eqref{O_ell_T_ell_R_ell}}{=}
\mathcal{O}_\ell \sum_{j=0}^{k-1-\ell} A^{k-1-\ell-j} B v(j) + \mathcal{T}_\ell \smat{v(k-\ell)\\
\rotatebox{90}{\tiny\dots}\\
v(k-1)\\}.
\end{aligned}
\end{equation}
\hrulefill
\end{figure*}
\eqref{forced_response_xi:long_eq_2} corresponds to~\eqref{all_steps:xi_k_alt} as we needed to show.
\end{proof}

As an immediate consequence of Claim~\ref{claim:forced_response_xi}, we have the next result.

\begin{claim}
\label{claim:sol_x_is_sol_xi}
For $(A,C)$ observable, let $x(\cdot)$ be the solution to~\eqref{sys_x} with initial condition $x(0) = 0$ and input $\{ u(k) \}_{k = 0}^\infty$, and $y(\cdot) = C x(\cdot)$ be the corresponding output response.
Then, the solution $\xi(\cdot)$ to~\eqref{sys_xi} with initial condition $\xi(0)=0$ and input $\{ v(k) \}_{k = 0}^\infty = \{ u(k) \}_{k = 0}^\infty$ satisfies
\begin{align}
\label{sol_x_is_sol_xi:xi_k=stack}
\xi(k) = 
\left[
\begin{array}{c}
\begin{smallmatrix}
y(k-\ell)\\
\vdotsS\\
y(k-1)\\
\end{smallmatrix}\\
\hline
\begin{smallmatrix}
u(k-\ell)\\
\vdotsS\\
u(k-1)\\
\end{smallmatrix}
\end{array}
\right] \quad \forall k \ge \ell.
\end{align}
\end{claim}

\begin{proof}
The solution $x(\cdot)$ to~\eqref{sys_x} with initial condition $x(0) = 0$ and input $\{ u(k) \}_{k = 0}^\infty$ satisfies, for each $k \ge 1$,
\begin{align*}
x(k) = A^{k-1} B u(0) + \dots + A B u(k-2) + B u(k-1)
\end{align*}
and the corresponding $y(\cdot)$ satisfies, for each $k \ge 1$,
\begin{align}
y(k) & = C A^{k-1} B u(0) + \dots + C A B u(k-2) + C B u(k-1) \notag \\
&= \sum_{j=0}^{k-1} C A^{k-1 - j} B u(j). \label{sol_x_is_sol_xi:y_k}
\end{align}
The solution $\xi(\cdot)$ to~\eqref{sys_xi} with initial condition $\xi(0)=0$ and input $\{ v(k) \}_{k = 0}^\infty = \{ u(k) \}_{k = 0}^\infty$ satisfies, by Claim~\ref{claim:forced_response_xi},
\begin{align*}
\xi(k) =
\left[
\begin{array}{c}
\begin{smallmatrix}
\sum_{j=0}^{k - 1 - \ell} C A^{k - 1 - \ell - j} B u(j) \\
\sum_{j=0}^{k - \ell} C A^{k - \ell - j} B u(j) \\
\vdotsS\\
\sum_{j=0}^{k - 3} C A^{k - 3 - j} B u(j) \\
\sum_{j=0}^{k - 2} C A^{k - 2 - j} B u(j) \\
\end{smallmatrix}\\
\hline
\begin{smallmatrix}
u(k-\ell)\\
\vdotsS\\
u(k-1)\\
\end{smallmatrix}
\end{array}
\right]\quad \forall k \ge \ell.
\end{align*}
We note that for $k \ge \ell+1$ all sums are nonempty, and for $k=\ell$ only the sum $\sum_{j=0}^{k - 1 - \ell} C A^{k - 1 - \ell - j} B u(j)$ is empty and, conventionally, equal to $0 = y(0)$.
By~\eqref{sol_x_is_sol_xi:y_k}, for each $k \ge \ell$,
\begin{align*}
& \sum_{j=0}^{k - 1 - \ell} C A^{k - 1 - \ell - j} B u(j) = y(k-\ell) , \dots,  \\
& \sum_{j=0}^{k - 2} C A^{k - 2 - j} B u(j) = y(k-1). 
\end{align*}
Hence,
\begin{align*}
\xi(k) =
\left[
\begin{array}{c}
\begin{smallmatrix}
y(k-\ell) \\
\vdotsS\\
y(k-1)
\end{smallmatrix}\\
\hline
\begin{smallmatrix}
u(k-\ell)\\
\vdotsS\\
u(k-1)\\
\end{smallmatrix}
\end{array}
\right]\quad \forall k \ge \ell,
\end{align*}
as we needed to prove.
\end{proof}

%%%%%%%%%%%%%%%%%%%%%%%%%%%%%%%%%%%%%%%%%%%%%%%%%%%%%%%%%%%%%%%%%%%%%%%%%%%%%%%%%%%
%%%%%%%%%%%%%%%%%%%%%%%%%%%%%%%%%%%%%%%%%%%%%%%%%%%%%%%%%%%%%%%%%%%%%%%%%%%%%%%%%%%
\section{Proof of Lemma~\ref{lemma:imH=R(A,B)}}
\label{app:proof_lemma_imH=R(A,B)}

We need to prove that $\im H_\ell = R(\mathbf{A}_\ell, \mathbf{B}_\ell)$.\newline
($\supseteq$) Consider an arbitrary $z\in  R(\mathbf{A}_\ell, \mathbf{B}_\ell)$.
We would like to show that there exists a vector $w$ such that $z=H_\ell w$ or, equivalently by~\eqref{Hell}, that there exist vectors $w_1$ and $w_2$ such that
\begin{align}
\label{w_to_find}
\bmat{z_1\\ z_2} \! \! = \! \! 
\left[
\begin{array}{c|c}
\mathcal{O}_\ell & \mathcal{T}_\ell\\
\hline
0 & I
\end{array}
\right] \! \!
\bmat{w_1\\ w_2}
\! \!  \iff  \! \!
\begin{cases}
z_1 \!  = \! \mathcal{O}_\ell w_1 \! + \! \mathcal{T}_\ell w_2 \\ 
z_2 \! = \! w_2.
\end{cases}
\end{align}
For each  $z\in  R(\mathbf{A}_\ell, \mathbf{B}_\ell)$, the definition of reachable state stipulates that there is a time $\tau \ge 0$ and a sequence $\{v(k)\}_{k = 0}^{\tau-1}$ such that the solution $\xi(\cdot)$ to~\eqref{sys_xi} with initial condition $\xi(0) = 0$ and input $\{v(k)\}_{k = 0}^{\tau-1}$ satisfies $\xi(\tau)=z$.
The solution $\xi(\cdot)$ to~\eqref{sys_xi} with initial condition $\xi(0) = 0$ and input $\{v(k)\}_{k = 0}^{\tau-1}$ is characterized by~\eqref{all_steps:xi_k_alt} in Claim~\ref{claim:forced_response_xi}, which considers fictitious inputs ${v}(-\ell)$, \dots ${v}(-1)$. 
We then divide into different cases based on the value of $\tau$ corresponding to each $z\in  R(\mathbf{A}_\ell, \mathbf{B}_\ell)$.
For $\tau =0$, $z = \xi(\tau) = 0$ and \eqref{w_to_find} holds with $w_1= 0$ and $w_2 = 0$.
For $\tau \in \{ 1, \dots, \ell \}$, we rewrite \eqref{all_steps:xi_k_alt} as
\begin{align*}
z =
\xi(\tau) = 
\left[
\begin{array}{c}
\begin{smallmatrix}
\xi_1(\tau) \\
\vdotsS\\
\xi_\ell(\tau) \\
\end{smallmatrix}\\
\hline
\begin{smallmatrix}
\xi_{\ell+1}(\tau)\\
\vdotsS\\
\xi_{2\ell}(\tau)
\end{smallmatrix}
\end{array}
\right]
=
\left[
\begin{array}{c}
\mathcal{T}_\ell \smat{
v(\tau-\ell)\\
\vdotsS \\
v(\tau-1)\\}
\\
\hline
\begin{smallmatrix}
v(\tau-\ell)\\
\vdotsS\\
v(\tau-1)\\
\end{smallmatrix}
\end{array}
\right]
\end{align*}
where the sum in~\eqref{all_steps:xi_k_alt} is always empty for $\tau \in \{ 1, \dots, \ell \}$ and some components of $\smat{v(\tau-\ell)\\
\vdotsS \\
v(\tau-1)}$
are zero based on the specific value of $\tau \in \{ 1, \dots, \ell \}$;
hence, \eqref{w_to_find} holds with $w_1=0$ and $w_2 = \smat{v(\tau-\ell)\\
\vdotsS \\
v(\tau-1)}$.
For $\tau \in \{ \ell + 1, \ell + 2 \dots,  \}$, we rewrite \eqref{all_steps:xi_k_alt} as
\begin{align*}
z & =
\xi(\tau)
=
\left[
\begin{array}{c}
\begin{smallmatrix}
\xi_1(\tau) \\
\vdotsS\\
\xi_\ell(\tau) \\
\end{smallmatrix}\\
\hline
\begin{smallmatrix}
\xi_{\ell+1}(\tau)\\
\vdotsS\\
\xi_{2\ell}(\tau)
\end{smallmatrix}
\end{array}
\right] \\
& =
\left[
\begin{array}{c}
\mathcal{O}_\ell \sum_{j=0}^{\tau-1-\ell} A^{\tau-1-\ell-j} B v(j) + \mathcal{T}_\ell \smat{v(\tau-\ell)\\
\vdotsS \\
v(\tau-1)\\}
\\
\hline
\begin{smallmatrix}
v(\tau-\ell)\\
\vdotsS\\
v(\tau-1)\\
\end{smallmatrix}
\end{array}
\right]
\end{align*}
where the sum is never empty since $\tau \ge \ell +1$ and no component of $\smat{v(\tau-\ell)\\
\vdotsS \\
v(\tau-1)}$ is zero since $\tau \ge \ell$; hence, \eqref{w_to_find} holds with $w_1=\sum_{j=0}^{\tau-1-\ell} A^{\tau-1-\ell-j} B v(j)$ and $w_2=\smat{v(\tau-\ell)\\
\vdotsS \\
v(\tau-1)}$. 

($\subseteq$) Consider an arbitrary $z \in \im H_\ell$.
We would like to show that $z \in R(\mathbf{A}_\ell,\mathbf{B}_\ell)$.
To do so, we proceed in this way: for each $z \in \im H_\ell$, we first construct a solution $x_z(\cdot)$ to~\eqref{sys_x} with zero initial condition and a suitable input sequence; with $x_z(\cdot)$ we associate a solution $\xi_z(\cdot)$ to~\eqref{sys_xi} by invoking Claim~\ref{claim:sol_x_is_sol_xi} to conclude that $z \in R(\mathbf{A}_\ell, \mathbf{B}_\ell)$.

Since $z = (z_1, \dots, z_\ell, z_{\ell+1}, \dots, z_{2 \ell}) \in \im H_\ell$, there exist a vector $\smat{\chi\\ \psi}$ such that $z = H_\ell \smat{\chi\\ \psi} =
\smat{
\mathcal{O}_\ell & \mathcal{T}_\ell\\
0 & I
}
\smat{\chi\\ \psi} $ or, equivalently,
\begin{align}
\label{rel_z_chi_psi}
\smat{
z_1 \\
\vdotsS \\
z_\ell\\[2pt]
\hline
z_{\ell+1}\rule{0pt}{5pt}\\
\vdotsS \\
z_{2\ell}
}
=
\bmat{
\mathcal{O}_\ell \chi + \mathcal{T}_\ell \psi\\
\psi
}.
\end{align}
For each $z$, we construct a solution $x_z(\cdot)$ as follows.
By reachability of the pair $(A,B)$, there exist a time $\tau \ge 0$ and a sequence $u_z(0)$, \dots, $u_z(\tau - 1)$ such that the solution $x_z(\cdot)$ to~\eqref{sys_x} with $x_z(0)=0$ and input $u_z(0)$, \dots, $u_z(\tau - 1)$ satisfies $x_z(\tau) = \chi$.
Let $\rho$ be the smallest integer for which $[B \ A B \ \dots \ A^{\rho-1} B]$ has full row rank, which exists by reachability of the pair $(A,B)$.
This guarantees that there exists $\tau \le \rho$; however, in view of the following, we can always take $\tau = \max\{ \rho, \ell \}$ by setting some of the initial elements of the sequence $u_z(0)$, \dots, $u_z(\tau - 1)$ to be the zero vector.
For the so-selected $\tau$, which ensures $x_z(\tau) = \chi$, we continue the input sequence with $u_z(\tau)=\psi_1$, \dots, $u_z(\tau + \ell - 1) = \psi_\ell$.
The output $y_z(\cdot)$ corresponding to the so-constructed solution $x_z(\cdot)$ satisfies
\begin{align*}
& y_z(\tau) = C x_z(\tau) , \,\, y_z(\tau+1) = C A x_z(\tau) + C B u_z(\tau), \dots, \\
& y_z(\tau+\ell-1)  = C A^{\ell -1} x_z(\tau)+  C A^{\ell -2} B u_z(\tau)+ \dots +  \\
& \hspace*{25.2mm}  C A B u_z(\tau+ \ell -3) + C B u_z(\tau+ \ell -2)
\end{align*}
or equivalently, by stacking the previous equalities,
\begin{align*}
\smat{
y_z(\tau)\\
\vdotsS\\
y_z(\tau+\ell-1)
}
\! = \!
\mathcal{O}_\ell x_z(\tau)
\! + \!
\mathcal{T}_\ell 
\smat{
u_z(\tau)\\
\vdotsS \\
u_z(\tau+\ell-1)}
\! = \!
\mathcal{O}_\ell \chi + \mathcal{T}_\ell \psi.
\end{align*}
So far, we have then constructed a solution $x_z(\cdot)$ to~\eqref{sys_x} with initial condition $x_z(0)=0$ and input sequence $u_z(0)$, \dots, $u_z(\tau - 1)$, $u_z(\tau)=\psi_1$, \dots, $u_z(\tau + \ell - 1) = \psi_\ell$ such that $x_z(\tau) = \chi$ and $\smat{
y_z(\tau)\\
\vdotsS \\
y_z(\tau+\ell-1)
}
=
\mathcal{O}_\ell \chi + \mathcal{T}_\ell \psi$.
By Claim~\ref{claim:sol_x_is_sol_xi}, the solution $\xi_z(\cdot)$ to~\eqref{sys_xi} with initial condition $\xi(0) = 0$ and input $u_z(0)$, \dots, $u_z(\tau - 1)$, $u_z(\tau)=\psi_1$, \dots, $u_z(\tau + \ell - 1) = \psi_\ell$ satisfies
\begin{align*}
\xi_z(\tau + \ell)
=
\left[
\begin{array}{c}
\begin{smallmatrix}
y_z(\tau) \\
\vdotsS\\
y_z(\tau+\ell-1) \\
\end{smallmatrix}\\
\hline
\begin{smallmatrix}
u_z(\tau) \\
\vdotsS\\
u_z(\tau+\ell-1) \\
\end{smallmatrix}
\end{array}
\right]
=
\left[
\begin{array}{c}
\mathcal{O}_\ell \chi + \mathcal{T}_\ell \psi\\
\hline
\psi
\end{array}
\right]
\overset{\eqref{rel_z_chi_psi}}{=}
z.
\end{align*}
By construction, $\xi_z(\tau + \ell) \in R(\mathbf{A}_\ell,\mathbf{B}_\ell)$ and so does $z$.

%%%%%%%%%%%%%%%%%%%%%%%%%%%%%%%%%%%%%%%%%%%%%%%%%%%%%%%%%%%%%%%%%%%%%%%%%%%%%%%%%%%
%%%%%%%%%%%%%%%%%%%%%%%%%%%%%%%%%%%%%%%%%%%%%%%%%%%%%%%%%%%%%%%%%%%%%%%%%%%%%%%%%%%
\section{Proof of Lemma~\ref{lemma:AbfBbf_reach_pl=n}}
\label{app:proof_lemma_AbfBbf_reach_pl=n}

From~\eqref{Hell}, $H_\ell = \smat{I & \mathcal{T}_\ell \\ 0 & I } \smat{\mathcal{O}_\ell & 0 \\ 0 & I_{m\ell}}$, so $\rank H_\ell = \rank \mathcal{O}_\ell + m\ell = n+m\ell$ since $(A,C)$ is observable.

($\Longrightarrow$) If $(\mathbf{A}_\ell,\mathbf{B}_\ell)$ is reachable, $(p+m)\ell = \dim R(\mathbf{A}_\ell,\mathbf{B}_\ell)$ and, by Lemma~\ref{lemma:imH=R(A,B)},  $ \dim R(\mathbf{A}_\ell,\mathbf{B}_\ell) = \rank H_\ell$; hence $\rank H_\ell = (p+m)\ell$.
Since we established $\rank H_\ell = n+m\ell$, we conclude $p\ell =n$. 

($\Longleftarrow$) If $p\ell = n$, we have $n + m\ell = \rank H_\ell = \dim R(\mathbf{A}_\ell,\mathbf{B}_\ell)$ by Lemma~\ref{lemma:imH=R(A,B)} and, thus, $\dim R(\mathbf{A}_\ell,\mathbf{B}_\ell) = n + m\ell = (p+m)\ell$.
Since $(p+m)\ell$ is the dimension of the state associated with $(\mathbf{A}_\ell,\mathbf{B}_\ell)$, the pair $(\mathbf{A}_\ell,\mathbf{B}_\ell)$ is reachable.

%%%%%%%%%%%%%%%%%%%%%%%%%%%%%%%%%%%%%%%%%%%%%%%%%%%%%%%%%%%%%%%%%%%%%%%%%%%%%%%%%%%
%%%%%%%%%%%%%%%%%%%%%%%%%%%%%%%%%%%%%%%%%%%%%%%%%%%%%%%%%%%%%%%%%%%%%%%%%%%%%%%%%%%
\section{Proof of Lemma~\ref{lemma:io_sys_x_is_io_sys_xi}}
\label{app:proof_lemma_io_sys_x_is_io_sys_xi}

We prove \eqref{io_sys_x_is_io_sys_xi:xi_k} first.
For each $\hat{x}$ and sequence $\{ u(k)\}_{k = 0}^\infty$, the solution $x(\cdot)$ to~\eqref{sys_x} with initial condition $x(0) = \hat x$ and input $\{u(k) \}_{k = 0}^\infty$ satisfies, for each $k \ge \ell$,
\begin{align}
\smat{
y(k-\ell) \\
\vdotsS \\ 
y(k-1)\\
}
& = \mathcal{O}_\ell  x(k-\ell) + \mathcal{T}_\ell 
\smat{
u(k-\ell)\\
\vdotsS \\
u(k-1)
}. \label{sol_sys_x:o_s_i_again}
\end{align}
by \eqref{sol_sys_x:o_s_i} in Claim~\ref{claim:sol_sys_x}.
For each $\hat{x}$ and $\{ u(k) \}_{k = 0}^\infty$, select $\hat{\xi} = (\hat{\xi}_1, \dots, \hat{\xi}_\ell, \hat{\xi}_{\ell+1}, \dots, \hat{\xi}_{\ell+\ell})$ such that
\begin{align}
\smat{
\hat{\xi}_1\\
\vdotsS\\
\hat{\xi}_\ell
}
=
\mathcal{O}_\ell \hat{x} + \mathcal{T}_\ell 
\smat{
u(0)\\
\vdotsS \\
u(\ell-1)
}
\text{ and }
\smat{
\hat{\xi}_{\ell+1}\\
\vdotsS\\
\hat{\xi}_{\ell+\ell}
}
=
\smat{
u(0)\\
\vdotsS \\
u(\ell-1)
}.
\label{selection_init_cond_sys_xi}
\end{align}
We can now show \eqref{io_sys_x_is_io_sys_xi:xi_k} by mathematical induction.
As for the base case, we need to show that $(y(0), \dots, y(\ell-1), u(0), \dots, u(\ell-1)) = \xi(\ell)$.
Indeed,
\begin{align*}
\xi(\ell) = \hat{\xi} \overset{\eqref{selection_init_cond_sys_xi}}{=}
\left[
\begin{array}{c}
\mathcal{O}_\ell \hat{x} + \mathcal{T}_\ell 
\smat{
u(0)\\
\vdotsS \\
u(\ell-1)
}\\
\hline
\smat{
u(0)\\
\vdotsS \\
u(\ell-1)
}
\end{array}
\right]
\overset{\eqref{sol_sys_x:o_s_i_again}}{=}
\left[
\begin{array}{c}
\begin{smallmatrix}
y(0) \\
\vdotsS \\
y(\ell-1)
\end{smallmatrix}
\\
\hline
\begin{smallmatrix}
u(0) \\
\vdotsS \\
u(\ell-1)
\end{smallmatrix}
\end{array}
\right].
\end{align*}
As for the induction step, we suppose that, for $k \ge \ell$,
\begin{align}
\xi(k) =
\left[
\begin{array}{c}
\begin{smallmatrix}
\xi_1(k)\\
\vdotsS\\
\xi_\ell(k)
\end{smallmatrix}
\\
\hline
\begin{smallmatrix}
\xi_{\ell+1}(k)\\
\vdotsS\\
\xi_{\ell+\ell}(k)
\end{smallmatrix}
\end{array}
\right]
=
\left[
\begin{array}{c}
\begin{smallmatrix}
y(k-\ell) \\
\vdotsS \\ 
y(k-1)
\end{smallmatrix}
\\
\hline
\begin{smallmatrix}
u(k-\ell)\\
\vdotsS\\
u(k-1)
\end{smallmatrix}
\end{array}
\right]
\label{hp_ind_step}
\end{align}
and we need to show that $\xi(k+1) = \big( 
y(k-\ell+1), \dots, y(k), u(k-\ell+1), \dots, u(k)
\big)$.
We have
\begin{align*}
& \xi(k+1) = 
\left[
\begin{array}{c}
\begin{smallmatrix}
\xi_1(k+1) \\
\vdotsS\\
\xi_{\ell-1}(k+1) \\
\xi_{\ell}(k+1) \\
\end{smallmatrix}\\
\hline
\begin{smallmatrix}
\xi_{\ell+1}(k+1) \\
\vdotsS\\
\xi_{\ell+\ell-1}(k+1) \\
\xi_{\ell+\ell}(k+1) \\
\end{smallmatrix}
\end{array}
\right]
=
\mathbf{A}_\ell \xi(k) + \mathbf{B}_\ell v(k) \\
& \overset{\eqref{sys_xi:Aell},\eqref{sys_xi:A0_L_B}}{=}
\left[
\begin{array}{c}
\begin{smallmatrix}
\xi_2(k) \\
\vdotsS\\
\xi_{\ell}(k) \\
C A^\ell \mathcal{O}_\ell^{\tu{L}}
\smat{
\xi_1(k) \\
\vdotsS \\
\xi_{\ell}(k) \\
}
+
(
C \mathcal{R}_\ell -CA^\ell \mathcal{O}_\ell^{\tu{L}} \mathcal{T}_\ell
)
\smat{
\xi_{\ell+1}(k) \\
\vdotsS \\
\xi_{\ell+\ell}(k) \\
}
\\
\end{smallmatrix}\\
\hline
\begin{smallmatrix}
\xi_{\ell+2}(k) \\
\vdotsS\\
\xi_{\ell+\ell}(k) \\
v(k) \\
\end{smallmatrix}
\end{array}
\right] \notag \\
& \overset{\eqref{hp_ind_step}}{=}
\left[
\begin{array}{c}
\begin{smallmatrix}
y(k-\ell+1) \\
\vdotsS\\
y(k-1) \\
C A^\ell \mathcal{O}_\ell^{\tu{L}}
\smat{
y(k-\ell) \\
\vdotsS \\
y(k-1) \\
}
+
(
C \mathcal{R}_\ell -CA^\ell \mathcal{O}_\ell^{\tu{L}} \mathcal{T}_\ell
)
\smat{
u(k-\ell) \\
\vdotsS \\
u(k-1)\\
} 
\\
\end{smallmatrix}\\
\hline
\begin{smallmatrix}
u(k-\ell+1)\\
\vdotsS\\
u(k-1)\\
v(k) \\
\end{smallmatrix}
\end{array}
\right] \notag \\
& =
\left[
\begin{array}{c}
\begin{smallmatrix}
y(k-\ell+1) \\
\vdotsS\\
y(k-1) \\
C A^\ell \mathcal{O}_\ell^{\tu{L}}
\smat{
y(k-\ell) \\
\vdotsS \\
y(k-1) \\
}
+
(
C \mathcal{R}_\ell -CA^\ell \mathcal{O}_\ell^{\tu{L}} \mathcal{T}_\ell 
)
\smat{
u(k-\ell) \\
\vdotsS \\
u(k-1)\\
}
\\
\end{smallmatrix}\\
\hline
\begin{smallmatrix}
u(k-\ell+1)\\
\vdotsS\\
u(k-1)\\
u(k) \\
\end{smallmatrix}
\end{array}
\right]
\end{align*}
where the last equality follows from the statement since $v(j) = u(j)$ for all $j \ge \ell$.
From the last expression, the proof of the induction step is concluded if we show that
\begin{align*}
C A^\ell \mathcal{O}_\ell^{\tu{L}}
\smat{
\! y(k-\ell) \! \\
\! \vdotsS \! \\
\! y(k-1) \! \\
}
\! + \!
(
C \mathcal{R}_\ell - CA^\ell \mathcal{O}_\ell^{\tu{L}} \mathcal{T}_\ell 
)
\smat{
\! u(k-\ell) \! \\
\! \vdotsS \! \\
\! u(k-1) \! \\
} \! = \! y(k),
\end{align*}
which holds by \eqref{sol_sys_x:o_o_i} in Claim~\ref{claim:sol_sys_x}.
After proving~\eqref{io_sys_x_is_io_sys_xi:xi_k}, \eqref{io_sys_x_is_io_sys_xi:y_k} holds because, by Claim~\ref{claim:sol_sys_x} and \eqref{sol_sys_x:o_o_i}, $y$ satisfies, for each $k \ge \ell$,
$y(k)=Z_\ell \big(y(k-\ell),\dots,y(k-1),u(k-\ell),\dots,u(k-1) \big) \overset{\eqref{io_sys_x_is_io_sys_xi:xi_k}}{=} Z_\ell \xi(k)$.

%%%%%%%%%%%%%%%%%%%%%%%%%%%%%%%%%%%%%%%%%%%%%%%%%%%%%%%%%%%%%%%%%%%%%%%%%%%%%%%%%%%
%%%%%%%%%%%%%%%%%%%%%%%%%%%%%%%%%%%%%%%%%%%%%%%%%%%%%%%%%%%%%%%%%%%%%%%%%%%%%%%%%%%
\section{Proof of Lemma~\ref{lemma:data_relation}}
\label{app:proof_lemma_data_relation}

By \eqref{sol_sys_x:o_o_i} in Claim~\ref{claim:sol_sys_x}, the quantities of the experiment satisfy, for $k=\ells, \dots, T$,
\begin{align}
& y^{\tu{m}}(k) \! - \! d^y(k) \! = \! y(k) \! = \! 
Z_{1\ells}
\smat{
y(k-\ells)\\
\vdotsS \\
y(k-1)
}
+ Z_{2\ells}
\smat{
u(k-\ells)\\
\vdotsS \\
u(k-1)
} \notag\\
& \! = \! 
Z_{1\ells}
\smat{
y^{\tu{m}}(k-\ells) - d^y(k-\ells)\\
\vdotsS \\
y^{\tu{m}}(k-1) - d^y (k-1)
}
\! + \! Z_{2\ells}
\smat{
u^{\tu{m}}(k-\ell) - d^u(k-\ells)\\
\vdotsS \\
u^{\tu{m}}(k-1) - d^u(k-1) 
}\notag
\end{align}
and this is equivalent to~\eqref{exp_data_ymk}.
\eqref{exp_data_ymk} is equivalent, for $k=\ells, \dots, T$, to
\begin{align*}
& \left[
\begin{array}{c}
\begin{smallmatrix}
y^{\tu{m}}(k-\ells+1)\\
\vdotsS\\
y^{\tu{m}}(k-1)\\
y^{\tu{m}}(k)
\end{smallmatrix}\\
\hline
\begin{smallmatrix}
u^{\tu{m}}(k-\ells+1)\\
\vdotsS\\
u^{\tu{m}}(k-1)\\
u^{\tu{m}}(k)
\end{smallmatrix}
\end{array}
\right] \\
& =
\left[
\begin{array}{c}
\begin{smallmatrix}
y^{\tu{m}}(k-\ells+1)\\
\vdotsS\\
y^{\tu{m}}(k-1)\\
\left\{
\begin{smallmatrix}
Z_{1\ells}
\smat{
y^{\tu{m}}(k-\ells)\\
\vdotsS \\
y^{\tu{m}}(k-1)
}
+ Z_{2\ells}
\smat{
u^{\tu{m}}(k-\ells)\\
\vdotsS \\
u^{\tu{m}}(k-1)
}\\
-
Z_{1\ells}
\smat{
 d^y(k-\ells)\\
\vdotsS \\
d^y (k-1)
}
-Z_{2\ells}
\smat{
d^u(k-\ells)\\
\vdotsS \\
d^u(k-1) 
} +d^y(k)
\end{smallmatrix}
\right\}
\end{smallmatrix}\\
\hline
\begin{smallmatrix}
u^{\tu{m}}(k-\ells+1)\\
\vdotsS\\
u^{\tu{m}}(k-1)\\
u^{\tu{m}}(k)
\end{smallmatrix}
\end{array}
\right], \notag 
\end{align*}
which is also equivalent to \eqref{exp_data_veck} for $\mathbf{F}_\ells$, $\mathbf{L}_\ells$ and $\mathbf{B}_\ells$ as in~\eqref{sys_xi:A0_L_B}.
Finally, with the definitions in~\eqref{Astar}, \eqref{Bstar_d}, \eqref{Psi1_Psi0_U0_Delta10}, we stack \eqref{exp_data_veck} column after column from $k=\ells$ to $k=T$ and obtain \eqref{exp_data_mat}.

%%%%%%%%%%%%%%%%%%%%%%%%%%%%%%%%%%%%%%%%%%%%%%%%%%%%%%%%%%%%%%%%%%%%%%%%%%%%%%%%%%%
%%%%%%%%%%%%%%%%%%%%%%%%%%%%%%%%%%%%%%%%%%%%%%%%%%%%%%%%%%%%%%%%%%%%%%%%%%%%%%%%%%%
\section{Proof of Lemma~\ref{lemma:cons_asmpt_for_set_C}}
\label{app:proof_lemma_cons_asmpt_for_set_C}

If $\mathscr{A} \succ 0$, $\mathscr{Z}$ and $\mathscr{Q}$ in~\eqref{setC:ZQ} are well-defined and algebraic computations yield \eqref{setC_ZAQ} from~\eqref{setC_ABC}.
\eqref{exp_data_mat} in Lemma~\ref{lemma:data_relation} and $\Delta_{10} \in \mathcal{D}$ imply $Z_\ells \in \mathcal{C}$ and, by $Z_\ells \in \mathcal{C}$, \eqref{setC_ZAQ} yields that $\mathscr{Q} \succeq (Z_\ells - \mathscr{Z}) \mathscr{A} (Z_\ells - \mathscr{Z})^\top \succeq 0$ by $\mathscr{A} \succ 0$.
$\mathscr{A} \succ 0$ and $\mathscr{Q} \succeq 0$ allow applying \cite[Proposition~1]{AndreaPetersen2022} and obtaining \eqref{setC_unitBall} from~\eqref{setC_ZAQ}; they also allow applying \cite[Lemma 2]{AndreaPetersen2022} and showing with analogous arguments that the nonempty set $\mathcal{C}$ is bounded with respect to any matrix norm.

%%%%%%%%%%%%%%%%%%%%%%%%%%%%%%%%%%%%%%%%%%%%%%%%%%%%%%%%%%%%%%%%%%%%%%%%%%%%%%%%%%%
%%%%%%%%%%%%%%%%%%%%%%%%%%%%%%%%%%%%%%%%%%%%%%%%%%%%%%%%%%%%%%%%%%%%%%%%%%%%%%%%%%%
\section{Proof of Lemma~\ref{lemma:Petersen}}
\label{app:proof_lemma_Petersen}

\begingroup
\thinmuskip=.8mu plus 1mu minus 1mu
\medmuskip=1mu plus 1mu minus 1mu
\thickmuskip=1.2mu plus 1mu minus 1mu
By $P \succ 0$ and Schur complement, \eqref{rob_contr_probl_sys_xi:ineq} is equivalent to
\begin{align}\label{eq:Petersen1-1}
& \bmat{
- P & (\mathbf{F}_\ells + \mathbf{L}_\ells Z + \mathbf{B}_\ells \mathbf{K}) P\\
\star & -P
} \prec 0 \quad \forall Z \in \mathcal{C}. 
\end{align}
There exist $\mathbf{K}$ and $P = P^\top \succ 0$ satisfying \eqref{eq:Petersen1-1} if and only if there exist $\mathbf{Y}$ and $P = P^\top \succ 0$ satisfying
\begin{align*}
& \bmat{
- P & \mathbf{F}_\ells P + \mathbf{L}_\ells Z P + \mathbf{B}_\ells \mathbf{Y} \\
\star & -P
} \prec 0 \quad \forall Z \in \mathcal{C},
\end{align*}
with $\mathbf{Y} = \mathbf{K} P$.
By $\mathcal{C}$ in~\eqref{setC_unitBall}, which is obtained under Assumption~\ref{ass:Psi0} in Lemma~\ref{lemma:cons_asmpt_for_set_C}, the previous condition holds if and only if
\begin{align*}
0 \succ & 
\bmat{
- P & \mathbf{F}_\ells P + \mathbf{L}_\ells (\mathscr{Z} + \mathscr{Q}^{1/2} \Upsilon \mathscr{A}^{-1/2} ) P + \mathbf{B}_\ells \mathbf{Y} \\
\star & -P
} \\
= & \bmat{
- P & \mathbf{F}_\ells P + \mathbf{L}_\ells \mathscr{Z} P + \mathbf{B}_\ells \mathbf{Y} \\
\star & -P
} \\
& + \bmat{\mathbf{L}_\ells \mathscr{Q}^{1/2}\\
0}
\Upsilon
\bmat{
0\\
P \mathscr{A}^{-1/2}}^\top \\
& + \bmat{
0\\
P \mathscr{A}^{-1/2}}
\Upsilon^\top
\bmat{\mathbf{L}_\ells \mathscr{Q}^{1/2}\\
0}^\top \quad \forall \Upsilon \text{ with } \| \Upsilon \| \le 1.
\end{align*}
By Petersen's lemma \cite{petersen1987stabilization} in the version reported in~\cite[Fact~1]{AndreaPetersen2022}, this condition holds if and only if there exists $\lambda > 0$ such that
\begin{align*}
& 0 \succ 
\bmat{
- P & \mathbf{F}_\ells P + \mathbf{L}_\ells \mathscr{Z} P + \mathbf{B}_\ells \mathbf{Y} \\
\star & -P
} \\
& +
\frac{1}{\lambda}
\bmat{\mathbf{L}_\ells \mathscr{Q}^{1/2}\\
0}
\bmat{
\mathbf{L}_\ells \mathscr{Q}^{1/2}\\
0}^\top \! \! \!  +
\lambda
\bmat{
0\\
P \mathscr{A}^{-1/2}}
\bmat{
0\\
P \mathscr{A}^{-1/2}}^\top \\
& =
\bmat{
- P + \frac{1}{\lambda} \mathbf{L}_\ells \mathscr{Q} \mathbf{L}_\ells^\top & \mathbf{F}_\ells P + \mathbf{L}_\ells \mathscr{Z} P + \mathbf{B}_\ells \mathbf{Y} \\
\star & -P + \lambda P \mathscr{A}^{-1} P
}.
\end{align*}
The existence of $\mathbf{Y}$, $P = P^\top \succ 0$, $\lambda >0$ such that this condition holds is equivalent to the existence of $\mathbf{Y}$, $P = P^\top \succ 0$ such that
\begin{align}
\label{LMIinZQA}
& 0 \succ 
\bmat{
- P + \mathbf{L}_\ells \mathscr{Q} \mathbf{L}_\ells^\top & \mathbf{F}_\ells P + \mathbf{L}_\ells \mathscr{Z} P + \mathbf{B}_\ells \mathbf{Y} \\
\star & -P + P \mathscr{A}^{-1} P
}
\end{align}
as can be shown by multiplying or dividing by $\lambda >0$.
By the definitions of $\mathscr{Q}$ and $\mathscr{Z}$ in~\eqref{setC:ZQ}, this condition is equivalent to
\begin{align*}
& 0 \succ \\ 
& \!\! \bmat{
\!- P - \mathbf{L}_\ells  \mathscr{C} \mathbf{L}_\ells^\top & \mathbf{F}_\ells P + \mathbf{B}_\ells \mathbf{Y}\! \\
\star & -P
}\!\!
+\!\!
\bmat{
\mathbf{L}_\ells \mathscr{B}\\
-P
}
\mathscr{A}^{-1}\!\!\!\!
\bmat{
\mathbf{L}_\ells \mathscr{B}\\
-P
}^\top
\end{align*}
and, by $\mathscr{A} \succ 0$ from Assumption~\ref{ass:Psi0} and Schur complement, this inequality is equivalent to~\eqref{rob_contr_probl_LMI:ineq}.
\endgroup

%%%%%%%%%%%%%%%%%%%%%%%%%%%%%%%%%%%%%%%%%%%%%%%%%%%%%%%%%%%%%%%%%%%%%%%%%%%%%%%%%%%
%%%%%%%%%%%%%%%%%%%%%%%%%%%%%%%%%%%%%%%%%%%%%%%%%%%%%%%%%%%%%%%%%%%%%%%%%%%%%%%%%%%
\section{Proof of Lemma~\ref{lemma:io_(x,chi)_vs_xi}}
\label{app:proof_lemma_io_(x,chi)_vs_xi}

For $(A,C)$ observable, we need to show that for each $\hat{x}$, $\hat{\chi}$ and sequence $\{ u(k) \}_{k = 0}^\infty$,  (i)~the solution $\smat{x(\cdot)\\ \chi(\cdot)}$ to~\eqref{sys_x_chi} with initial condition $\smat{x(0)\\ \chi(0)} = \smat{\hat{x}\\ \hat{\chi}}$ and with input $\{ u(k) \}_{k = 0}^\infty$ and the corresponding output response $y(\cdot) = C x(\cdot)$ satisfy, for all $k \ge \ell$,
\begin{align*}
& \chi(k) = \big(y(k-\ell), \dots, y(k-1), u(k-\ell), \dots, u(k-1)\big) ;
\end{align*}
(ii)~there exists $\hat{\xi}$ such that such solution $\smat{x(\cdot)\\ \chi(\cdot)}$, the corresponding output response $y(\cdot) = C x(\cdot)$, and the solution $\xi(\cdot)$ to \eqref{sys_xi} with initial condition $\xi(\ell) = \hat{\xi}$ and input $\{ v(k) \}_{k = \ell}^\infty = \{ u(k) \}_{k = \ell}^\infty$ satisfy, for all $k \ge \ell$
\begin{align*}
& \xi(k) = \big(y(k-\ell), \dots, y(k-1), u(k-\ell), \dots, u(k-1)\big) .
\end{align*}
As for (ii), \eqref{sys_x_chi} has a ``triangular'' structure so that the component $x(\cdot)$ of the solution $\smat{x(\cdot)\\ \chi(\cdot)}$ to~\eqref{sys_x_chi} is also a solution to~\eqref{sys_x}.
Hence, Lemma~\ref{lemma:io_sys_x_is_io_sys_xi} applies and (ii) holds.
As for (i), we use mathematical induction. 
As for the base case we need to show that $\chi(\ell) = \big(y(0), \dots, y(\ell-1), u(0), \dots, u(\ell-1)\big)$.
Indeed, by~\eqref{sys_x_chi:chi},
\begin{align*}
\chi(1) \! = \! \left[
\begin{array}{c}
\begin{smallmatrix}
\!\! \hat{\chi}_2 \!\!\\
\vdotsS\\
\!\! \hat{\chi}_\ell \!\!\\
\!\!\!y(0)\!\!\!\\
\end{smallmatrix}\\
\hline
\begin{smallmatrix}
\!\!\!\hat{\chi}_{\ell+2}\!\!\! \\
\vdotsS\\
\!\!\!\hat{\chi}_{\ell+\ell}\!\!\!\\
\!\!u(0)\!\!\\
\end{smallmatrix}
\end{array}
\right]\!, ...,
\chi(\ell-1) \! =  \! \left[
\begin{array}{c}
\begin{smallmatrix}
\!\! \hat{\chi}_\ell \!\!\\
\!\!\!y(0)\!\!\!\\
\vdotsS\\
\!\!\!y(\ell-2)\!\!\!\\
\end{smallmatrix}\\
\hline
\begin{smallmatrix}
\!\!\!\hat{\chi}_{\ell+\ell}\!\!\!\\
\!\!u(0)\!\!\\
\vdotsS\\
\!\!\!u(\ell-2)\!\!\!\\
\end{smallmatrix}
\end{array}
\right]\!,
\chi(\ell) \! = \! \left[
\begin{array}{c}
\begin{smallmatrix}
\!\!\!y(0)\!\!\!\\
\vdotsS\\
\!\!\!y(\ell-1)\!\!\!\\
\end{smallmatrix}\\
\hline
\begin{smallmatrix}
\!\!u(0)\!\!\\
\vdotsS\\
\!\!\!u(\ell-1)\!\!\!\\
\end{smallmatrix}
\end{array}
\right]\!.
\end{align*}
As for the induction step, we suppose that, for $k\ge \ell$,
\begin{align}
\label{io_(x,chi)_vs_xi:ind_step_chik}
\chi(k)=
\left[
\begin{array}{c}
\begin{smallmatrix}
y(k-\ell)\\
\vdotsS\\
y(k-1)
\end{smallmatrix}\\
\hline
\begin{smallmatrix}
u(k-\ell)\\
\vdotsS\\
u(k-1)
\end{smallmatrix}
\end{array}
\right]
\end{align}
and we need to show that $\chi(k+1) = \big(y(k-\ell+1), \dots, y(k),u(k-\ell+1), \dots, u(k) \big)$.
We have
\begin{align*}
& \chi(k+1) \overset{\eqref{sys_x_chi:chi}}{=} \mathbf{F}_\ell \chi(k) + \mathbf{L}_\ell y(k) + \mathbf{B}_\ell u(k) \\
& = 
\left[
\begin{array}{c}
\begin{smallmatrix}
\chi_2(k)\\
\vdotsS\\
\chi_\ell(k)\\
y(k)
\end{smallmatrix}\\
\hline
\begin{smallmatrix}
\chi_{\ell+2}(k)\\
\vdotsS\\
\chi_{\ell+\ell}(k)\\
u(k)
\end{smallmatrix}
\end{array}
\right]
\overset{\eqref{io_(x,chi)_vs_xi:ind_step_chik}}{=}
\left[
\begin{array}{c}
\begin{smallmatrix}
y(k-\ell+1)\\
\vdotsS\\
y(k-1)\\
y(k)
\end{smallmatrix}\\
\hline
\begin{smallmatrix}
u(k-\ell+1)\\
\vdotsS\\
u(k-1)\\
u(k)
\end{smallmatrix}
\end{array}
\right]
\end{align*}
as we needed to show.

%%%%%%%%%%%%%%%%%%%%%%%%%%%%%%%%%%%%%%%%%%%%%%%%%%%%%%%%%%%%%%%%%%%%%%%%%%%%%%%%%%%
%%%%%%%%%%%%%%%%%%%%%%%%%%%%%%%%%%%%%%%%%%%%%%%%%%%%%%%%%%%%%%%%%%%%%%%%%%%%%%%%%%%
\section{Proof of Lemma~\ref{lemma:GAS_cl_aux_implies_GAS_cl}}
\label{app:proof_lemma_GAS_cl_aux_implies_GAS_cl}

Since $\mathbf{K}$ makes the matrix $\mathbf{A}_\ell+ \mathbf{B}_\ell \mathbf{K}$ Schur, $\xi = 0$ is globally asymptotically stable for
\begin{subequations}
\label{sys_xi_CL}
\begin{align}
\xi^+ & = \mathbf{A}_\ell \xi + \mathbf{B}_\ell v \label{sys_xi_CL:xi}\\
v & = \mathbf{K} \xi. \label{sys_xi_CL:u}
\end{align}
\end{subequations}
For each $\hat{x}$, $\hat{\chi}$, the solution $\smat{x(\cdot)\\ \chi(\cdot)}$ to the closed-loop system~\eqref{sys_x_chi_CL} with initial condition $\smat{x(0)\\ \chi(0)} = \smat{\hat{x}\\ \hat{\chi} }$ has $\{ u(k) \}_{k = 0}^\infty = \{ \mathbf{K} \chi(k) \}_{k = 0}^\infty$.
From Lemma~\ref{lemma:io_(x,chi)_vs_xi}, for each $\hat{x}$, $\hat{\chi}$ and the sequence $\{ u(k) \}_{k = 0}^\infty = \{ \mathbf{K} \chi(k) \}_{k = 0}^\infty$, there exists $\hat{\xi}$ such that: the solution $\smat{x(\cdot)\\ \chi(\cdot)}$ to~\eqref{sys_x_chi_CL:x}-\eqref{sys_x_chi_CL:chi} with initial condition $\smat{x(0)\\ \chi(0)} = \smat{\hat{x}\\ \hat{\chi} }$ and with input $\{ u(k) \}_{k = 0}^\infty = \{ \mathbf{K} \chi(k) \}_{k = 0}^\infty$, and the solution $\xi(\cdot)$ to~\eqref{sys_xi_CL:xi} with initial condition $\xi(\ell) = \hat{\xi}$ and input $\{ v(k) \}_{k = \ell}^\infty = \{ \mathbf{K} \chi(k) \}_{k = \ell}^\infty$ satisfy
\begin{align}
\label{GAS_cl_aux_implies_GAS_cl:xik=chik}
\xi(k) = 
\left[
\begin{array}{c}
\begin{smallmatrix}
y(k-\ell)\\
\vdotsS\\
y(k-1)
\end{smallmatrix}\\
\hline
\begin{smallmatrix}
u(k-\ell)\\
\vdotsS\\
u(k-1)
\end{smallmatrix}
\end{array}
\right]
= \chi(k), \quad \forall k \ge \ell.
\end{align}
In other words, we have considered an arbitrary initial condition $\smat{x(0)\\ \chi(0)} = \smat{\hat{x}\\ \hat{\chi} }$ and the input $\{ u(k) \}_{k = 0}^\infty = \{ \mathbf{K} \chi(k) \}_{k = 0}^\infty$ for the open-loop system \eqref{sys_x_chi_CL:x}-\eqref{sys_x_chi_CL:chi} and obtained a solution $\smat{x(\cdot)\\ \chi(\cdot)}$;
Lemma~\ref{lemma:io_(x,chi)_vs_xi} allows us to associate, with this solution, a solution $\xi(\cdot)$ to the open-loop system \eqref{sys_xi_CL:xi} with some initial condition $\xi(\ell) = \hat{\xi}$ and input $\{ v(k) \}_{k = \ell}^\infty = \{ \mathbf{K} \chi(k) \}_{k = \ell}^\infty$ such that $\xi(k) = \chi(k)$ for all $k \ge \ell$.
As a consequence, $\{ v(k) \}_{k = \ell}^\infty = \{ \mathbf{K} \chi(k) \}_{k = \ell}^\infty = \{ \mathbf{K} \xi(k) \}_{k = \ell}^\infty$.
Hence, for some initial condition $\xi(\ell) = \hat{\xi}$, $\xi(\cdot)$ is also a solution to the closed-loop system \eqref{sys_xi_CL}.
Thanks to the assumption of $\mathbf{A}_\ell+\mathbf{B}_\ell\mathbf{K}$ being Schur, $\xi(\cdot)$ converges asymptotically to $0$.
By~\eqref{GAS_cl_aux_implies_GAS_cl:xik=chik}, if $k\mapsto\xi(k)$ converges asymptotically to $0$, so do $k \mapsto \smat{y(k-\ell)\\ \vdotsS\\ y(k-1)}$, $k \mapsto \smat{u(k-\ell)\\ \vdotsS\\ u(k-1)}$ and $k \mapsto \chi(k)$.
By $(A,C)$ observable and \eqref{sol_sys_x:s_o_i} in Claim~\ref{claim:sol_sys_x}, for each $k \ge \ell$,
\begin{align*}
x(k) =
A^\ell
\mathcal{O}_\ell^{\tu{L}}
\smat{
y(k-\ell)\\
\vdotsS \\
y(k-1)
}
+
(\mathcal{R}_\ell -A^\ell \mathcal{O}_\ell^{\tu{L}} \mathcal{T}_\ell )
\smat{
u(k-\ell)\\
\vdotsS \\
u(k-1)
}
\end{align*} 
and so also $k \mapsto x(k)$ converges asymptotically to $0$.
In summary, we have shown that for each $\hat{x}$, $\hat{\chi}$ the solution $\smat{x(\cdot)\\ \chi(\cdot)}$ to~\eqref{sys_x_chi_CL} with initial condition $\smat{x(0)\\ \chi(0)} = \smat{\hat{x}\\ \hat{\chi} }$ converges asymptotically to $0$.
For linear time-invariant systems, attractivity implies global attractivity and \cite[Prop.~4]{lewis2017remarks} stability.
Then, we can conclude that $(x,\chi)=0$ is globally asymptotically stable for~\eqref{sys_x_chi_CL}.

\section{Proof of Lemma~\ref{lemma:implicationSNRasmpt}}
\label{app:proof_lemma_implicationSNRasmpt}

Let $\{ u(k) \}_{k=0}^{T-1}$ be the data-collection input sequence applied to~\eqref{sys_x_star}, which generates the state and output sequences $\{ x(k) \}_{k=0}^{T-\ells}$ and $\{ y(k) \}_{k=0}^{T-1}$.
Let
\begin{align*}
W_0 := 
\left[
\begin{array}{c}
\begin{smallmatrix}
~x(0)~ & ~x(1) & \dots & x(T-\ells)\\
\end{smallmatrix}\\
\hline
\begin{smallmatrix}
u(0) & u(1) & \dots & u(T-\ells)  \rule{0pt}{8pt}\\
\vdotsS & \vdotsS &  & \vdotsS\\
u(\ells-1) & u(\ells) & \dots & u(T-1)
\end{smallmatrix}
\end{array}
\right].
\end{align*}
By the definition of solution to the linear system~\eqref{sys_x_star} and the definitions of $\mathcal{O}_{\ells}$ and $\mathcal{T}_{\ells}$, data satisfy
\begin{align*}
\smat{
y(0)\\
\vdotsS \\
y(\ells-1)
}
& =
\mathcal{O}_\ells x(0)
+
\mathcal{T}_\ells 
\smat{
u(0)\\
\vdotsS \\
u(\ells-1)
}, \cdots \\
\smat{
y(T-\ells)\\
\vdotsS \\
y(T-1)
}
& =
\mathcal{O}_\ells x(T-\ells)
+
\mathcal{T}_\ells 
\smat{
u(T-\ells)\\
\vdotsS \\
u(T-1)
}, \\
\smat{
u(0)\\
\vdotsS \\
u(\ells-1)
} & = \smat{
u(0)\\
\vdotsS \\
u(\ells-1)
}, \dots, 
\smat{
u(T-\ells)\\
\vdotsS \\
u(T-1)
} = \smat{
u(T-\ells)\\
\vdotsS \\
u(T-1)
}.
\end{align*}
These relations can be assembled as
\begin{align}
\label{S0=[Ol,Tl;0,I]...}
S_0 & = \bmat{\mathcal{O}_\ells & \mathcal{T}_\ells \\ 0 & I_{m\ells}} W_0 \in \real^{(p \ells + m \ells) \times( T-\ells+1)}.
\end{align}
Based on this relation, the proof can be carried out in three steps.\newline
\textit{Step 1: $\Psi_0 \Psi_0^{\top} \succ \Theta_{22} \implies \rank S_0 = (p+m)\ells$.} 
Suppose by contradiction that $\Psi_0 \Psi_0^{\top} \succ \Theta_{22}$ but $S_0$ has not full row rank.
Then, there exists a nonzero vector $v\in \real^{p \ells + m \ells}$ such that $v^\top S_0=0$.
We have
\begin{align}
\label{eq:asspimply}
& v^\top \Psi_0 \Psi_0^{\top} v \overset{\eqref{Psi0=S0+N0}}{=} 
v^\top (S_0 + N_0) (S_0 + N_0)^{\top} v \notag\\
& = v^\top N_0 N_0^{\top} v > v^\top \Theta_{22} v.
\end{align} 
by $\Psi_0 \Psi_0^{\top} \succ \Theta_{22}$.
On the other hand, the definitions of $\Delta_{10}$ in~\eqref{Delta10} and $N_0$ in~\eqref{N0}, and $\Delta_{10} \in \mathcal{D}$ yield
\begin{align*}
\Delta_{10} \Delta_{10}^\top =
\smat{
d^y(\ells)\,\cdots\, d^y(T) \\
\hline
N_0
}
\smat{
d^y(\ells)\,\cdots\, d^y(T) \\
\hline
N_0
}^\top
\! \preceq
\smat{
\Theta_{11} & \Theta_{12}\\
\Theta_{12}^\top & \Theta_{22}
}\!\!,
\end{align*}
which implies $N_0 N_0^\top \preceq \Theta_{22}$.
This is contradicted by~\eqref{eq:asspimply}, so $S_0$ has full row rank.
\newline
\textit{Step 2: $\rank S_0 = (p+m)\ells \implies \rank \mathcal{O}_{\ells} = p \ells$.}
From~\eqref{S0=[Ol,Tl;0,I]...}, we have that $\smat{\mathcal{O}_\ells & \mathcal{T}_\ells \\ 0 & I_{m\ells}}$ must have full row rank, i.e., its rank is $(p+m)\ells$.
We have
\begin{align*}
& (p+m) \ells = \rank \smat{\mathcal{O}_\ells & \mathcal{T}_\ells \\ 0 & I_{m\ells}} \\
& = \rank \left( \smat{I_{p\ells} & \mathcal{T}_\ells \\ 0 & I_{m\ells}} \smat{\mathcal{O}_\ells & 0_{p\ells \times m\ells} \\ 0_{m\ells \times n} & I_{m\ells}} \right)  \\
& = \rank \mathcal{O}_\ells + m\ells.
\end{align*}
Hence, $p\ells =\rank \mathcal{O}_\ells$.
\newline
\textit{Step 3: $\rank \mathcal{O}_{\ells} = p \ells \implies p\ells=n$.} 
By observability of $(\As,\Cs)$, $\rank \mathcal{O}_\ells=n$ and $p\ells = n$.

\section{Proof of Lemma~\ref{lemma:SNR}}
\label{app:proof_lemma_SNR}

Since $S_0 = \Psi_0 + (-N_0)$ by~\eqref{Psi0=S0+N0}, \cite[(7.3.13)]{MatrixAnalysis2013} yields
\begin{align}
\sigma_{\min} (\Psi_0) & \geq \sigma_{\min}(S_0) - \sigma_{\max}(N_0) \notag \\
& \overset{\eqref{ratio_sing_values}}{>} 2 \sqrt{\sigma_{\max}(\Theta_{22})} - \sigma_{\max}(N_0). \label{part_res_SNR}
\end{align}
$\Delta_{10} \overset{\eqref{Delta10}}{=} \smat{
d^y(\ells)\,\cdots\, d^y(T) \\
\hline
N_0
} \in \mathcal{D}$ implies $N_0 N_0^\top \preceq \Theta_{22}$; 
since $\Theta_{22} = \Theta_{22}^\top$ and $N_0 N_0^\top \preceq \Theta_{22}$, \cite[Cor.~7.7.4(c)]{MatrixAnalysis2013} yields $\lambda_{\max}(N_0 N_0^\top) \le \lambda_{\max}(\Theta_{22})$ or, equivalently, $\sigma_{\max}(N_0) \le \sqrt{\sigma_{\max}(\Theta_{22})}$. 
By using this condition in~\eqref{part_res_SNR}, we obtain $\sigma_{\min} (\Psi_0) > \sqrt{\sigma_{\max}(\Theta_{22})}$ or, equivalently, $\sigma_{\min}(\Psi_0 \Psi_0^\top) > \sigma_{\max} (\Theta_{22})$.
Hence,
\begin{align*}
\Psi_0 \Psi_0^\top \succeq \sigma_{\min}(\Psi_0 \Psi_0^\top) I \succ \sigma_{\max}(\Theta_{22}) I \succeq \Theta_{22}.
\end{align*}
Since $\Psi_0 \Psi_0^\top \succ \Theta_{22}$, Assumption~\ref{ass:Psi0} holds.

\section{Proof of Lemma~\ref{lemma:assumptions_implication+obs_idx}}
\label{app:proof_lemma_assumptions_implication+obs_idx}

From the proof of Lemma~\ref{lemma:implicationSNRasmpt}, we note that
\begin{align*}
& \Psi_0 \Psi_0^\top\! \succ \! \Theta_{22} \!\!\implies\!\! \rank S_0 \!= \!(p+m)\ells  \!\!\implies\!\! \rank \mathcal{O}_{\ells} \!\!=\! p \ells
\end{align*}
hold independently of the observability of $(\As, \Cs)$ and the fact that $\ells$ is the observability index of $(\As, \Cs)$.
Hence, from analogous arguments to those in the proof of Lemma~\ref{lemma:implicationSNRasmpt}, $\Psi_0^{\aug} {\Psi_0^{\aug}}^\top \succ \Theta_{22}^{\aug}$ implies that $\rank \mathcal{O}_{\ells}^{\aug} = p \ells$.
This equality implies that $\mathcal{O}_{\ells}^{\aug}$ in~\eqref{eq:obserNonsinguAug} is nonsingular, since $\mathcal{O}_{\ells}^{\aug}$  is square by construction. 
Since $\mathcal{O}_{\ells}^{\aug}$ has full rank, the pair $(A_{\aug}, C_{\aug})$ is observable. 
Moreover, again since $\mathcal{O}_{\ells}^{\aug}$ is square and nonsingular, the observability index of $(A_{\aug}, C_{\aug})$ is, by its definition, equal to $\ells$.

\bibliographystyle{ieeetr}
\bibliography{mybib}

\begin{thebibliography}{10}

\bibitem{CAMPI2002Virtual}
M.~C. Campi, A.~Lecchini, and S.~M. Savaresi, ``Virtual reference feedback tuning: a direct method for the design of feedback controllers,'' {\em Automatica}, vol.~38, no.~8, pp.~1337--1346, 2002.

\bibitem{TANASKOVIC2017}
M.~Tanaskovic, L.~Fagiano, C.~Novara, and M.~Morari, ``Data-driven control of nonlinear systems: An on-line direct approach,'' {\em Automatica}, vol.~75, pp.~1--10, 2017.

\bibitem{formulas2020}
C.~De~Persis and P.~Tesi, ``Formulas for data-driven control: Stabilization, optimality, and robustness,'' {\em IEEE Transactions on Automatic Control}, vol.~65, no.~3, pp.~909--924, 2020.

\bibitem{Florian2019DeePC}
J.~Coulson, J.~Lygeros, and F.~Dörfler, ``Data-enabled predictive control: In the shallows of the {DeePC},'' in {\em European Control Conference (ECC)}, pp.~307--312, 2019.

\bibitem{GrahamAdaptive2009}
G.~C. Goodwin and K.~S. Sin, {\em Adaptive Filtering Prediction and Control}.
\newblock Courier Corporation, 2009.

\bibitem{berberich2021combining}
J.~Berberich, C.~W. Scherer, and F.~Allgöwer, ``Combining prior knowledge and data for robust controller design,'' {\em IEEE Transactions on Automatic Control}, vol.~68, no.~8, pp.~4618--4633, 2023.

\bibitem{Jong2023output}
T.~O. de~Jong, V.~Breschi, M.~Schoukens, and S.~Formentin, ``Data-driven model-reference control with closed-loop stability: The output-feedback case,'' {\em IEEE Control Systems Letters}, vol.~7, pp.~2431--2436, 2023.

\bibitem{makdah2023output}
A.~A.~A. Makdah and F.~Pasqualetti, ``Model-based and data-based dynamic output feedback for externally positive systems,'' {\em arXiv preprint arXiv:2305.02472}, 2023.

\bibitem{Dai2023IO}
X.~Dai, C.~{De Persis}, and N.~Monshizadeh, ``Data-driven optimal output feedback control of linear systems from input-output data,'' {\em IFAC-PapersOnLine}, vol.~56, no.~2, pp.~1376--1381, 2023.
\newblock 22nd IFAC World Congress.

\bibitem{William2023output}
W.~D’Amico, A.~Bisoffi, and M.~Farina, ``Data-based control design for output-error linear discrete-time systems with probabilistic stability guarantees,'' {\em IEEE Control Systems Letters}, vol.~7, pp.~2035--2040, 2023.

\bibitem{Guanru2023output}
G.~Pan, R.~Ou, and T.~Faulwasser, ``On a stochastic fundamental lemma and its use for data-driven optimal control,'' {\em IEEE Transactions on Automatic Control}, vol.~68, no.~10, pp.~5922--5937, 2023.

\bibitem{jo2022IOpredictive}
N.~H. Jo and H.~Shim, ``Data-driven output-feedback predictive control: Unknown plant's order and measurement noise,'' {\em arXiv preprint arXiv:2201.03136}, 2022.

\bibitem{LinearSystems2005}
P.~J. Antsaklis and A.~N. Michel, {\em Linear Systems}.
\newblock Springer Science \& Business Media, 2005.

\bibitem{mohammad2023MIMO}
M.~Alsalti, V.~G. Lopez, and M.~A. Müller, ``Notes on data-driven output-feedback control of linear {MIMO} systems,'' {\em arXiv preprint arXiv:2311.17484}, 2023.

\bibitem{Tomonori2022}
T.~Sadamoto, ``On equivalence of data informativity for identification and data-driven control of partially observable systems,'' {\em IEEE Transactions on Automatic Control}, vol.~68, no.~7, pp.~4289--4296, 2023.

\bibitem{Henk2023IO}
H.~van Waarde, J.~Eising, M.~Camlibel, and H.~Trentelman, ``A behavioral approach to data-driven control with noisy input-output data,'' {\em IEEE Transactions on Automatic Control}, pp.~1--14, 2023.

\bibitem{FrankDissipativity2022}
A.~Koch, J.~Berberich, and F.~Allgöwer, ``Provably robust verification of dissipativity properties from data,'' {\em IEEE Transactions on Automatic Control}, vol.~67, no.~8, pp.~4248--4255, 2022.

\bibitem{Miller2023ErrorVariable}
J.~Miller, T.~Dai, and M.~Sznaier, ``Data-driven stabilizing and robust control of discrete-time linear systems with error in variables,'' {\em arXiv preprint arXiv:2210.13430}, 2023.

\bibitem{petersen1987stabilization}
I.~R. Petersen, ``A stabilization algorithm for a class of uncertain linear systems,'' {\em Systems \& Control Letters}, vol.~8, no.~4, pp.~351--357, 1987.

\bibitem{stateinputerrors}
A.~Bisoffi, L.~Li, C.~De~Persis, and N.~Monshizadeh, ``Controller synthesis for input-state data with measurement errors.'' Published on arXiv in February, 2024.

\bibitem{miller2023superstabilizing}
J.~Miller, T.~Dai, and M.~Sznaier, ``Superstabilizing control of discrete-time {ARX} models under error in variables,'' {\em IFAC-PapersOnLine}, 2023.

\bibitem{AndreaPetersen2022}
A.~Bisoffi, C.~{De Persis}, and P.~Tesi, ``Data-driven control via {P}etersen's lemma,'' {\em Automatica}, vol.~145, p.~110537, 2022.

\bibitem{kailath1980linear}
T.~Kailath, {\em Linear Systems}.
\newblock Prentice-Hall, 1980.

\bibitem{Green1995}
M.~Green and D.~J. Limebeer, {\em Linear Robust Control}.
\newblock Courier Corporation, 2012.

\bibitem{cvx}
M.~Grant and S.~Boyd, {\em {CVX}: Matlab Software for Disciplined Convex Programming, version 2.1}.
\newblock \url{http://cvxr.com/cvx}, 2014.

\bibitem{gb08}
M.~Grant and S.~Boyd, ``Graph implementations for nonsmooth convex programs,'' in {\em Recent Advances in Learning and Control}, pp.~95--110, Springer, 2008.

\bibitem{lewis2017remarks}
A.~D. Lewis, ``Remarks on stability of time-varying linear systems.,'' {\em IEEE Transactions on Automatic Control}, vol.~62, no.~11, pp.~6039--6043, 2017.

\bibitem{MatrixAnalysis2013}
R.~A. Horn and C.~R. Johnson, {\em Matrix Analysis}.
\newblock Cambridge University Press, 2013.

\end{thebibliography}

\end{document}